\documentclass[a4paper,11pt]{article}

\usepackage{fullpage}
\usepackage{times}
\usepackage{soul}
\usepackage{url}
\usepackage[utf8]{inputenc}
\usepackage[small]{caption}
\usepackage{graphicx}
\usepackage{xcolor}
\usepackage{amsmath}
\usepackage{booktabs}
\usepackage{algorithm}
\usepackage{dsfont}
\usepackage{enumerate}
\usepackage{multirow}

\usepackage{amsthm}
\usepackage{amssymb}
\usepackage{algpseudocode}
\usepackage{amsfonts}
\usepackage{pifont}
\usepackage{cleveref}

\usepackage{mathrsfs}
\usepackage{enumitem}
\usepackage{nicefrac}

\usepackage{thm-restate}

\newcommand{\yes}{{\color{teal}\checkmark}}
\newcommand{\no}{{\color{red} \text{\sffamily X}}}

\usepackage{bbm}

\newtheorem{theorem}{Theorem}[section]

\newtheorem{proposition}[theorem]{Proposition}
\newtheorem{lemma}[theorem]{Lemma}

\theoremstyle{definition}
\newtheorem{definition}[theorem]{Definition}
\newtheorem{example}[theorem]{Example}
\newtheorem{remark}[theorem]{Remark}

\newcommand*{\innerproofname}{Proof}
\newenvironment{innerproof}[1][\innerproofname]{\begin{proof}[#1]}{\end{proof}}

\newcommand{\util}{\textsc{Util}}
\newcommand{\egal}{\textsc{Egal}}
\newcommand{\avg}{\textsc{Avg}}
\newcommand{\maxrule}{\textsc{Max}}
\newcommand{\minrule}{\textsc{Min}}
\newcommand{\med}{\textsc{Med}}
\newcommand{\geo}{\textsc{Geo}}
\newcommand{\im}{\textsc{IM}}
\newcommand{\ladder}{\textsc{Ladder}}
\newcommand{\pu}{\textsc{PU}}

\renewcommand{\epsilon}{\varepsilon}

\usepackage{natbib}

\allowdisplaybreaks

\usepackage{authblk}

\title{\bf Settling the Score: Portioning with Cardinal Preferences}

\author[1]{Edith Elkind}
\author[2]{Matthias Greger}
\author[3]{Patrick Lederer}
\author[4]{Warut Suksompong}
\author[5]{Nicholas Teh}

\affil[1]{Northwestern University, USA}
\affil[2]{LAMSADE, Université Paris Dauphine - PSL, France}
\affil[3]{University of New South Wales, Australia}
\affil[4]{National University of Singapore, Singapore}
\affil[5]{University of Oxford, UK}

\newcommand{\mathI}{\mathcal{I}}
\date{\vspace{-10mm}}

\begin{document}

\maketitle

\begin{abstract}
We study a portioning setting in which a public resource such as time or money is to be divided among a given set of candidates,
and each agent proposes a division of the resource. We consider two families of aggregation rules for this setting---those based on coordinate-wise aggregation and those that optimize some notion of welfare---as well as the recently proposed independent markets rule. We provide a detailed analysis of these rules from
an axiomatic perspective, both for classic axioms, such as strategyproofness and Pareto optimality, and for novel axioms, some of which aim to capture proportionality in this setting. 
Our results indicate that a simple rule that computes the average of the proposals satisfies many of our axioms and fares better than all other considered rules in terms of fairness properties. 
We complement these results by presenting two characterizations of the average rule.
\end{abstract}

\section{Introduction}
A town council has just received its annual funding from the government, and it needs to determine how to split the budget among constructing new facilities, maintaining clean streets, and ensuring public safety.
The mayor is in favor of making decisions democratically, so she asks each resident of the town to propose a division of the budget.
After collecting the proposals, how should the council aggregate them into an actual allocation?

In the problem of \emph{portioning}, the aim is to divide a homogeneous resource among a given set of candidates.
Besides dividing money, another important application of portioning is the division of time---for example, a conference needs to distribute its time among research talks, panels, and social gatherings.
Several prior works on portioning assumed that each agent submits her preferences in the form of either an approval ballot \citep{bogomolnaia2005collective,duddy2015fairsharing,aziz2020fairmixing} or an ordinal ranking~\citep{airiau2019portioning}.
However, in many portioning scenarios, these preference formats are not expressive enough to fully describe agents' intentions.
For instance, if a citizen wants the budget to be used both for constructing new facilities and for cleaning the streets, but with twice as much money spent on the former than the latter, her preference cannot be captured by a ranking or an approval set.
Likewise, a conference attendee who would like 75\% of the time to be spent on research talks, 15\% on panels, and 10\% on social gatherings ranks these activities in the same way as another attendee who prefers a 40\%--35\%--25\% split, but the actual preferences of these two attendees are quite different.

In an important recent work,~\cite{freeman2021truthfulbudget} studied portioning with cardinal preferences, where every agent is asked to propose a division of the resource.
This input format allows each agent to specify what she views as the ideal portioning outcome, and is therefore much more descriptive than the two formats discussed earlier.
Assuming that an agent's disutility is given by the $\ell_1$ distance between her ideal distribution and the actual outcome,\footnote{\citet{freeman2021truthfulbudget} noted that $\ell_1$ preferences arise naturally when agents are endowed with separable uniform utilities over candidates together with a funding cap.} Freeman et al.~observed that even though the rule that maximizes the utilitarian social welfare is known to be strategyproof (for a specific tie-breaking convention) \citep{lindner2008allocating,goel2019knapsack}, it tends to put too much weight on majority preferences.
In light of this observation, they introduced the \emph{independent markets} (\im) rule, which is strategyproof and, in some sense, more proportional.
Their work inspired a number of follow-up papers in this fundamental social choice setting, mostly focusing on strategyproofness \citep{caragiannis2024truthful,deberg2024truthful,freeman2024project,brandt2024optimal}.\footnote{For an overview of this and other related lines of work, we refer to the survey by \citet{suksompong2026voting}.}
However, while strategyproofness is an important consideration, there may be scenarios where other features of aggregation rules are just as---if not more---desirable. 
Thus, to help decision-makers identify suitable aggregation rules for their applications, it would be useful to \emph{(i)} build catalogues of axioms and popular aggregation rules for the portioning setting, \emph{(ii)} determine which of these axioms are satisfied by each of the aggregation rules, and \emph{(iii)} characterize some of the most important rules in terms of these axioms. 

\subsection{Overview of Contributions}

We consider a diverse set of axioms for portioning with cardinal preferences. Besides classic axioms such as strategyproofness and Pareto optimality, we put forward new axioms including score-unanimity and score-representation (see Section~\ref{sec:preliminaries} for definitions).
Several of our axioms do not depend on the underlying utility functions; for those where the utility functions matter, following most prior works in this domain, we assume $\ell_1$ utilities.
We then conduct a systematic study of aggregation rules with respect
to these axioms. 
We focus on two families of portioning rules---those that are based on coordinate-wise aggregation and those that optimize some notion of welfare---as well as the recently proposed IM rule \citep{freeman2021truthfulbudget}.
We also include observations regarding relationships between the axioms.
\Cref{tab:summary} summarizes our results. 

\begin{table}[t] 
\begin{center}
\begin{tabular}{ |c||c|c|c|c|c||c|c||c| } 
 \hline
 \multicolumn{1}{|c||}{} & 
 \multicolumn{5}{|c||}{ Coordinate-wise } &
 \multicolumn{2}{|c||}{Welfare-based} &
 \multicolumn{1}{|c|}{Other}
 \\
 \hline
 $F$ & \avg{} & \maxrule{} & \minrule{} & \med{} & \geo{} & \util{} & \egal{} & \im{} \\
 \hline
Pareto Optimality & \no$^{\ddagger}$  & \no$^\dagger$ & \no$^\dagger$ & \no$^\ddagger$ & \no$^\dagger$ & \yes & \yes &  \no$^\dagger$ \\
  \hline
Range-respect & \yes & \no$^\dagger$ & \no$^\dagger$ & \no$^\ddagger$  & \no$^\dagger$ & \yes & \yes & \no$^\dagger$ \\
 \hline
 Score-unanimity & \yes & \no$^\dagger$ & \no$^\dagger$ &  \no$^\ddagger$ & \no$^\dagger$ & \yes & \yes & \no$^\dagger$ \\
 \hline
  Score-representation & \yes & \no & \no & \no$^*$ & \no & \no & \no$^*$ &  \no$^\dagger$
 \\
  \hline
   Single-minded Proportionality & \yes & \no$^*$ & \no & \no$^*$ &   \no & \no$^*$ &  \no$^*$ & \yes
 \\
  \hline
  Independence & \yes & \no$^\dagger$ & \no$^\dagger$ & \no$^\ddagger$ & \no$^\dagger$ & \no$^\dagger$ & \no$^\ddagger$ & \no$^\dagger$ 
  \\
   \hline
 Score-monotonicity & \yes & \yes & \yes & \yes & \yes & \yes & \no$^\ddagger$ & \yes \\
 \hline
  Reinforcement & \yes & \yes & \yes & \no & \yes & \yes & \yes & \yes \\
 \hline
 Strategyproofness & \no & \no & \no & \no & \no & \yes & \no & \yes \\
  \hline
 Participation & \yes & \yes & \yes & \no & \yes & \yes & \yes &  \yes \\
 \hline
\end{tabular}
\caption[Summary]{Summary of our results.
The asterisk symbol ($^*$) means that the axiom is satisfied for $n=2$, but may fail when $n \geq 3$ (even if $m = 2$).
The dagger symbol ($\dagger$) indicates that the axiom is satisfied for $m=2$, but may fail when $m \geq 3$ (even if $n = 2$).
The double dagger symbol ($\ddagger$) indicates that the axiom is satisfied when $\min(n,m) = 2$.   
The results on single-minded proportionality, score-monotonicity, reinforcement, strategyproofness, and participation for \util{} and \im{} were obtained by \citet{freeman2021truthfulbudget}.
All results on \geo{} and on independence, as well as reinforcement and participation for \med{} and score-monotonicity for \egal{}, are new compared to the conference version of this paper \citep{elkind2023portioning}.
In \Cref{app:moving-phantoms}, we consider two moving phantoms rules that have been proposed more recently---the \emph{piecewise uniform} rule \citep{caragiannis2024truthful} and the \emph{ladder} rule \citep{freeman2024project}---and show that their axiomatic behavior is similar to that exhibited by \im{}.
}
\label{tab:summary}
\end{center}
\end{table}

Our findings offer several insights on portioning rules.
As shown in Table 1, the most promising rules with respect to the axioms that we study are the average rule (\avg), which simply returns the average of all the proposals, and the utilitarian welfare-maximizing rule (\textsc{Util}), with the trade-off being that \avg{} fails strategyproofness and Pareto optimality whereas \textsc{Util} fails fairness and consistency notions such as single-minded proportionality,\footnote{This property was simply called ``proportionality'' by \citet{freeman2021truthfulbudget}. 
However, the property only applies to instances in which all agents are ``single-minded'', thereby making it rather weak compared to proportionality notions in other settings (e.g., fair division \citep{procaccia2016cake}). Hence, in this paper we call this property ``single-minded proportionality''.} score-representation, and independence. 
While \im{} satisfies both strategyproofness and single-minded proportionality, it fails other intuitive properties such as score-unanimity and score-representation; these failures can lead to highly counterintuitive outcomes.
The axiomatic properties of various rules can be used to inform decision-making in a wide range of settings.
For instance, consider again the scenario where a conference organizer needs to divide time among different activities at a conference. 
In this case, it is likely difficult for an attendee to accurately predict what other attendees’ preferences are,
making strategyproofness arguably less relevant as a consideration. 
On the other hand, strategyproofness could be more important in smaller-scale settings where agents know each other well, e.g., portioning within a family or a small organization.  
Moreover, intuitive properties such as score-unanimity and range-respect may be essential in settings where votes are revealed: for example, if all agents allocate $80\%$ of the budget to a certain activity, but the rule allocates $60\%$ to it, this may well lead to dissatisfaction among agents regarding the use of that rule.

In addition to fulfilling several axioms, the average rule is intuitive and easy to explain to laypeople.
We further strengthen the case for using this rule by providing two characterizations of it.
Specifically, we show that the average rule is the only aggregation rule satisfying score-unanimity (i.e., if all agents allocate a $\gamma$ fraction of the resource to a candidate, then the rule also assigns a $\gamma$ fraction of the resource to it), independence (i.e., the fraction allocated to a candidate only depends on the fractions that the agents allocate to this candidate), and a mild fairness condition called anonymity, when there are at least three candidates.
We also prove that, within the class of coordinate-wise rules, the average rule is the unique rule that satisfies score-unanimity, anonymity, and continuity whenever the number of candidates is at least four.

\subsection{Further Related Work} 
\label{sec:relatedwork}

Portioning can be viewed as a variant of \emph{participatory budgeting}, a framework that allows citizens to democratically decide how the public budget should be spent.
Participatory budgeting has been used in over 7,000 cities around the world~\citep{PB2025} and received much recent interest in computational social choice---see, for example, the surveys by \citet{aziz-shah} and \citet{RSM25a}.
Nevertheless, most of the participatory budgeting literature focuses on the discrete setting, where each project is either implemented in full or not implemented at all, and projects may have varying costs (see, however, the recent work of \citet{goyal2023pb}).
This makes the nature of the problem quite different from that of portioning.

As mentioned earlier, \citet{freeman2021truthfulbudget} investigated portioning with cardinal preferences and introduced \im{}, which is strategyproof and single-minded proportional under $\ell_1$ utilities.
In fact, \im{} belongs to a class of \emph{moving phantoms} mechanisms, all of which are strategyproof.\footnote{However, not all strategyproof rules are moving phantoms mechanisms \citep{deberg2024truthful}.}
\citet{caragiannis2024truthful} followed up on their work by examining the deviation of moving phantoms mechanisms from the average rule according to the $\ell_1$ distance, while \citet{freeman2024project} explored a similar question using the $\ell_\infty$ distance.
\citet{brandt2024optimal} showed that no rule can simultaneously be strategyproof, single-minded proportional, and Pareto optimal under $\ell_1$ or $\ell_\infty$ utilities, but such a rule exists for an alternative utility model. 
\citet{schmidt2025discrete} studied a discrete version of portioning in which only integral amounts can be assigned to candidates.

Portioning also bears similarities to the domain of \emph{probabilistic social choice}, where the output is likewise a fractional allocation.
However, unlike in portioning, in probabilistic social choice 
the fraction allocated to each candidate is interpreted as the probability of eventually choosing this candidate as the unique winner.
As a consequence, it is usually desirable for a rule to minimize the use of randomness, with fairness and proportionality being relatively less important.
Moreover, much of the work in probabilistic social choice assumes that agents' preferences are given as ordinal rankings (see, e.g., the survey by \citet{brandt2017rolling}).
A notable exception is the work of \citet{intriligator1973probsc}, which postulates that each agent has an ideal distribution over the candidates, but does not consider utility functions of agents.
Intriligator gave a characterization of the average rule, which was later identified as incorrect and fixed by \citet{rice1977comment}.
In particular, Intriligator's (and Rice's) characterization relies on two axioms: \emph{(i)}~loser-unanimity, which requires that a candidate receiving a score of~$0$ from all agents should also receive~$0$ from the actual distribution, and \emph{(ii)} strict and equal sensitivity to individual preferences, which requires that if an agent has more ``power'' and allocates that power to a particular candidate, then the outcome should only change for this candidate, and this change should only depend on the absolute magnitude of the change in the agent's preference.
While loser-unanimity is weaker than score-unanimity, which we use for our characterizations, strict and equal sensitivity to individual preferences is an extremely strong condition---for example, it is violated even by a dictatorial or a constant aggregation rule.

Finally, another related topic is \emph{probabilistic opinion pooling}, where the aim is to aggregate probabilistic beliefs representing, for example, weather forecasts \citep{genest1986combining,clemen1989forecasts}.
The focus of probabilistic opinion pooling is mainly to preserve epistemic and stochastic properties, which again leads to different axioms being considered.

\section{Preliminaries} \label{sec:preliminaries}

We present the model of portioning with cardinal preferences, and introduce the rules and axioms that we will study.

\subsection{Model}

Let $[t] := \{1,\dots,t\}$ for any positive integer~$t$.
Assume that there is a set of $n \ge 2$ agents, $N = [n]$, who report their preferences as ideal distributions of a homogeneous resource among a set $C=\{c_1,\dots, c_m\}$ of $m \ge 2$ candidates. 
Specifically, letting $\Delta^m := \{\mathbf{x} \in \mathbb{R}^m_{\ge 0} \colon \sum_{j \in [m]}x_j=1\}$ denote the set of probability distributions over $C$, we assume that each agent $i\in N$ reports her preferences as a distribution $\mathbf{s}_i \in \Delta^m$. We typically refer to $\mathbf{s}_i$ as agent $i$'s \emph{score vector}, and write $\mathbf{s}_i=(s_{i,1}, \dots, s_{i,m})$ to specify this vector. 
An {\em instance} $\mathcal{I}$ of our problem is the collection of the preferences of all agents, i.e., $\mathcal{I} = (\mathbf{s}_1,\dots,\mathbf{s}_n)$. For each vector $\mathbf{x} = (x_1,\dots,x_m)$, agent $i$'s \emph{disutility} is defined as $d_i(\mathbf{x}) := \sum_{j \in [m]} |s_{i,j} - x_j|$, which is the $\ell_1$ distance between the agent's score vector $\mathbf{s}_i$ and~$\mathbf{x}$. 
Given an instance~$\mathcal I$, we aim to find a vector $\mathbf{x} \in \Delta^m$ that reflects the agents' collective preferences. To this end, we use \emph{aggregation rules}, which are defined as follows.

\begin{definition}[Aggregation rule]
    An \emph{aggregation rule} $F$ is a function $F : \Delta^{m \times n} \to \Delta^m$ that maps every instance $\mathcal{I} \in \Delta^{m \times n}$ to an outcome vector $\mathbf{x} \in \Delta^m$.
\end{definition}

We will frequently use the notation $F(\mathcal{I})_j$ to denote the probability that the aggregation rule $F$ assigns to candidate $c_j$ on instance $\mathcal{I}$.

\subsection{Aggregation Rules}\label{subsec:rules}

In this paper, we will focus on two natural classes of aggregation rules, namely, coordinate-wise rules and welfare-optimizing rules. In addition, we will also consider the independent markets rule
of~\cite{freeman2021truthfulbudget}.

\subsubsection{Coordinate-wise Aggregation Rules}

We first introduce the class of \emph{coordinate-wise} aggregation rules. The idea behind these rules is to aggregate the reported scores for each candidate individually and then normalize the aggregated scores so that they sum up to $1$.

\begin{definition} \label{defn:coord_rules}
    An aggregation rule $F$ is \emph{coordinate-wise}
    if, for each $n$ and each $j\in [m]$, there exist {\em coordinate-aggre\-ga\-tion functions} $f_j^n: (\mathbb{R}_{\ge 0})^n \to \mathbb{R}_{\ge 0}$
     such that $F(\mathI)_j= \frac{f_j^n(s_{1,j}, \dots, s_{n,j})}{\sum_{k \in [m]} f_{k}^n(s_{1,k}, \dots, s_{n,k})}$ for all instances $\mathcal{I}$.
\end{definition}

We extend Definition~\ref{defn:coord_rules} to allow for the possibility that, for certain instances, $f_j^n(s_{1,j},\dots, s_{n,j})=0$ for all $j\in [m]$:
in this case, we assign each candidate the same probability $\frac1{m}$. 
We remark that our negative results do not depend on this tie-breaking convention.
When $f_j^n$ is the same for all $j\in [m]$, we omit the subscript~$j$ and write $f^n$.
Furthermore, we may omit
the superscript~$n$ when it is clear from the context, and simply write $f_j$ or $f$.

We will focus on five natural coordinate-wise aggregation rules, where $f$ is the \emph{average}, \emph{maximum}, \emph{minimum}, 
\emph{median} (if the number of agents is even, we take the average of the two middle scores),\footnote{This is the most commonly used definition of the median in practice. However, even when $m = 2$, this definition does not lead to a strategyproof rule, whereas other definitions do (e.g., taking the smaller of the two middle scores) \citep{Moul80a}.} or \emph{geometric mean} function. 
For brevity, we refer to these rules as \avg, \maxrule{}, \minrule{}, \med{}, and \geo{}, respectively.
The advantage of these rules is that they are intuitive and easily computable.
The proposition below identifies two special cases where the normalization factor 
$\sum_{k \in [m]} f_{k}^n(s_{1,k}, \dots, s_{n,k})$ is equal to~$1$ and hence normalization is not required.
\begin{proposition}\label{prop:nonorm}
    Consider an instance $\mathI$, and let $\mathbf x$ and $\mathbf y$ be the outputs of \avg{} and \med{} on $\mathI$, respectively.
    \begin{itemize}
    \item[(a)]
    For each $j\in [m]$, it holds that $x_j=\frac{1}{n}\cdot\sum_{i\in N} s_{i, j}$.
    \item[(b)]
    If $m=2$, then for $j\in[2]$, it holds that
    $y_j = \text{\em med}(s_{1, j}, \dots, s_{n, j})$.
    \end{itemize}
    \end{proposition}
    \begin{proof}
        For part~(a), it suffices to note that
        $$
        \sum_{j \in [m]} \frac{1}{n}\left(\sum_{i\in N} s_{i, j}
        \right) = \frac{1}{n}\sum_{i\in N}\sum_{j\in [m]} s_{i, j} = \frac{1}{n}\sum_{i\in N} 1 = 1.
        $$
        For part~(b), assume without loss of generality that 
        $s_{1, 1}\ge s_{2, 1}\ge\dots\ge s_{n, 1}$. As $s_{i, 1}=1-s_{i, 2}$ for all $i\in N$, we have 
        $s_{1, 2}\le s_{2, 2}\le\dots\le s_{n, 2}$.
        Consequently, if $n$ is odd, we have 
        $\text{med}(s_{1, 1}, \dots, s_{n, 1}) = s_{\frac{n+1}{2}, 1}$ and 
        $\text{med}(s_{1, 2}, \dots, s_{n, 2}) = s_{\frac{n+1}{2}, 2}$, while
        if $n$ is even, we have 
        $\text{med}(s_{1, 1}, \dots, s_{n, 1}) = (s_{\frac{n}{2}, 1}+s_{\frac{n+2}{2}, 1})/2$ and 
        $\text{med}(s_{1, 2}, \dots, s_{n, 2}) = (s_{\frac{n}{2}, 2}+s_{\frac{n+2}{2}, 2})/2$.
        In both cases, since $s_{i, 1}+s_{i, 2}=1$ for all $i\in N$, it holds that
        $\text{med}(s_{1, 1}, \dots, s_{n, 1})+\text{med}(s_{1, 2}, \dots, s_{n, 2})=1$.
    \end{proof}

\subsubsection{Welfare-based Aggregation Rules}

For our second class, we consider rules that are based on welfare optimization. In particular, we focus on two popular welfare criteria:\footnote{There is a third popular welfare criterion called \emph{Nash welfare}, which is defined based on the {product} of utilities. However, this welfare notion is not well-defined in our setting, as we are considering disutilities. For example, it has been observed that there is no natural equivalent of Nash welfare in the fair allocation of chores \citep{freeman2020chores,ebadian2022chores}.} the \emph{utilitarian welfare} $-\sum_{i\in N}d_i(\mathbf{x})$ and the \emph{egalitarian welfare} $\min_{i\in N}(-d_i(\mathbf{x}))$. 
The minus sign ensures that these definitions indeed capture the welfare, as $d_i$ measures the disutility (rather than the utility) of agent $i$. 
The \emph{utilitarian rule} (\util{}) and the \emph{egalitarian rule} (\egal{}) then return an outcome that maximizes the utilitarian and egalitarian welfare, respectively. 
More formally, for each instance $\mathI$, \util{} chooses an outcome $\mathbf{x}_\util{}$ such that $\mathbf{x}_\util{}\in \arg\max_{\mathbf{x}} (-\sum_{i\in N}d_i(\mathbf{x}))$, while \egal{} returns 
$\mathbf{x}_\egal{}$ such that $\mathbf{x}_\egal{}\in \arg\max_{\mathbf{x}} \min_{i\in N} (-d_i(\mathbf{x}))$.

For both of these rules, tie-breaking is important. 
Following~\cite{freeman2021truthfulbudget}, we consider the variant of \util{} that breaks ties in favor of the maximum-entropy division. Specifically, we assume that \util{} outputs the utilitarian welfare-maximizing outcome $\mathbf{x}$ that minimizes the quantity $\sum_{j \in [m]} (x_j - \frac1m)^2$,
i.e., the $\ell_2$ distance to the uniform distribution $\mathbf{x}_u = (\frac1m, \dots, \frac1m)$.
This tie-breaking choice is neutral for candidates and ensures strategyproofness \citep{lindner2008allocating}.

For \egal{}, if there are multiple outcomes that maximize the egalitarian welfare, then we break ties in a ``leximin'' manner.
Formally, given two vectors $v, v'\in \mathbb R^t$ for some $t\in{\mathbb N}$, we write $v>_\mathit{lex} v'$ if there is an integer $k$ such that $v_{\ell}=v'_\ell$ for all $\ell\in [k-1]$ and $v_{k}>v_{k}'$. Given an outcome ${\mathbf x}$ for an $n$-agent instance $\mathI$, let $v({\mathbf x})$ be the list of all agents' disutilities from $\mathbf x$ (as non-negative numbers), sorted in non-increasing order. 
We require that, if \egal{} outputs $\mathbf x$ on $\mathI$, then for each outcome ${\mathbf x}'$ it holds that either 
$v({\mathbf x})=v({\mathbf x}')$ 
or $v({\mathbf x}')>_\mathit{lex} v({\mathbf x})$.
This type of leximin tie-breaking is standard when dealing with egalitarian welfare \citep[e.g.,][]{bogomolnaia2004random,kurokawa2018leximin}.
However, even after this tie-breaking process, there may still be multiple \egal{} outcomes. 
We show next that if $n=2$, \avg{} always returns an \egal{} outcome, so we assume that the \egal{} rule coincides with \avg{} in this case. By contrast, our results for $n\ge 3$ will not depend on this latter tie-breaking rule, and we allow \egal{} to break such ties in any consistent manner (i.e., if $\mathbf{x}$ and $\mathbf{x}'$ are both \egal{} outcomes for two instances $\mathI$ and $\mathI'$, then we assume that if \egal{} chooses $\mathbf{x}$ for $\mathI$, it does not choose $\mathbf{x}'$ for $\mathI'$).

\begin{proposition} \label{prop:welfare-egal_is_sum}
    When $n=2$, the output of \avg{} is an \emph{\egal{}} outcome.
\end{proposition}
\begin{proof}
    Let $\mathbf{x}$ be the output of \avg{} for the case $n=2$. Then $x_j = \frac{s_{1,j}+s_{2,j}}{2}$ for each $j \in [m]$, so
    $d_1(\mathbf{x}) = d_2(\mathbf{x})=\frac12\,\sum_{j\in [m]} |s_{1,j} - s_{2,j}|$. 
    On the other hand, for every $\mathbf{x}'$, we have
    $d_1(\mathbf{x}') + d_2(\mathbf{x}') \ge \sum_{j\in [m]} |s_{1,j} - s_{2,j}|$ and therefore $\max\{d_1(\mathbf{x}'), d_2(\mathbf{x}')\}\ge \frac12\,\sum_{j\in [m]} |s_{1,j} - s_{2,j}|$.  It follows that $\mathbf{x}$ is an \egal{} outcome.
\end{proof}

We also note that both \util{} and \egal{} (with the given tie-breaking conventions) can be computed in polynomial time. For \util{}, this follows from the results of \citet{freeman2021truthfulbudget}, and for \egal{}, we prove this claim in the appendix (see \Cref{thm:egal_compute}).

\subsubsection{Independent Markets Rule}

The last aggregation rule that we will study is the \emph{independent markets} (\im) rule of \citet{freeman2021truthfulbudget}, which belongs to the class of \emph{moving phantoms} rules. 
For each candidate, these rules take the median of the agents' reports and $n+1$ phantom values. 
When there are two candidates, one may fix these phantom values in advance---the resulting ``phantom median rules'' correspond to all strategyproof rules \citep{Moul80a}.
However, with more than two candidates, this approach may lead to an unnormalized score vector.
Because the final scores of all candidates must sum up to $1$, the phantom values therefore cannot be constant. 
Instead, \citet{freeman2021truthfulbudget} proposed using $n+1$ phantom functions $f_0,\dots,f_n:[0,1]\to [0,1]$ that are continuous, weakly increasing, and satisfy $f_k(0)=0$ and $f_{\sigma(k)}(1)\ge k/n$ for all $k\in\{0,\dots,n\}$ and some permutation $\sigma$ of $\{0,\dots,n\}$.\footnote{The original definition of \citet{freeman2021truthfulbudget} requires $f_k(1) = 1$ for all $k$. However, in their Proposition~3 (and the remarks thereafter), Freeman et al.~showed that this weaker condition is sufficient.}
Then, a moving phantoms rule determines a value $t^*$ such that $\sum_{j\in[m]} \text{med}(s_{1,j}, \dots, s_{n,j}, f_0(t^*), \dots, f_{n}(t^*))=1$, and returns the vector $\mathbf{x}$ given by $x_j=\text{med}(s_{1,j}, \dots, s_{n,j}, f_0(t^*), \dots, f_{n}(t^*))$ for each $j\in [m]$. 
Freeman et al.~showed that if multiple values of $t^*$ exist, then the returned vector is the same for all such~$t^*$.
The \im{} rule belongs to this class of rules,\footnote{In \Cref{app:moving-phantoms}, we study two additional rules in this class: the \emph{piecewise uniform} rule \citep{caragiannis2024truthful} and the \emph{ladder} rule \citep{freeman2024project}.} and is defined by setting $f_k^{\im}(t)=\min(kt,1)$ for all $k\in \{0,\dots, n\}$ and all $0\le t\le 1$. 

\medskip

Before proceeding further, let us present a simple example which demonstrates how the different aggregation rules work.

\begin{example}\label{ex:rules}
Consider an instance with $n = 2$ agents and $m = 3$ candidates.
The first agent has a preferred distribution $\mathbf{s}_1 = (\frac{4}{5}, \frac{1}{5}, 0)$ and the second agent has a preferred distribution $\mathbf{s}_2 = (\frac{4}{5}, 0, \frac{1}{5})$.

It is immediate that \avg{} and \med{} output $(\frac{4}{5}, \frac{1}{10}, \frac{1}{10})$, \minrule{} and \geo{} output $(1, 0, 0)$, and \maxrule{} outputs $(\frac{2}{3}, \frac{1}{6}, \frac{1}{6})$.
For \util{}, the distribution $(\frac{4}{5}, x, \frac{1}{5}-x)$ maximizes the utilitarian welfare for every $x\in [0, \frac{1}{5}]$; among these distributions, ${\mathbf x} = (\frac{4}{5}, \frac{1}{10}, \frac{1}{10})$ minimizes the $\ell_2$ distance to the uniform distribution $(\frac{1}{3}, \frac{1}{3}, \frac{1}{3})$, so \util{} returns $\mathbf x$.
Moreover, \egal{} returns $\mathbf x$ as well.
It remains to consider \im{}. 
We have $f_0^\im(t)=0$, $f_1^\im(t)=t$, and
$f_2^\im(t)=\min(2t, 1)$ for $0\le t\le 1$. 
Accordingly, we need to find 
the smallest $t^*$
such that $x_1(t^*)+x_2(t^*)+x_3(t^*)=1$, 
where 
\begin{align*}
x_1(t) &=\text{med}\left(\frac{4}{5}, \frac{4}{5}, 0, t, \min(2t, 1)\right),\\ 
x_2(t) &=\text{med}\left(\frac{1}{5}, 0, 0, t, \min(2t, 1)\right),\\ 
x_3(t) &=\text{med}\left(0, \frac{1}{5}, 0, t, \min(2t, 1)\right).
\end{align*}
As $x_2(t), x_3(t)\le \frac{1}{5}$ for all $t\in [0, 1]$, we must have $x_1(t)\ge \frac{3}{5}$.
The smallest value of $t$ that guarantees this is $t^*=\frac{3}{10}$, and the resulting distribution is $(\frac{3}{5}, \frac{1}{5}, \frac{1}{5})$.
\end{example}

\subsection{Axioms}\label{subsec:axioms}

We next introduce the axioms that we will use to evaluate our aggregation rules. 
We group the axioms into four rough categories: efficiency properties, fairness properties, consistency properties, and incentive properties. However, we note that the boundaries between these categories are fluid, particularly when considering weak axioms. 

\subsubsection{Efficiency Properties}

We start by introducing efficiency properties. 
Intuitively, efficiency requires that no outcome is preferred to the chosen outcome by all agents. 
Perhaps the most prominent efficiency property is Pareto optimality, which postulates that it should not be possible to make one agent better off without making another agent worse off. 

\begin{definition}[Pareto optimality]
    An outcome $\mathbf{x}$ is \emph{Pareto optimal} in an instance $\mathcal{I}$ if there is no other outcome $\mathbf{x'}$ such that $d_i(\mathbf{x'})\leq d_i(\mathbf{x})$ for all agents $i\in N$ and $d_i(\mathbf{x'})< d_i(\mathbf{x})$ for some agent $i\in N$. An aggregation rule $F$ is \emph{Pareto optimal} if $F(\mathcal{I})$ is Pareto optimal for every instance $\mathcal{I}$.
\end{definition}

As it will turn out, Pareto optimality is a rather restrictive property in our setting: among the rules we consider, it is only satisfied by \util{} and \egal{}. 
We therefore introduce two further axioms based on the agents' scores, which only exclude obvious and easily identifiable efficiency violations. 
The first axiom is \emph{range-respect}, previously studied by \citet{freeman2021truthfulbudget}. This axiom states that the score assigned to a candidate should always lie between the minimum and the maximum scores that agents assign to this candidate. 

\begin{definition}[Range-respect]
    An outcome $\mathbf x$ is \emph{range-respecting} in an instance $\mathcal{I}$ if $\min_{i\in N}s_{i,j}\le x_j \le \max_{i\in N}s_{i,j}$ for all $j\in [m]$.
    An aggregation rule $F$ is \emph{range-respecting} if $F(\mathcal{I})$ is range-respecting for every instance $\mathcal{I}$.
\end{definition}

Next, we introduce a new property which we call \emph{score-unanimity}. 
This property demands that if all agents report the same score for some candidate, then this candidate should receive exactly that score. 
Score-unanimity is inspired by a property called unanimity in single-winner voting, which states that if all agents agree on a favorite candidate, this candidate should be chosen. 
In single-winner voting, this property appears almost indispensable, and we believe that much of its appeal carries over to score-unanimity. 
Specifically, like in voting, it seems difficult to justify not assigning a score of $x$ to a candidate when all voters report an ideal score of~$x$ for the candidate. 
Moreover, as we show in \Cref{thm:efficiency_implications}, a violation of score-unanimity constitutes a straightforward violation of Pareto optimality, which provides further motivation for this axiom. 

\begin{definition}[Score-unanimity]
    An outcome $\mathbf x$ is
    \emph{score-unanimous} in an instance $\mathcal I$ if, for every $j\in [m]$
    such that there exists $\gamma\in [0, 1]$ satisfying $s_{i, j}=\gamma$ for all $i\in N$, it holds that $x_j=\gamma$.
     An aggregation rule $F$ is \emph{score-unanimous} if $F(\mathcal{I})$ is score-unanimous for every instance $\mathcal{I}$.
\end{definition}

We show that our three efficiency notions are logically related. 

\begin{proposition}\label{thm:efficiency_implications}
The following claims hold.
    \begin{enumerate}[label=(\arabic*),topsep=4pt,itemsep=0pt]
        \item Pareto optimality implies range-respect. 
        \item Range-respect implies Pareto optimality if and only if $m=2$ or $n=2$. 
        \item Range-respect implies score-unanimity. 
    \end{enumerate}
\end{proposition}

\begin{proof}
    We prove each of the claims separately.\medskip

    \noindent\textbf{Claim 1}: Suppose for contradiction that there is an outcome $\mathbf{x}$ that is Pareto optimal but not range-respecting for an instance $\mathcal{I}$. Without loss of generality, this means that there exists some $j \in [m]$ such that $x_j > \max_{i\in N} s_{i,j}$ (the case $x_j<\min_{i\in N} s_{i,j}$ allows for a symmetric argument). Since $\sum_{k\in [m]} x_k=\sum_{k\in [m]} s_{1,k}=1$, there exists an index $\ell\in [m]$ such that $x_\ell<s_{1,\ell}$. Define $\epsilon=\min(x_j-\max_{i\in N} s_{i,j},\,\,s_{1,\ell}-x_\ell)$,
    and note that $\epsilon>0$. 
    Consider the outcome $\mathbf{x'}$ given by $x_j'=x_j-\epsilon$, $x_\ell'=x_\ell+\epsilon$, and $x'_k=x_k$ for all $k\in [m]\setminus \{j,\ell\}$. Since $x_j-\epsilon\geq s_{i,j}$ for all $i\in N$, it holds that $|(x_j-\epsilon)-s_{i,j}|=|x_j-s_{i,j}|-\epsilon$. By combining this observation with the triangle inequality, for every agent $i\in N\setminus\{1\}$ we obtain
    \begin{align*}
        d_i(\mathbf{x'})&=|(x_j-\epsilon)-s_{i,j}|+|(x_\ell+\epsilon) -s_{i,\ell}|+\sum_{k\in [m]\setminus \{j,\ell\}} |x_k-s_{i,k}|\\
        &\leq |x_j-s_{i,j}|-\epsilon+|x_\ell -s_{i,\ell}|+\epsilon+\sum_{k\in [m]\setminus \{j,\ell\}}|x_k-s_{i,k}|
        =d_i(\mathbf{x}).
    \end{align*}
    Moreover, for agent $1$, it additionally holds that $s_{1,\ell}\geq x_\ell+\epsilon$, so we have $|(x_\ell+\epsilon)-s_{1,\ell}|=|x_\ell-s_{1,\ell}|-\epsilon$. 
    Thus, repeating the calculation above for $i=1$, we obtain $d_1(\mathbf{x'})=d_1(\mathbf{x})-2\epsilon<d_1(\mathbf{x})$. This means that $\mathbf{x}$ is not Pareto optimal, a contradiction. 
    It follows that Pareto optimality implies range-respect.\medskip

    \noindent\textbf{Claim 2}: For our second claim, we show that range-respect implies Pareto optimality if and only if $m=2$ or $n=2$. First, we consider the case $n=2$, and let $\mathbf{s}_1$ and $\mathbf{s}_2$ denote the preferences of the agents.  If $\mathbf{x}$ is range-respecting---i.e., 
    $\min(s_{1,j}, s_{2,j})\leq x_j\leq \max(s_{1,j}, s_{2,j})$ for all $j\in [m]$---it follows that $|s_{1,j}-x_j|+|s_{2,j}-x_j|=|s_{1,j}-s_{2,j}|$ for all $j\in [m]$ and hence
    $d_1(\mathbf{x})+d_2(\mathbf{x}) = \sum_{j \in [m]} |s_{1,j} - s_{2,j}|$. Now, suppose that some outcome ${\mathbf x}'$ Pareto dominates ${\mathbf x}$.
    Then ${\mathbf x}'$  would need to have strictly less total disutility than $\mathbf x$.
    However, by the triangle inequality we have $d_1(\mathbf{x}')+d_2(\mathbf{x}') \geq \sum_{j \in [m]} |s_{1,j} - s_{2,j}| = d_1(\mathbf{x})+d_2(\mathbf{x})$, a contradiction.

    Next, consider the case $m=2$, and let $\mathbf{x}$ be a range-respecting outcome for an instance $\mathcal{I}$. 
    Furthermore, let $\mathbf{x'}$ be a different outcome, and assume without loss of generality that $x_1 > x'_1$. 
    Since $m=2$, this means that $x_2 < x'_2$.
    Let $i$ be an agent such that $s_{i,1}\geq x_1$, which implies that $s_{i,2}\leq x_2$; such an agent exists since $\mathbf{x}$ is range-respecting. 
    Observe that
        $d_i(\mathbf{x}) = (s_{i,1} - x_1) + (x_2 - s_{i,2}) < (s_{i,1} - x'_1) + (x'_2 - s_{i,2}) = d_i(\mathbf{x}')$.
    Consequently, agent $i$ strictly prefers $\mathbf{x}$ to every outcome $\mathbf{x'}$ with $x_1>x_1'$. Since a similar argument applies if $x_1<x_1'$, we infer that $\mathbf{x}$ is Pareto optimal. 

    We now show that range-respect does not imply Pareto optimality if $m\geq 3$ and $n\geq 3$. We provide a counterexample for $m=3$; to extend it to larger $m$, it suffices to have all agents assign score $0$ to additional candidates. Consider the following instance $\mathcal I$ and the outcome $\mathbf{x} = (\frac{1}{6}, \frac{1}{3}, \frac{1}{2})$, which is range-respecting for $\mathcal{I}$. However, $\mathbf x$ is Pareto dominated
    by ${\mathbf x}'=(0, \frac{1}{2}, \frac{1}{2})$, which is strictly better for agent $2$ and no worse for the other agents.

    \begin{center}
        \begin{tabular}{ c | c c c c}
          $\mathcal{I}$ & $s_{i,1}$ & $s_{i,2}$ & $s_{i,3}$  \\ 
         \hline \hline
         $1$ & $0$ & $0$ & $1$ \\ 
         $2$ & $0$ & $\frac{1}{2}$ & $\frac{1}{2}$ \\ 
         $i\in \{3,\dots, n\}$ & $\frac{1}{2}$ & $\frac{1}{2}$ & $0$ \\ 
        \end{tabular}        
    \end{center}

    \medskip
    
    \noindent\textbf{Claim 3}: It remains to show that range-respect implies score-unanimity. 
    To this end, consider an instance $(\mathbf{s}_1, \dots, \mathbf{s}_n)$, an index $j\in[m]$, and a value $\gamma\in [0,1]$ such that $s_{i,j}=\gamma$ for all $i\in N$. 
    If an outcome $\mathbf{x}$ satisfies range-respect, it must hold that $\gamma=\min_{i\in N} s_{i,j}\leq x_j\leq \max_{i\in N} s_{i,j}=\gamma$, so $x_j=\gamma$ and score-unanimity is satisfied. 
\end{proof}

\subsubsection{Fairness Properties}

We now turn to fairness concepts, which intuitively demand that every group of agents with similar preferences is proportionally represented by the outcome. A rather mild axiom based on this idea has been formulated by \citet{freeman2021truthfulbudget}. Specifically, these authors call an agent \emph{single-minded} if she assigns score $1$ to some candidate. Their proportionality notion then requires that, if all agents are single-minded, each candidate should be allocated a probability proportional to the number of agents that assign score $1$ to it. To formalize this concept (as well as the more demanding concept of fairness to be introduced later), we define $\mathcal{N}(\mathcal{I},c_j,\gamma):=|\{i\in N\colon s_{i,j}\geq\gamma\}|$ as the number of agents who assign a score of at least $\gamma$ to candidate $c_j$ in the instance $\mathcal{I}$. 

\begin{definition}[Single-minded Proportionality]
\label{def:proportionality}
    An aggregation rule $F$ satisfies \emph{single-minded proportionality} if,
    for every instance $\mathcal{I}$ in which all agents are single-minded and for all $j\in [m]$, it holds that $F(\mathcal{I})_j=\frac{\mathcal{N}(\mathcal{I},c_j,1)}{n}$.
\end{definition}

However, there are many applications of portioning where agents are unlikely to be single-minded (such as dividing time among different activities at a conference), so we need an appropriate notion of proportionality for general preferences.
To this end, we formulate the axiom of \emph{score-representation}. The idea behind this notion is that if a $\frac{k}{n}$ fraction of the agents assign a score of at least $\gamma$ to a candidate, then this candidate should receive a probability of at least $\gamma\cdot\frac{k}{n}$.

\begin{definition}[Score-representation] \label{defn-score_representation}
    An aggregation rule $F$ satisfies \emph{score-representation} if, for all instances $\mathcal{I}$, all $j\in [m]$, and all $\gamma\in [0,1]$, it holds that $F(\mathcal{I})_j\geq \gamma\cdot \frac{\mathcal{N}(\mathcal{I}, c_j, \gamma)}{n}$.
\end{definition}

It follows directly from the definitions that score-representation is strictly stronger than single-minded proportionality. 

\subsubsection{Consistency Properties}

As the third type of axioms, we consider consistency properties, which aim to ensure that voting rules behave, in some sense, consistently across instances. The first such axiom that we examine is \emph{independence}, which has previously been studied by \citet{intriligator1973probsc}. 
The idea of this axiom is that the score assigned to a candidate by an aggregation rule should only depend on the scores that the agents assign to this candidate, and should therefore be independent of the scores assigned to other candidates. 
We remark that this axiom is conceptually similar but formally unrelated to Arrow's ``independence of irrelevant alternatives'' \citep{Arro51a}, as it postulates that we can compute the outcome for a candidate without taking into account the remaining candidates.

\begin{definition}[Independence]
    An aggregation rule $F$ satisfies \emph{independence} if, for all instances $\mathI$ and $\mathI'$ that have the same set of agents $N$ and the same set of $m$ candidates, and for all $j\in [m]$ such that $s_{i,j}=s'_{i,j}$ for all $i\in N$, it holds that $F(\mathI)_j=F(\mathI')_j$. 
\end{definition}

Independence is a demanding axiom. In particular, every rule that satisfies it must be coordinate-wise, because we can define the $j$-th coordinate-aggregation function by $f^n_j(s_{1,j},\dots, s_{n,j})=F(\mathcal{I})_j$. 
Nevertheless, we believe that this axiom is appealing in practice, as it is intuitive and greatly simplifies the task of aggregating the agents' score vectors.
Moreover, a violation of independence can lead to complaints from candidates that receive a smaller portion of the resource when agents change their preferences over other candidates, even if their own scores from all agents remain the same.

Next, we introduce \emph{score-monotonicity}. This axiom requires that, if an agent increases the score of some candidate, then the aggregated score of this candidate should weakly increase as well. 
Score-monotonicity was previously studied by \citet{freeman2021truthfulbudget},\footnote{\citet{freeman2021truthfulbudget} simply called this notion ``monotonicity''.} and similar monotonicity notions are omnipresent in social choice theory. 

\begin{definition}[Score-monotonicity] \label{defn-score_monotonicity}
    An aggregation rule $F$ is \emph{score-monotone} if $F(\mathI)_j\leq F(\mathI')_j$
     for all instances $\mathcal{I}$, $\mathcal{I'}$ that have the same set of agents $N$ and the same set of $m$ candidates, and for all $j\in[m]$ for which there exists an agent $i\in N$ such that \emph{(i)} $\mathbf{s}_{\ell}=\mathbf{s}'_{\ell}$ for all $\ell\in N\setminus \{i\}$, \emph{(ii)} $s_{i,j} < s'_{i,j}$, and \emph{(iii)} $s_{i,k} \geq  s'_{i,k}$ for all $k\in [m]\setminus \{j\}$.
\end{definition}

The final consistency notion that we study is \emph{reinforcement}, which demands that if an aggregation rule chooses the same outcome for two instances with disjoint sets of agents, then it also chooses that outcome when combining the two instances. 
Variants of this axiom feature prominently in numerous results in social choice theory \citep[e.g.,][]{Youn75a,Fish78d,YoLe78a,Bran13a}. 

\begin{definition}[Reinforcement]
    An aggregation rule $F$ satisfies \emph{reinforcement} if, for all instances $\mathI=(\mathbf{s}_1,\dots, \mathbf{s}_n)$ and $\mathI'=(\mathbf{s}_1',\dots, \mathbf{s}_{n'}')$ that have the same set of $m$ candidates and satisfy $F(\mathI)=F(\mathI')$, it holds that $F(\mathbf{s}_1,\dots, \mathbf{s}_n,\mathbf{s}_1',\dots, \mathbf{s}_{n'}')=F(\mathI)$. 
\end{definition}

Independence, score-monotonicity, and reinforcement are logically unrelated, as they formalize rather different notions of consistency. However, we emphasize that variants of these three axioms are well-established in the literature and we therefore believe that it is important to study all of them.

\subsubsection{Incentive Properties}

Our last category of axioms is concerned with the incentives of agents: aggregation rules should incentivize agents to participate and to report their preferences truthfully. These ideas lead to the well-known notions of \emph{participation} and \emph{strategyproofness}. We start by defining strategyproofness, which stipulates that agents should not be able to benefit from lying about their true preferences. 

\begin{definition}[Strategyproofness] \label{defn-strategyproofness}
    An aggregation rule $F$ is \emph{strategyproof} if, 
     for all instances $\mathcal{I}$ and $\mathcal{I'}$ that have the same set of agents $N$ and the same set of $m$ candidates, and for each agent $i\in N$ such that $\mathbf{s}_{i'}=\mathbf{s}_{i'}'$ for all $i'\in N\setminus \{i\}$,
     it holds that 
     $\sum_{j \in [m]} |s_{i,j} - F(\mathI)_j|\le \sum_{j \in [m]} |s_{i,j} - F(\mathI')_j|$.
\end{definition}

While it is known that \im{} and \util{} satisfy strategyproofness \citep{goel2019knapsack,freeman2021truthfulbudget}, this property is in general rather demanding, as demonstrated by the impossibility theorem of \citet{brandt2024optimal} stating that no aggregation rule simultaneously satisfies strategyproofness, Pareto optimality, and single-minded proportionality.

Participation is a property closely related to strategyproofness. It dictates that agents should not be able to profit by abstaining. 
Put differently, participation ensures that it is always weakly better for every agent to express her preference.

\begin{definition}[Participation] \label{defn-participation}
    An aggregation rule $F$ satisfies \emph{participation} if,  for all instances $\mathcal{I}$ and $\mathcal{I}'$ such that $\mathcal{I}'$ is obtained from $\mathcal{I}$ by removing agent $i$, it holds that
    $\sum_{j \in [m]} |s_{i,j} - F(\mathI)_j|\le \sum_{j \in [m]} |s_{i,j} - F(\mathI')_j|$.
\end{definition}

We note that participation and strategyproofness are logically independent in general. However, when imposing reinforcement and the very mild condition that $F(\mathbf{x})=\mathbf{x}$ (i.e., if there is a single agent, we choose her ideal distribution), it can be shown that strategyproofness implies participation.

\section{Efficiency Properties} \label{sec:score_unanimity_rr}

We now analyze our aggregation rules with respect to the axioms defined in \Cref{subsec:axioms}. In this section, we study the three efficiency properties and show that, while \util{} and \egal{} satisfy Pareto optimality, all other rules except \avg{} fail even score-unanimity. 
Recall from \Cref{thm:efficiency_implications} that Pareto optimality implies range-respect which in turn implies score-unanimity, and that range-respect implies Pareto optimality if $m = 2$ or $n = 2$.

\begin{theorem} \label{thm:efficiencyprops}
    The following claims hold.
        \begin{enumerate}[label=(\arabic*),topsep=4pt,itemsep=0pt]
        \item \emph{\util{}} and \emph{\egal{}} are Pareto optimal (and thus range-respecting and score-unanimous).
        \item \avg{} is range-respecting (and thus Pareto optimal when $m = 2$ or $n = 2$, and score-unanimous for all $m,n\ge 2$), but fails Pareto optimality for all $m\geq 3$ and $n\geq 3$.
        \item \med{} is range-respecting when $m\leq 3$ or $n=2$ (and thus Pareto optimal when $m = 2$ or $n = 2$).
        If $m = 3$, it is Pareto optimal if $n\ge 3$ is odd, but fails Pareto optimality if $n\ge 4$ is even.
        It fails score-unanimity for all $m \geq 4$ and $n \geq 3$.
        \item \maxrule{}, \minrule{}, \geo{}, and \im{} are Pareto optimal when $m=2$, but fail score-unanimity for all $m\geq 3$ and $n\geq 2$.
        \end{enumerate}
\end{theorem}
\begin{proof}
We prove each of the claims separately.\medskip

\noindent\textbf{Claim 1}: The claim follows directly from the definitions of \emph{\util{}} and \emph{\egal{}}, since a Pareto improvement would also give rise to an improvement with respect to the welfare measure. In more detail, 
for \emph{\egal{}}, we use the fact that its definition lexicographically minimizes the maximum disutility of the agents. In particular, if an outcome $\mathbf{x}$ returned by \emph{\egal{}} was Pareto dominated by another outcome $\mathbf{x'}$, then $\mathbf{x'}$ would have been selected over $\mathbf x$ by the leximin optimization procedure, which contradicts the definition of \emph{\egal{}}.
By \Cref{thm:efficiency_implications}, both rules are range-respecting and score-unanimous as well. \medskip

\noindent\textbf{Claim 2}: Consider an instance $\mathcal{I}$. For each $j \in [m]$ it holds that $\min_{i\in N} s_{i,j}\leq \frac{1}{n}\sum_{i\in N} s_{i,j}\leq \max_{i\in N} s_{i,j}$, which shows that \avg{} is range-respecting. By Claim 2 of \Cref{thm:efficiency_implications}, this also means that \avg{} is Pareto optimal if $m = 2$ or $n = 2$.

To show that \avg{} fails Pareto optimality when $m\geq 3$ and $n\geq 3$, consider the following instance $\mathI^1$, where all candidates in $C\setminus \{c_1,c_2,c_3\}$ receive score $0$ from all agents and can thus be ignored. For this instance, \avg{} outputs the vector $(\frac{1}{2n}, \frac{2}{2n}, 1-\frac{3}{2n})$. 
    However, the outcome $\mathbf{x'}=(0,\frac{3}{2n}, 1-\frac{3}{2n})$ decreases the disutility of agent $1$ without increasing the disutility of the other agents, so \avg{} is not Pareto optimal.

    \begin{center}
        \begin{tabular}{ c | c c c c}
          $\mathcal{I}^1$ & $s_{i,1}$ & $s_{i,2}$ & $s_{i,3}$ \\ 
         \hline \hline
         $1$ & $0$ & $\frac{1}{2}$ & $\frac{1}{2}$\\ 
         $2$ & $\frac{1}{2}$ & $\frac{1}{2}$ & $0$\\ 
         $i\in\{3,\dots,n\}$ & $0$ & $0$ & $1$ \\
        \end{tabular}       
    \end{center}
\medskip

\noindent\textbf{Claim 3}: 
Since \med{} is equivalent to \avg{} when $n=2$, it is range-respecting for $n=2$ by Claim~2, and also Pareto optimal by Claim~2 of \Cref{thm:efficiency_implications}. For the case $m=2$, it follows from Proposition~\ref{prop:nonorm} that if \med{} outputs $(x_1, x_2)$
on an instance $\mathI$, then $x_j=\text{med}(s_{1, j}, \dots, s_{n, j})$ for $j\in[2]$. 
This implies that \med{} is range-respecting for $m=2$, and Claim 2 of \Cref{thm:efficiency_implications} entails that it is also Pareto optimal. 

Now, when $m=3$, suppose for contradiction that there exists an instance $\mathcal{I}$ such that the outcome~$\mathbf{x}$ chosen by \med{} fails range-respect. Without loss of generality, we assume that $x_1<\min_{i\in N} s_{i,1}$ (the choice of the candidate does not matter, and if $x_1>\max_{i\in N} s_{i,1}$, we can reverse all inequalities in the proof). Moreover, let $m_1$, $m_2$, $m_3$ denote the medians (before normalization) for the candidates $c_1$, $c_2$, $c_3$, respectively, and let $M=m_1+m_2+m_3$. This means that $\mathbf{x}=\left(\frac{m_1}{M}, \frac{m_2}{M}, \frac{m_3}{M}\right)$. Since $m_1\geq \min_{i\in N} s_{i,1}$, we infer from $x_1<\min_{i\in N} s_{i,1}$ that $M>1$. 

Let $\gamma=\min_{i\in N} s_{i,1}$; we will next show that $m_2+m_3\leq 1-\gamma$. 
To this end, we note that each median by itself is monotone (i.e., if we increase some value $s_{i,j}$, then $m_j$ does not decrease). 
Recall that by our choice of $\gamma$, we have $s_{i, 2}+s_{i, 3}\le 1-\gamma$ for all $i\in N$.
Consider modified values $\hat s_{i,2}$ and $\hat s_{i,3}$ that satisfy $\hat s_{i,2}\geq s_{i,2}$, $\hat s_{i,3}\geq s_{i,3}$, and $\hat s_{i,2}+\hat s_{i,3}=1-\gamma$ for all $i\in N$. 
By the monotonicity of the medians, it follows that the corresponding medians $\hat m_2$ and $\hat m_3$ satisfy $\hat m_2\geq m_2$ and $\hat m_3\geq m_3$. 
Moreover, it holds that $\hat s_{i,2}=1-\gamma-\hat s_{i,3}$. 
This means that, for all agents $i$ and $\ell$, we have $\hat s_{i,2}\leq \hat s_{\ell,2}$ if and only if $\hat s_{i,3}\geq \hat s_{\ell,3}$.
Consequently, in the modified instance the median agent(s) for $c_2$ are also the median agent(s) for $c_3$, which implies that $\hat m_2+\hat m_3=1-\gamma$ and hence $m_2+m_3\le 1-\gamma$. 

We can now derive a contradiction. Since $M>1$, 
it follows that $x_2=\frac{m_2}{M}<m_2$ and $x_3=\frac{m_3}{M}<m_3$.
This means that $x_1+x_2+x_3< x_1+m_2+m_3< \gamma+(1-\gamma)=1$, a contradiction with $\mathbf x$ being an outcome for $\mathcal I$.

    \medskip
    Next, we prove that \med{} satisfies Pareto optimality for $m=3$ and every odd $n \geq 3$.
    Let $m_1$, $m_2$, $m_3$ again denote the medians (before normalization) for the candidates $c_1$, $c_2$, $c_3$, respectively, and set $M=m_1 + m_2 + m_3$.
    Assume that $M \leq 1$; we will explain how to modify the argument for $M>1$ toward the end of the proof.
    We show that for each $j\in[3]$, there is a non-empty set of agents $T_j$ such 
    that all agents in $T_j$ prefer ${\mathbf x} = (\frac{m_1}{M}, \frac{m_2}{M}, \frac{m_3}{M})$ to every outcome ${\mathbf x}'$ with $x'_j < x_j$.
    Thus, for an outcome ${\mathbf x}'$ to Pareto dominate $\mathbf x$, it would have to be the case that $x'_j\ge x_j$ for all $j\in\{1, 2, 3\}$; as $x'_1+x'_2+x'_3=x_1+x_2+x_3$, this implies ${\mathbf x}'=\mathbf x$.
    
    Specifically, for $j \in [3]$, let $T_j \subseteq N$ be the set of agents $i$ such that $s_{i,j} \geq m_j$ and $s_{i,k} \leq m_{k}$ for all other $k \neq j$.
    We first show that all agents in $T_j$
    prefer $\mathbf x$ to every outcome ${\mathbf x}'$ with $x'_j < x_j$; then, we will argue that $T_j$ is non-empty for all $j\in[3]$.
    By symmetry, without loss of generality we can focus on the case $j=1$.
    
    Consider an agent $i\in T_1$. Since $M \leq 1$, for $k\in\{2, 3\}$ we have $x_{k} = \frac{m_{k}}{M} \geq m_{k}\ge s_{i, k}$.
    Moreover, as $s_{i,1} + s_{i,2} + s_{i,3} = x_1+x_2+x_3 = 1$, it follows that $s_{i,1} \geq x_1$.
    Thus, 
    $d_i({\mathbf x})=(s_{i, 1}-x_1)+(x_2-s_{i, 2})+(x_3-s_{i, 3})$. Now, consider an outcome 
    ${\mathbf x}'$ with $x'_1 < x_1$, and note that this implies 
    $x'_2+x'_3> x_2+x_3$.
    We then have 
    \begin{align*}
    d_i({\mathbf x}')
    &= |s_{i, 1}-x'_1|+|s_{i, 2} - x'_2|+|s_{i, 3}-x'_3|\\
    &\ge (s_{i, 1}-x'_1)+(x'_2-s_{i, 2})+(x'_3-s_{i, 3})\\ 
    &= (s_{i, 1}-x'_1)+(x'_2+x'_3)-s_{i, 2}-s_{i,3} > (s_{i, 1}-x_1)+(x_2+x_3)-s_{i, 2}-s_{i,3} = d_i({\mathbf x}),
    \end{align*}
    i.e., agent $i$ strictly prefers $\mathbf x$ to ${\mathbf x}'$.
    
    It remains to show that $T_1 \neq \emptyset$.
    To this end, 
    let $N_1=\{i\in N\colon s_{i,1} \geq m_1\}$.
    Suppose for contradiction that $T_1 = \emptyset$. This implies that for every agent $i \in N_1$, there exists $k \neq 1$ with $s_{i,k} > m_{k}$.
    Now, consider an agent $i \in N\setminus N_1$; we have $s_{i,1} < m_1$.
    We claim that for this agent, too, there exists some $k \neq 1$ with $s_{i,k} > m_{k}$---indeed, otherwise $s_{i,1} + s_{i,2} + s_{i,3} < m_1 + m_2 + m_3 \leq 1$, a contradiction.
    Thus, for each agent $i \in N$, there exists $k \neq 1$ with $s_{i,k} > m_{k}$.
    However, by definition of the median, it holds that $\sum_{k \in\{2, 3\}} |\{i \in N \colon s_{i,k} > m_{k}\}| \leq \frac{n-1}{2} + \frac{n-1}{2} = n-1$, a contradiction.
    Hence, $T_1 \neq \emptyset$.
    
    As the same argument works for all $j\in[3]$, it follows that ${\mathbf x}$ is Pareto optimal.
    The argument for the case $M >1$ is symmetric: it suffices to reverse the signs (including in the definition of $T_j$).
    
    On the other hand, we show that \med{} fails Pareto optimality for $m = 3$ and any even $n\geq 4$. To this end, consider the following instance $\mathcal{I}^2$ for an even number of agents $n$. 
        \begin{center}
            \begin{tabular}{ c | c c c}
              $\mathcal{I}^2$ & $s_{i,1}$ & $s_{i,2}$ & $s_{i,3}$  \\ 
             \hline \hline
             $1$ & $\frac{3}{10}$ & $\frac{7}{20}$ & $\frac{7}{20}$ \\ 
             $2$ & $\frac{6}{10}$ & $0$ & $\frac{4}{10}$ \\
             $i \in \{3,\dots,\frac{n}{2}+1 \}$ & $\frac{7}{10}$ & $\frac{1}{4}$ & $\frac{1}{20}$ \\ 
             $i \in \{\frac{n}{2}+2,\dots,n \}$ & $\frac{1}{4}$ & $\frac{7}{20}$ & $\frac{4}{10}$ \\ 
              
            \end{tabular}       
        \end{center}
    In this instance, it holds that $\text{med}_{i\in N} s_{i,1}=\frac{9}{20}$, $\text{med}_{i\in N} s_{i,2}=\frac{3}{10}$, and $\text{med}_{i\in N} s_{i,3}=\frac{3}{8}$, so \med{} returns the outcome $\mathbf{x} = \left( \frac{2}{5}, \frac{4}{15}, \frac{1}{3} \right)$.
    However, the outcome $\left( \frac{2}{5}, \frac{5}{20}, \frac{7}{20}\right)$ strictly benefits agent $2$ and does not hurt the other agents, so $\mathbf{x}$ is not Pareto optimal.
        \medskip

Finally, we show that \med{} fails score-unanimity when $m\geq 4$ and $n\geq 3$. To this end, consider the following instance $\mathcal{I}^3$ for an odd number of agents $n$, where all candidates $c_j$ with $5\leq j\leq m$ receive score $0$ from all agents. 
    \begin{center}
        \begin{tabular}{ c | c c c c}
          $\mathcal{I}^3$ & $s_{i,1}$ & $s_{i,2}$ & $s_{i,3}$ & $s_{i,4}$ \\ 
         \hline \hline
         $1$ & $\frac{3}{10}$ & $\frac{5}{10}$ & $0$ & $\frac{2}{10}$ \\ 
         $i\in \{2,\dots, \frac{n+1}{2}\}$ & $\frac{3}{10}$ & $\frac{1}{10}$ & $\frac{6}{10}$ & $0$ \\ 
         $i\in \{\frac{n+3}{2},\dots, n\}$ & $\frac{3}{10}$ & $\frac{7}{10}$ & $0$ & $0$ 
        \end{tabular}       
    \end{center}
    In this instance, we have $\text{med}_{i\in N} s_{i,1}=\frac{3}{10}$, $\text{med}_{i\in N} s_{i,2}=\frac{5}{10}$, and $\text{med}_{i\in N} s_{i,3}=\text{med}_{i\in N} s_{i,4}=0$, so \med{} assigns probability $\frac{3}{8}$ to $c_1$. 
To obtain a counterexample for even $n$, it suffices to duplicate agent~$1$. 
For $n\ge 6$, all medians remain the same as for odd $n$, and for $n=4$
the fourth median $\text{med}_{i\in N} s_{i,4}$ changes to $\frac{1}{10}$ while all other medians remain the same.
In either case, \med{} assigns probability different from $\frac{3}{10}$ to candidate $c_1$.
\medskip

\noindent\textbf{Claim 4}: 
We first present an example showing that 
\geo{}, \maxrule{}, and \minrule{}
fail score-unanimity when $m\geq 3$ and $n\geq 2$. To this end, consider the following instance $\mathcal{I}^4$. 

\begin{center}
        \begin{tabular}{ c | c c c c }
          $\mathcal{I}^4$ & $s_{i,1}$ & $s_{i,2}$ & $s_{i,3}$\\ 
         \hline \hline
         $1$ & $\frac{1}{4}$ & $\frac{1}{2}$ & $\frac{1}{4}$\\ 
         $i\in \{2,\dots, n\}$ &  $\frac{1}{4}$ & $\frac{1}{4}$ & $\frac{1}{2}$ \\ 
        \end{tabular}       
    \end{center}\smallskip    
    
It is immediate that on this instance \maxrule{} outputs $(\frac15, \frac25, \frac25)$ and \minrule{} outputs $(\frac13, \frac13, \frac13)$; both outcomes fail score-unanimity with respect to $c_1$.
Furthermore, 
\geo{} outputs the vector $$\mathbf{x}=\left(\frac{1}{1+2^{\frac{1}{n}} + 2^{\frac{n-1}{n}}}, \frac{2^{\frac{1}{n}}}{1+2^{\frac{1}{n}} + 2^{\frac{n-1}{n}}}, \frac{2^{\frac{n-1}{n}}}{1+2^{\frac{1}{n}} + 2^{\frac{n-1}{n}}}\right).$$ 
For every $n\ge 2$ it holds that
$(2^{\frac{n-1}{n}} - 1)(2^{\frac{1}{n}} - 1) > 0$; expanding, we obtain
$2^{\frac{1}{n}}+2^{\frac{n-1}{n}}< 1+2^{\frac{n-1}{n}+\frac{1}{n}} = 3$, and hence $x_1>\frac14$.
This implies that $\mathbf x$
fails score-unanimity. 
To extend the counterexample to larger~$m$, one can add candidates that receive a score of $0$ from all agents.

Next, we will show that \geo{}, \maxrule{}, and \minrule{} are range-respecting (and hence Pareto optimal) for $m=2$.
Finally, we will consider \im{} and show that it satisfies range-respect for $m=2$ but fails score-unanimity for $m\ge 3$ and all $n\ge 2$.
\medskip

\emph{\geo{}}: 
To show that \geo{} is range-respecting for $m=2$, fix some instance $\mathcal{I}$ and let $\mathbf{x}$ be the outcome returned by \geo{}. Without loss of generality, we focus on $c_1$ and show that $\min_{i \in N} s_{i,1} \leq x_1 \leq \max_{i \in N} s_{i,1}$. 
For $j\in[2]$, let $y_j=\sqrt[n]{\prod_{i\in N} s_{i,j}}$ and observe that $\min_{i\in N} s_{i,j}\leq y_j\leq \max_{i\in N} s_{i,j}$. Multiplying these inequalities for $y_1$
by $(1-\max_{i\in N} s_{i,1})$ and $(1-\min_{i\in N} s_{i,1})$, we obtain 
 \begin{align*}
        y_1 \cdot (1-\max_{i \in N} s_{i,1}) &\leq \max_{i \in N} s_{i,1} \cdot (1-\max_{i \in N} s_{i,1});\\
        y_1 \cdot (1-\min_{i \in N} s_{i,1}) &\geq \min_{i \in N} s_{i,1} \cdot (1-\min_{i \in N} s_{i,1}).
\end{align*}
Rearranging the terms yields 
\begin{equation*}
        \frac{y_1}{y_1 + 1 - \max_{i \in N} s_{i,1}} \leq \max_{i \in N} s_{i,1}\qquad \text{and}\qquad \frac{y_1}{y_1 + 1 - \min_{i \in N} s_{i,1}} \geq \min_{i \in N} s_{i,1}.
\end{equation*}
Since $m=2$, it holds that $1 - \min_{i \in N} s_{i,1}=\max_{i\in N} s_{i,2}$ and  $1 - \max_{i \in N} s_{i,1}=\min_{i\in N} s_{i,2}$. Hence, we derive that 
\begin{align*}
    \min_{i\in N} s_{i,1} \leq \frac{y_1}{y_1 + \max_{i \in N} s_{i,2}}\leq \frac{y_1}{y_1 +y_2}\leq \frac{y_1}{y_1 + \min_{i \in N} s_{i,2}}\leq \max_{i\in N} s_{i,1}. 
\end{align*}
As $x_1=\frac{y_1}{y_1 +y_2}$, this shows that \geo{} is range-respecting if $m=2$.

 \medskip
 
\emph{\maxrule{}}: To show that \maxrule{} is range-respecting if $m=2$, we fix an instance $\mathcal{I}$ and focus on candidate~$c_1$. 
Observe that $\max_{i\in N} s_{i,1}+\max_{i\in N} s_{i,2}\geq 1$, which implies that 
$$
x_1=\frac{\max_{i\in N} s_{i,1}}{\max_{i\in N} s_{i,1}+\max_{i\in N} s_{i,2}}\leq \max_{i\in N} s_{i,1}.
$$
We will now argue that $x_1\ge \min_{i\in N} s_{i,1}$.
For $m=2$, it holds that $\max_{i\in N} s_{i,2}=1-\min_{i\in N} s_{i,1}$, so 
we have
\begin{equation*}
    x_1  = \frac{\max_{i\in N} s_{i,1}}{\max_{i\in N} s_{i,1} + \max_{i\in N} s_{i,2}}= \frac{\max_{i\in N} s_{i,1}}{\max_{i\in N} s_{i,1} + 1-\min_{i\in N} s_{i,1}}.
\end{equation*}
Thus, it suffices to show that
$$
\min_{i\in N} s_{i,1} (1+\max_{i\in N} s_{i,1} - \min_{i\in N} s_{i,1}) \le 
\max_{i\in N} s_{i,1}. 
$$
To see this, observe that $\min_{i\in N} s_{i,1}\le 1$,
and therefore
\begin{align*}
    \min_{i\in N} s_{i,1} (1+\max_{i\in N} s_{i,1} - \min_{i\in N} s_{i,1})
    &= \min_{i\in N} s_{i,1} + \min_{i\in N} s_{i,1} (\max_{i\in N} s_{i,1} - \min_{i\in N} s_{i,1}) \\
    &\leq \min_{i\in N} s_{i,1} + (\max_{i\in N} s_{i,1} - \min_{i\in N} s_{i,1}) = \max_{i\in N} s_{i,1},  
\end{align*}
as required.
Hence, \maxrule{} is range-respecting if $m=2$. 

\medskip

\emph{\minrule{}}: To show that \minrule{} is range-respecting if $m=2$, we consider an instance $\mathcal{I}$ and focus again on candidate $c_1$.
The argument is symmetric to that for \maxrule{}. Specifically, we observe that $\min_{i\in N} s_{i,1}+\min_{i\in N} s_{i,2}\leq 1$, so $\min_{i\in N} s_{i,1}\leq \frac{\min_{i\in N} s_{i,1}}{\min_{i\in N} s_{i,1}+\min_{i\in N} s_{i,2}}=x_1$. 
To see that $x_1\le \max_{i\in N} s_{i, 1}$, 
we note that $\min_{i\in N} s_{i,1}\le\max_{i\in N} s_{i,1}\le 1$, and therefore
\begin{align*}
    \max_{i\in N} s_{i,1}(1+\min_{i\in N} s_{i,1}-\max_{i\in N} s_{i,1})
    &= \max_{i\in N} s_{i,1} - \max_{i\in N} s_{i,1} (\max_{i\in N} s_{i,1}-\min_{i\in N} s_{i,1}) \\
    &\geq \max_{i\in N} s_{i,1} -  (\max_{i\in N} s_{i,1}-\min_{i\in N} s_{i,1})
    = \min_{i\in N} s_{i,1}.
\end{align*}
Using that $\min_{i\in N} s_{i,2}=1-\max_{i\in N} s_{i,1}$, we thus have 
\begin{equation*}
    \max_{i\in N} s_{i,1}\geq \frac{\min_{i\in N} s_{i,1}}{\min_{i\in N} s_{i,1} + \min_{i\in N} s_{i,2}}=x_1.
\end{equation*}
This completes the proof that \minrule{} is range-respecting. 

\medskip

\emph{\im{}}: The last rule we consider is \im{}.
We first show that it is range-respecting (and therefore Pareto optimal) when $m=2$. 
Consider an instance $\mathcal{I}$, and let $\mathbf{x}$ be the output of \im{}. 
Note that $x_j\leq  \max_{i\in N} s_{i,j}$ for $j\in [2]$, as one phantom is always at $0$. Moreover, since 
$s_{i,j}=1-s_{i,3-j}$ for all $i\in N$ and $x_j=1-x_{3-j}$, our previous observation implies that
$$
x_j=1-x_{3-j}\geq  1 - \max_{i\in N} s_{i,3-j} = 1 - \max_{i\in N} (1 - s_{i,j}) = 1- 1+ \min_{i\in N} s_{i, j} = \min_{i\in N} s_{i, j}.
$$
Hence, \im{} is range-respecting when $m=2$. 

To show that \im{} fails score-unanimity when $m\geq 3$ and $n\geq 2$, we consider the instance $\mathcal{I}^5$ shown below, where all candidates $c_j$ with $j\geq 4$ receive a score of $0$ from all agents.
 
\begin{center}
        \begin{tabular}{ c | c c c c }
          $\mathcal{I}^5$ & $s_{i,1}$ & $s_{i,2}$ & $s_{i,3}$  \\ 
         \hline \hline
         $1$ & $\frac{n+1}{n+2}$ & $\frac{1}{n+2}$ & $0$ \\ 
         $i\in \{2,\dots, n\}$ &  $\frac{n+1}{n+2}$ & $0$ & $\frac{1}{n+2}$ 
        \end{tabular}       
    \end{center}
    For this instance, score-unanimity requires that $x_1 = \frac{n+1}{n+2}$. 
    However, we claim that \im{} assigns probability $\frac{n}{n+2}$ to $c_1$. Indeed, 
for $t^*=\frac{1}{n+2}$ we have $x_1(t^*)=\frac{n}{n+2}$
(with $f_n^\im{}(t^*)$ selected as the median) and $x_2(t^*) = x_3(t^*)=\frac{1}{n+2}$, 
so that $x_1(t^*)+x_2(t^*)+x_3(t^*)=1$. On the other hand, for $t<t^*$ we have
$x_1(t)\le x_1(t^*)$ and $x_3(t)\le x_3(t^*)$ (by monotonicity of the median), and $x_2(t)=f_1^\im{}(t) = t<x_2(t^*)$, so
$x_1(t)+x_2(t)+x_3(t)<1$.
It follows that \im{} violates score-unanimity.
\end{proof}

\begin{remark}
   \citet[p. 22]{freeman2021truthfulbudget} suggested a variant of \im{} where the last moving phantom is fixed to $1$, i.e., $f^{\im{}'}_n(t)=1$. This modified rule satisfies range-respect, single-minded proportionality, and strategyproofness. However, it is unclear whether it inherits other desirable properties of \im{} such as reinforcement. 
\end{remark}

\begin{remark}
    In some settings, it is computationally challenging to determine whether an outcome is efficient \citep[e.g.,][]{aziz2019efficient}. This is not the case in our context: it is straightforward to check whether an outcome satisfies score-unanimity and range-respect, and we give a linear programming formulation for deciding whether an outcome is Pareto optimal in the appendix (see~\Cref{thm:po_check}). 
\end{remark}

\section{Fairness Properties} \label{sec:proc_sr}

We next turn to fairness properties, and study our rules with respect to single-minded proportionality and the more demanding notion of score-representation. In particular, we will show that, among the rules we consider, the only one to satisfy score-representation is \avg{}, thus making a strong case for this rule. 
Note that some of our results follow from the work of \citet{freeman2021truthfulbudget}, who have shown that \im{} satisfies single-minded proportionality, while \emph{\util{}} fails this property (without specifying the ranges of $n$ and $m$ for which this is the case). 

\begin{theorem}
    The following claims hold. 
    \begin{enumerate}[label=(\arabic*),topsep=4pt,itemsep=0pt]
    \item \avg{} satisfies score-representation (and therefore single-minded proportionality).
    \item \im{} satisfies score-representation when $m=2$ and single-minded proportionality for all $m\ge 2$, but fails score-representation for all $n\geq 2$ and $m\geq 3$.
    \item \maxrule{} and \emph{\util{}} satisfy score-representation when $n=m=2$ and single-minded proportionality when $n=2$. Both rules fail score-representation for all $m\geq 3$ and $n\geq 2$, and single-minded proportionality for all $m\geq 2$ and $n\geq 3$. 
    \item \med{} and \emph{\egal{}} satisfy score-representation when $n=2$, but fail single-minded proportionality for all $m\geq 2$ and $n\geq 3$.
    \item \minrule{} and \geo{} satisfy single-minded proportionality when $n=m=2$, but fail to do so for all $n \geq 3$ or $m \geq 3$.
    Both rules fail score-representation for all $m\geq 2$ and $n\geq 2$.
        \end{enumerate}
\end{theorem}
\begin{proof} We prove each of the claims separately.\medskip

\noindent\textbf{Claim 1}: Consider an instance $\mathcal{I}$, and fix a candidate $c_j$ and a value $\gamma\in [0,1]$. We need to show that
the vector $\mathbf{x}$ chosen by \avg{} satisfies $x_j\geq \gamma\cdot \frac{\mathcal{N}(\mathcal{I}, c_j, \gamma)}{n}$. For this, let $S:=\{i\in N\colon s_{i,j}\geq \gamma\}$ denote the set of agents that report a score of at least $\gamma$ for $c_j$, so that $|S|=\mathcal{N}(\mathcal{I}, c_j, \gamma)$. It holds that $x_j=\frac{1}{n}\sum_{i\in N} s_{i,j}\geq \frac{1}{n}\sum_{i\in S} s_{i,j} \geq \gamma \cdot \frac{\mathcal{N}(\mathcal{I}, c_j, \gamma)}{n}$, which shows that \avg{} satisfies score-representation.\medskip

\noindent\textbf{Claim 2}: First, it was shown by \citet{freeman2021truthfulbudget} that \im{} satisfies single-minded proportionality. Next, we show that \im{} also satisfies score-representation if $m=2$. To this end, consider an instance~$\mathI$ and assume for contradiction that the outcome $\mathbf{x}$ of \im{} fails score-representation. Without loss of generality, suppose that there exists $\gamma\in [0,1]$ such that $x_1<\gamma\cdot \frac{\mathcal{N}(\mathI, c_1, \gamma)}{n}$. Clearly, this means that $\gamma>x_1$. Now, consider the instance $\mathI'$ such that each agent $i$ with $s_{i,1}>x_1$ reports $s_{i,1}'=1$ and $s'_{i, 2}=0$ while all other agents submit the same scores as in $\mathI$, and let ${\mathbf x}'$ be the output of \im{} on $\mathI'$.
We then have $x_1'=x_1$, because increasing the score of $c_1$ for agents who already assign to it a score above the median does not change the position of the median. Since $\gamma>x_1$, we have $\mathcal{N}(\mathI, c_1, \gamma)\leq \mathcal{N}(\mathI', c_1, 1)$ and hence $x_1<\gamma\cdot\frac{\mathcal{N}(\mathI, c_1, \gamma)}{n}\le \frac{\mathcal{N}(\mathI', c_1, 1)}{n}$. 
Next, let $\mathI''$ denote the instance derived from $\mathI'$ by setting $s_{i,1}''=0$ and $s_{i, 2}''=1$ for all agents $i\in N$ with $s_{i,1}\leq x_1$, and let ${\mathbf x}''$ be the output of \im{} on this instance. Since \im{} is known to be score-monotone \citep[see][Thm.~3]{freeman2021truthfulbudget}, it follows that $x_1''\leq x_1'$. Moreover, it holds that $\mathcal{N}(\mathI', c_1, 1)=\mathcal{N}(\mathI'',c_1,1)$, so we can again infer that $x_1''< \frac{\mathcal{N}(\mathI'', c_1, 1)}{n}$. However, in $\mathI''$, all agents are single-minded, so $x_1''= \frac{\mathcal{N}(\mathI'', c_1, 1)}{n}$ because \im{} satisfies single-minded proportionality. This yields the desired contradiction, which means that \im{} satisfies score-representation for $m=2$. 

Finally, to show that \im{} fails score-representation for all $m\geq 3$ and $n\geq 2$, it suffices to consider the instance $\mathI^5$ in the proof of \Cref{thm:efficiencyprops}. In this instance, all agents assign score $\frac{n+1}{n+2}$ to $c_1$, but \im{} only assigns a score of $\frac{n}{n+2}$ to this candidate, so score-representation is violated.\medskip

\noindent\textbf{Claim 3}: 
We will first argue that both \maxrule{} and \util{} satisfy single-minded proportionality for $n=2$, but fail it if 
$n\ge 3$ and $m\ge 2$.

Suppose first that $n=2$. Let $\mathI$ denote a $2$-agent instance where both agents are single-minded. 
If both agents assign score $1$ to the same candidate (and hence $0$ to all other candidates), then both \maxrule{} and \util{} assign score $1$ to this candidate as well. 
Now, suppose that the two agents assign score~$1$ to different candidates, say $c_j$ and $c_k$. We claim that both of these candidates receive a probability of $\frac{1}{2}$ under both \maxrule{} and \util{}.
For \maxrule{}, this is immediate from the definition of the rule. 
For \util{}, every outcome that assigns probability~$0$ to candidates in $C\setminus\{c_j, c_k\}$
maximizes the utilitarian welfare, so \util{} assigns probability $\frac12$ to each of $c_j$ and $c_k$ due to the maximum-entropy tie-breaking.
In either case, both \maxrule{} and \util{} satisfy single-minded proportionality. 

To show that \maxrule{} and \util{} fail single-minded proportionality if $n\geq 3$ and $m\geq 2$, consider the following instance $\mathI^6$ (where $s_{i,j}=0$ for all $j\geq 3$ and $i\in N$).  
    \begin{center}
        \begin{tabular}{ c | c c}
          $\mathcal{I}^6$ & $s_{i,1}$ & $s_{i,2}$ \\
         \hline \hline
         $1$ & $1$ & $0$ \\ 
         $i\in \{2,\dots, n\}$ & $0$ & $1$ 
        \end{tabular}       
    \end{center}
    In this instance, \maxrule{} assigns probability $\frac{1}{2}$ to both $c_1$ and $c_2$, but single-minded proportionality requires that $x_2= \frac{n-1}{n} > \frac{1}{2}$.
    On the other hand, \util{} assigns probability $1$ to $c_2$, but single-minded proportionality requires that $x_1=\frac{1}{n} > 0$.
    
To establish our results regarding score-representation, we consider \maxrule{} and \util{} separately.\medskip

    \emph{\maxrule{}}:
    We first prove that \maxrule{} satisfies score-representation when $n=m=2$. Consider an instance $\mathI$ with $n=m=2$, and assume without loss of generality that $s_{1,1} \geq s_{2,1}$ and $s_{2,2} \geq s_{1,2}$.
    Then, score-representation is equivalent to the following set of constraints: \emph{(i)} $x_1 \geq \frac{s_{1,1}}{2}$ and $x_2 \geq \frac{s_{2,2}}{2}$, and \emph{(ii)} $x_1 \geq s_{2,1}$ and $x_2 \geq s_{1,2}$.
    Property \emph{(ii)} follows from the fact that \maxrule{} is range-respecting when $m=2$ (see \Cref{thm:efficiencyprops}). On the other hand, $x_1= \frac{s_{1,1}}{s_{1,1} + s_{2,2}}\geq \frac{s_{1,1}}{2}$ and $x_2= \frac{s_{2,2}}{s_{1,1} + s_{2,2}}\geq \frac{s_{2,2}}{2}$ since $s_{i,j}\leq 1$ for all $i\in N$, $j\in [m]$. This demonstrates that \emph{(i)} also holds.

    To show that \maxrule{} fails score-representation when $n\geq 2$ and $m\geq 3$, it is sufficient to consider instance $\mathcal{I}^4$ in the proof of \Cref{thm:efficiencyprops}, which shows that it violates score-unanimity. In this instance, all agents assign score $\frac{1}{4}$ to candidate $c_1$, but \maxrule{} only assigns probability $\frac{1}{5}$ to this candidate. 

    \medskip
    
    \emph{\util{}}: We first show that \util{} satisfies score-representation when $n=m=2$. Just as for \maxrule{}, we consider an instance $\mathI$ with $n=m=2$, assume that $s_{1,1} \geq s_{2,1}$ and $s_{2,2} \geq s_{1,2}$, and prove that \emph{(i)} $x_1 \geq \frac{s_{1,1}}{2}$ and $x_2 \geq \frac{s_{2,2}}{2}$, and \emph{(ii)} $x_1 \geq s_{2,1}$ and $x_2 \geq s_{1,2}$. Property \emph{(ii)} follows immediately, as \util{} is Pareto optimal and therefore also range-respecting. 
    As for \emph{(i)}, observe that any outcome $(x_1, x_2)$ with $s_{2, 1}\le x_1\le s_{1, 1}$ and $s_{1, 2}\le x_2\le s_{2, 2}$
    maximizes the utilitarian welfare, so \util{} picks a maximum-entropy outcome in this range.
    In particular, if  $s_{1,1}\geq \frac{1}{2}\geq s_{1,2}$ and $s_{2,2}\geq \frac{1}{2}\geq s_{2,1}$, it outputs $(\frac{1}{2},\frac{1}{2})$, and we have $\frac12\ge \frac{s_{1, 1}}2$ and $\frac12\ge \frac{s_{2, 2}}2$. Otherwise,  $s_{1,1}<\frac{1}{2}$ or  $s_{2,2}<\frac{1}{2}$; assume without loss of generality that $s_{1,1}<\frac{1}{2}$.
    In this case we also have $s_{2,1}<\frac{1}{2}$ since $s_{2,1}\leq s_{1,1}$, and moreover,  
    $s_{2, 2}=1-s_{2, 1}\ge 1-s_{1, 1}=s_{1, 2}>\frac12$.
    Hence, 
    maximum-entropy tie-breaking forces \util{} to choose the outcome $(s_{1,1}, s_{1,2})$, and we have $s_{1, 1}\geq \frac{s_{1,1}}{2}$ and $s_{1, 2} = 1 -s_{1, 1}\ge \frac12\geq \frac{s_{2,2}}{2}$.  
    
    Finally, to show that \util{} fails score-representation when $n=2$ and $m\geq 3$, 
    we consider the following instance $\mathI^7$, where \util{} chooses the outcome $(\frac{1}{3},\frac{1}{3},\frac{1}{3})$ due to the maximum-entropy tie-breaking. This outcome fails score-representation since the axiom requires $x_1$ to be at least $\frac{1}{2}$. As before, the counterexample can be extended to larger $m$ by adding candidates that receive score $0$ from all agents.
    \begin{center}
        \begin{tabular}{ c | c c c c }
          $\mathcal{I}^7$ & $s_{i,1}$ & $s_{i,2}$ & $s_{i,3}$ \\ 
         \hline \hline
         $1$ & $1$ & $0$ & $0$\\ 
         $2$ & $\frac{1}{3}$ & $\frac{1}{3}$ & $\frac{1}{3}$
        \end{tabular}       
    \end{center}\smallskip

    \noindent\textbf{Claim 4}: If $n=2$, then \med{} and \egal{} coincide with \avg{}: for \med{} this is true by definition, and for \egal{} this was shown in Proposition~\ref{prop:welfare-egal_is_sum}. Thus, both of these rules satisfy score-representation for $n=2$ due to Claim~1. By contrast, if $n\geq 3$, we consider the instance $\mathI^6$ used to show that \maxrule{} and \util{} fail single-minded proportionality.
    For this instance, \med{} returns the outcome $\mathbf{x}=(0,1)$ while \egal{} returns the outcome $\mathbf{x}=(\frac{1}{2}, \frac{1}{2})$.
    Both of these outcomes violate single-minded proportionality.\medskip

\noindent\textbf{Claim 5}: 
    We first prove that \minrule{} and \geo{} satisfy single-minded proportionality when $n=m=2$. There are two cases to consider.
    If both agents give a score of $1$ to the same candidate, it will be assigned a score of $1$ by both rules. On the other hand, if both agents give a score of $1$ to different candidates, then both rules return $\mathbf{x} = (\frac{1}{2}, \frac{1}{2})$, which also satisfies single-minded proportionality.

    To see that \minrule{} and \geo{} fail single-minded proportionality for the case where $n\geq 3$ or $m\geq 3$, we again consider the instance $\mathI^6$ used to show that \maxrule{} and \util{} fail single-minded proportionality (recall that all candidates $c_j$ with $j\geq 3$ receive score $0$ from all agents). 
    In this instance, both rules assign probability $\frac{1}{m}$ to all candidates, which violates single-minded proportionality unless $m=n=2$.
    
    Finally, to see that these rules also fail score-representation for the case $n=m=2$, consider the instance with $\mathbf{s}_1 = (1,0)$ and $\mathbf{s}_2 = (\frac{1}{2},\frac{1}{2})$. 
    Then, score-representation mandates that $x_2 \geq \frac{1}{4}$, but both rules return $\mathbf{x} = (1,0)$.
\end{proof}

\begin{remark}
    In the instance $\mathI^6$ used to show that \minrule{} and \geo{} fail single-minded proportionality, it holds that $\min_{i\in N} s_{i,j}=0$ for all $j\in [m]$. Indeed, whenever $s_{i, j}\in\{0, 1\}$ for all $i\in N$, $j\in [m]$, either there is a unique $j\in [m]$ with $s_{i, j}=1$ for all $i\in N$, or $\min_{i\in N}s_{i, j}=0$ for all $j\in [m]$. 
     Because $\min_{i\in N} s_{i,j}=0$ for all $j\in [m]$, both rules assign a probability of $\frac{1}{m}$ to every candidate, as discussed in the paragraph after Definition~\ref{defn:coord_rules}. 
     However, our results do not depend on this specific convention: for each pair $(m, n)$ where $m\ge 3$ or $n\ge 3$, if we fix an outcome $(x_1, \dots, x_m)$ to be returned by our rule whenever all $m $ coordinate-aggregation functions return~$0$, it is possible to construct an $m$-candidate instance with $n$ single-minded agents for which all coordinate-aggregation functions return $0$ but $(x_1, \dots, x_m)$ is not the proportional outcome.
\end{remark}

\section{Consistency Properties}

In this section, we analyze our aggregation rules with respect to various consistency axioms. In particular, we will show that score-monotonicity and reinforcement are satisfied by almost all of our rules, whereas \avg{} is the only one that fulfills independence. We note that \citet{freeman2021truthfulbudget} have already shown that \im{} and \util{} satisfy reinforcement and score-monotonicity. 

\begin{theorem}\label{thm:consistencyprops}
The following claims hold. 
    \begin{enumerate}[label=(\arabic*),topsep=4pt,itemsep=0pt]
    \item \avg{} satisfies independence. \med{} and \egal{} satisfy independence when $m=2$ or $n=2$ but fail this condition for all $m\geq 3$ and $n\geq 3$. 
        \maxrule{}, \minrule{}, \geo{}, \util{}, and \im{} satisfy independence when $m=2$, but fail to do so for all $m\geq 3$ and $n\geq 2$.
        \item All five coordinate-wise aggregation rules, \util{}, and \im{} satisfy score-monotonicity. \egal{} satisfies score-monotonicity when $m=2$ or $n=2$, but fails to do so for all $m \geq 4$ and $n \geq 4$.
        \item \avg{}, \minrule{}, \maxrule{}, and \geo{} as well as \util{}, \egal{}, and \im{} satisfy reinforcement. \med{} satisfies reinforcement when $m=2$, but fails to do so for all $m \geq 3$. 
    \end{enumerate}
\end{theorem}

We remark that while the bounds on $m$ and $n$ are tight for almost all results in our paper, it remains open whether \egal{} satisfies score-monotonicity when $m=3$ or $n=3$.

\begin{proof}
    We prove each of the claims separately.\medskip

    \noindent\textbf{Claim 1}: We first note that all aggregation rules satisfy independence when $m=2$: indeed, if $s_{i,j}=s_{i,j}'$ for some candidate $c_j$ and all $i\in N$, then $\mathI=\mathI'$.  
    Moreover, if $n=2$, then \med{} and \egal{} coincide with \avg{}: for \med{} this is true by definition, and for \egal{} this was shown in Proposition~\ref{prop:welfare-egal_is_sum}.
     Hence, to show that \med{} and \egal{} satisfy independence for $n=2$, it suffices to show that \avg{} satisfies independence. We prove this next, and then provide examples that establish the remaining parts of the claim.\medskip

\emph{\avg}: Consider two instances $\mathI$ and $\mathI'$ and a candidate $c_{j^*}$ such that $s_{i,j^*}=s'_{i,j^*}$ for all $i\in N$. 
We have $\frac{1}{n} \sum_{i\in N} s_{i,j^*}=\frac{1}{n} \sum_{i\in N} s_{i,j^*}'$. 
It follows from Proposition~\ref{prop:nonorm} that \avg{} assigns the same probability to $c_j$ in $\mathI$ and $\mathI'$.\medskip

\emph{\maxrule{}, \minrule{}, \geo{}, and \im{}}: 
We first observe that all these rules (as well as all other rules we consider) are {\em unanimous}, i.e., if every agent reports the same score vector ${\mathbf s} = (s_1, \dots, s_m)$, the rule outputs $\mathbf s$. This is immediate for all coordinate-wise rules and both welfare-based rules. 
For \im{}, observe that since $f_0^\im{}\equiv 0$, the vector $\mathbf x$ output by \im{} on our instance satisfies $x_j\le s_j$ for all $j\in [m]$; together with $x_1+\dots+x_m = s_1+\dots+s_m=1$, this implies that \im{} is unanimous as well.

Now, consider the instances $\mathI^4$ and $\mathI^5$ used to show that \maxrule{}, \minrule{}, \geo{}, and \im{} fail score-unanimity. In each of these instances, there is a candidate $c_j$ receiving the same score (say, $\gamma$) from each agent, yet the rule under consideration fails to allocate $\gamma$ to $c_j$.
Modify these instances so that $c_j$ still receives score $\gamma$ from each agent, while every other candidate receives score $\frac{1-\gamma}{m-1}$ from each agent. The resulting instances are unanimous, so the rule allocates $\gamma$ to $c_j$. This shows that all these rules violate independence.

\medskip

\emph{\util{}}: First, if $m\geq 3$ and $n\geq 3$ is odd, consider the following instances $\mathI$ and $\mathI'$: in $\mathI$, the first agent assigns score $1$ to $c_3$, the next
$\frac{n-1}{2}$ agents assign score $1$ to $c_1$, and the remaining $\frac{n-1}{2}$ agents assign score $1$ to $c_2$. 
In $\mathI'$, the first agent assigns score $1$ to $c_2$, while all other agents have the same preferences as in $\mathI$. Let $\mathbf x$ and ${\mathbf x}'$ be the outputs of \util{} on $\mathI$ and $\mathI'$, respectively.
If $n = 3$, then $\mathbf{x} = (\frac{1}{3}, \frac{1}{3}, \frac{1}{3})$, while if $n\ge 5$, then $\mathbf{x}=(\frac{1}{2}, \frac{1}{2}, 0)$ (due to the maximum-entropy tie-breaking).
On the other hand, $\mathbf{x}'=(0, 1, 0)$, thus showing that independence is violated for $c_1$. 

If $m\geq 3$ and $n\geq 4$ is even, consider the same instances $\mathI$ and $\mathI'$, but with the last agent removed. That is, 
in $\mathI$, the first agent assigns score $1$ to $c_3$, 
$\frac{n}{2}$ agents assign score $1$ to $c_1$, and $\frac{n}{2}-1$ agents assign score $1$ to $c_2$, and $\mathI'$ is obtained
from $\mathI$ by changing the preference of the first agent to $(0, 1, 0)$. Then \util{} outputs $(1, 0, 0)$ on $\mathI$, but its output on $\mathI'$ is $(\frac12, \frac12, 0)$ (due to the maximum-entropy tie-breaking), so independence is violated for $c_1$. 

Finally, if  $m\ge 3$ and $n=2$, we consider the following
instances $\mathI$ and $\mathI'$: in $\mathI$ agent $1$ reports $(1, 0, 0)$ and agent $2$ reports $(0, 1, 0)$, while in $\mathI'$ agent $1$  reports $(1, 0, 0)$ and agent $2$ reports $(0, \frac12, \frac12)$.
\util{} outputs $(\frac12, \frac12, 0)$ on $\mathI$
and $(\frac13, \frac13, \frac13)$ on $\mathI'$
 due to the maximum-entropy tie-breaking. Hence, independence is violated for $c_1$.\medskip

\emph{\egal{}}: To show that \egal{} fails independence for all $m\geq 3$ and $n\geq 3$, we consider the instances $\mathI$ and $\mathI'$ defined as follows.
In $\mathI$, agent $1$ assigns score $1$ to $c_1$, agent $2$ assigns score $1$ to $c_2$, and every other agent assigns score $1$ to $c_3$. By contrast, in $\mathI'$, agent $1$ assigns score $1$ to $c_1$ and every other agent assigns score $1$ to $c_2$. 
On $\mathI$, \egal{} outputs $\mathbf x$ with $x_1=x_2=x_3=\frac13$,  whereas on $\mathI'$, it outputs ${\mathbf x}'$ with $x'_1=x'_2=\frac12$.
Thus, independence is violated for $c_1$. \medskip

\emph{\med{}}: For \med{} we consider the following instances $\mathI^8$ and $\mathI^9$, where all agents assign score $0$ to all~$c_j$ with $j\geq 4$. For all $n\ge 3$ we have $\text{med}_{i\in N} s_{i,1}^8=\frac{1}{2}$ and  $\text{med}_{i\in N} s_{i,3}^8=\frac{1}{2}$; moreover,  
$\text{med}_{i\in N} s_{i,2}^8=\frac{1}{2}$ if $n$ is odd and $\text{med}_{i\in N} s_{i,2}^8=\frac{1}{4}$ if $n$ is even. This means that \med{} assigns a probability of less than $\frac{1}{2}$ to $c_1$ in $\mathI^8$. By contrast, in $\mathI^9$, the medians are $\text{med}_{i\in N} s_{i,1}^9=\frac{1}{2}$, $\text{med}_{i\in N} s_{i,2}^9=0$, and $\text{med}_{i\in N} s_{i,3}^9=\frac{1}{2}$. Hence, \med{} assigns a probability of $\frac{1}{2}$ to $c_1$ in this instance, and independence is violated. 
\begin{center}
    \begin{tabular}{ c | c c c}
          $\mathcal{I}^8$ & $s_{i,1}$ & $s_{i,2}$ & $s_{i,3}$\\ 
         \hline \hline
         $1$ & $\frac{1}{2}$ & $\frac{1}{2}$ & $0$\\
         $i\in \{2,\dots, \lceil\frac{n+1}{2}\rceil\}$ & $\frac{1}{2}$ & $0$ & $\frac{1}{2}$\\
         $i\in \{\lceil\frac{n+1}{2}\rceil+1,\dots, n\}$ & $0$ & $\frac{1}{2}$ & $\frac{1}{2}$ \\
        \end{tabular}   
    \qquad\qquad
    \begin{tabular}{ c | c c c}
          $\mathcal{I}^9$ & $s_{i,1}$ & $s_{i,2}$ & $s_{i,3}$\\ 
         \hline \hline
         $1$ & $\frac{1}{2}$ & $0$ & $\frac{1}{2}$ \\
         $i\in \{2,\dots, \lceil\frac{n+1}{2}\rceil\}$ & $\frac{1}{2}$ & $0$ & $\frac{1}{2}$\\
         $i\in \{\lceil\frac{n+1}{2}\rceil+1,\dots, n\}$ & $0$ & $\frac{1}{2}$ & $\frac{1}{2}$ \\
        \end{tabular}   
    \end{center}\medskip

    \noindent\textbf{Claim 2}: Theorem 3 of \citet{freeman2021truthfulbudget} shows that \im{} and \util{} are score-monotone. We first deal with coordinate-wise aggregation rules, and then consider \egal{}.\medskip
    
    \emph{Coordinate-wise aggregation rules}: 
    Recall that a function $f:{\mathbb R}^n\to \mathbb R$ is {\em monotone} if, for every pair of vectors ${\mathbf x} = (x_1, \dots, x_n)$ and ${\mathbf y} = (y_1, \dots, y_n)$ with $x_i\ge y_i$ for all $i\in [n]$, it holds that $f({\mathbf x})\ge f({\mathbf y})$. 
    Observe that for each of our coordinate-wise aggregation rules, it holds that
    their coordinate-aggregation functions are monotone. Our claim is now implied by the following lemma.
    \begin{lemma}
        Suppose that a coordinate-wise aggregation rule has the property that all of its coordinate-aggregation functions are monotone.
        Then this rule is score-monotone.
    \end{lemma}
    \begin{innerproof}
     Consider two instances $\mathcal{I}$ and $\mathcal{I'}$, an agent $i\in N$, and a candidate $c_j$ such that $\mathbf{s}_\ell=\mathbf{s}'_{\ell}$ for all $\ell\in N\setminus \{i\}$, $s_{i,j}> s_{i,j}'$, and $s_{i,k}\leq s_{i,k}'$ for all $c_{k}\in C\setminus \{c_j\}$. For ease of presentation, we will slightly abuse notation and write $f_k(\mathI)$ to mean $f_k(s_{1,k}, \dots, s_{n, k})$. We consider four cases.

    \medskip
    \underline{Case~1}:
    $f_k(\mathI)=f_k(\mathI')=0$ for all $k\in [m]$. Then $F(\mathI)_j=F(\mathI')_j=\frac{1}{m}$ by definition, and score-monotonicity holds.

    \medskip
    \underline{Case~2}:
    $f_k(\mathI)=0$ for all $k\in [m]$, but there is some $k^*$ with $f_{k^*}(\mathI')>0$. By definition of $F$, this means that 
    $F(\mathI')_k=0$ for all $k$ with $f_k(\mathI')=0$.
    Further, we have $s_{i,j}>s_{i,j}'$ and 
    $s_{\ell, j}=s'_{\ell, j}$ for $\ell\in N\setminus\{i\}$, so monotonicity of
    $f_j$ implies that $f_{j}(\mathI')\leq f_{j}(\mathI)=0$.
    Hence, as argued above, $F(\mathI')_j=0$.  
    On the other hand, we have  $F(\mathI)_j=\frac{1}{m}>0$, so score-monotonicity holds. 

    \medskip
    \underline{Case~3}:
    $f_k(\mathI')=0$ for all $k\in [m]$, but there is some $k^*$ with $f_{k^*}(\mathI)>0$. By monotonicity, we have $f_k(\mathI)\leq f_k(\mathI')$ for all $k\in [m]\setminus \{j\}$, so it must be the case that $j=k^*$, i.e., $f_j(\mathI)>0$ and 
    $f_k(\mathI)=0$ for all $k\in [m]\setminus\{j\}$. This implies that $F(\mathI)_j=1>\frac{1}{m}=F(\mathI')_j$, and score-monotonicity holds again. 

    \medskip
    \underline{Case~4}:
    Neither $f_k(\mathI)=0$ for all $k\in [m]$ nor $f_k(\mathI')=0$ for all $k\in [m]$. By monotonicity, $f_j(\mathI)\geq f_j(\mathI')$ and $f_k(\mathI)\leq f_k(\mathI')$ for all $k\in [m]\setminus \{j\}$.
    As the function $g(x) = \frac{x}{x+\lambda}$ is monotonically increasing for $x, \lambda>0$, it follows that
    \begin{align*}
    F(\mathI)_j=\frac{f_j(\mathI)}{\sum_{k\in [m]} f_k(\mathI)}\geq \frac{f_j(\mathI')}{f_j(\mathI')+\sum_{k\in [m]\setminus \{j\}} f_k(\mathI)}\geq \frac{f_j(\mathI')}{\sum_{k\in [m]} f_k(\mathI')}=F(\mathI')_j.
    \end{align*}

    \medskip
    This completes the proof of the lemma (and therefore our claim regarding coordinate-wise aggregation rules).
    \end{innerproof}

\medskip

    \emph{\egal{}}:
    When $n=2$, \egal{} is equivalent to \avg{} (by Proposition~\ref{prop:welfare-egal_is_sum}), and we have just shown that \avg{} satisfies score-monotonicity.
    Thus, we only consider the case $m=2$.  
    Let $\mathcal{I}$ and $\mathcal{I}'$ denote two instances such that $s_{i,1} < s'_{i,1}$ and $s_{i,2}  > s'_{i,2}$ for some agent $i\in N$, whereas all other agents have the same preferences in both instances.
    Let $\mathbf{x}$ and $\mathbf{x}'$ be the outcome vectors returned by \egal{} on $\mathcal{I}$ and $\mathcal{I}'$, respectively. Then,
    $$
    x_1 = \frac{\max_{\ell\in N}s_{\ell, 1}+\min_{\ell\in N}s_{\ell, 1}}{2} \le 
    \frac{\max_{\ell\in N}s'_{\ell, 1}+\min_{\ell\in N}s'_{\ell, 1}}{2} = x'_1,
    $$
    which means that score-monotonicity is satisfied.    

    Next, we prove that \egal{} fails score-monotonicity
    when $m\geq 4$ and $n\geq 4$. As usual, we will provide a counterexample for $m=4$ candidates; it can be extended to $m>4$ by adding candidates that receive score $0$ from all agents. Moreover, we focus on the case $n=4$; the example can be extended to $n>4$ by duplicating agents. 
  
    Consider the following two instances $\mathcal{I}^{10}$ and $\mathcal{I}^{11}$, and let $\mathbf{x}^{10}$ and $\mathbf{x}^{11}$ denote the outputs of \egal{} on $\mathcal{I}^{10}$ and $\mathcal{I}^{11}$, respectively. 
    \begin{center}
        \begin{tabular}{ c | c c c c c}
          $\mathcal{I}^{10}$ & $s_{i,1}$ & $s_{i,2}$ & $s_{i,3}$ & $s_{i,4}$  \\ 
         \hline \hline
         $1$ & $1$ & $0$ & $0$ & $0$ \\ 
         $2$ & $\frac{1}{2}$ & $\frac{1}{4}$ & $\frac{1}{4}$ &  $0$  \\ 
         $3$ & $0$ & $\frac{1}{2}$ & $0$ & $\frac{1}{2}$  \\ 
         $4$ & $0$ & $0$ & $\frac{1}{2}$ & $\frac{1}{2}$  \\ 
        \end{tabular}
        \qquad\qquad
        \begin{tabular}{ c | c c c c c}
          $\mathcal{I}^{11}$ & $s_{i,1}$ & $s_{i,2}$ & $s_{i,3}$ & $s_{i,4}$  \\ 
         \hline \hline
         $1$ & $\frac{1}{2}$ & $0$ & $0$ & $\frac{1}{2}$  \\ 
         $2$ & $\frac{1}{2}$ & $\frac{1}{4}$ & $\frac{1}{4}$ &  $0$   \\ 
         $3$ & $0$ & $\frac{1}{2}$ & $0$ & $\frac{1}{2}$   \\ 
         $4$ & $0$ & $0$ & $\frac{1}{2}$ & $\frac{1}{2}$   \\ 
        \end{tabular}
    \end{center}\smallskip
    We claim that 
    $$
    \mathbf{x}^{10}=\left(\frac12, 0, 0, \frac12\right), \qquad
    \mathbf{x}^{11}=\left(\frac15, \frac15, \frac15, \frac25\right).
    $$
    Indeed, in $\mathcal{I}^{10}$, all agents get disutility $1$ from $(\frac12, 0, 0, \frac12)$. Thus, to show that \egal{} returns this vector, it suffices to prove that any other outcome would increase the disutility of at least one agent. To this end, 
    consider an outcome vector ${\mathbf y} = (y_1, y_2, y_3, y_4)$. If $y_1<\frac12$, we have $d_1({\mathbf y})>1$. Similarly, if $y_4<\frac12$, we have $d_3({\mathbf y})>1$ or $d_4({\mathbf y})>1$. Thus, 
     $(\frac12, 0, 0, \frac12)$ is indeed the outcome returned by \egal{} on $\mathcal{I}^{10}$. 
     
     Similarly, in $\mathcal{I}^{11}$ all agents get disutility $\frac45$ from $(\frac15, \frac15, \frac15, \frac25)$.
     Consider an outcome vector ${\mathbf y} = (y_1, y_2, y_3, y_4)$. If $y_1+y_4<\frac35$ then $d_1({\mathbf y})>\frac45$. Similarly, 
     if $y_2+y_4<\frac35$ then $d_3({\mathbf y})>\frac45$, and if $y_3+y_4<\frac35$ then $d_4({\mathbf y})>\frac45$. Thus, for 
     $\max_{i\in N}d_i({\mathbf y})\le \frac45$ to hold, it has to be the case that
     \begin{equation}\label{eq:egal-sm}
     y_1+y_4\ge \frac35, \qquad 
     y_2+y_4\ge \frac35, \qquad
     y_3+y_4\ge \frac35;
     \end{equation}
     as $y_1+y_2+y_3+y_4=1$, this implies $y_4\ge \frac25$. Moreover, if $y_4>\frac25$, we have $d_2({\mathbf y})>\frac45$. It follows that \egal{} assigns probability $\frac25$ to $c_4$. Substituting $y_4=\frac25$ into each inequality in~\eqref{eq:egal-sm} and using the fact that $y_1+y_2+y_3=\frac35$, we conclude that 
     \egal{} indeed outputs $(\frac15, \frac15, \frac15, \frac25)$ on $\mathI^{11}$.
         
    It remains to observe that $x_4^{10}=\frac{1}{2}>\frac{2}{5}=x_4^{11}$ despite the fact that $s_{1,4}^{10}<s_{1,4}^{11}$, so score-monotonicity is violated.\medskip

    \noindent\textbf{Claim 3}: We next show that among all the rules considered in this paper, only \med{} fails reinforcement. For this, we first note that \citet[Thms.~9 and 13]{freeman2021truthfulbudget} have shown that \im{} and \util{} satisfy reinforcement. We hence focus on the remaining rules and consider three instances $\mathI=(\mathbf{s}_1,\dots, \mathbf{s}_n)$ (defined for electorate $N$), $\mathI'=(\mathbf{s}_1',\dots, \mathbf{s}_{n'}')$ (defined for electorate $N'$), and $\mathI''=(\mathbf{s}_1,\dots, \mathbf{s}_n, \mathbf{s}_1',\dots, \mathbf{s}_{n'}')$ (defined for electorate $N\cup N'$; we assume $N\cap N'=\emptyset$). 
    
    We consider each rule separately.\medskip

     \emph{\avg}: Let $\mathbf{x}$, $\mathbf{x}'$, and $\mathbf{x}''$ be the outcomes chosen by \avg{} for $\mathI$, $\mathI'$, and $\mathI''$, respectively, and suppose that $\mathbf{x}=\mathbf{x}'$.
      By Proposition~\ref{prop:nonorm}, 
     for each $j\in [m]$, we have 
     $x_j = \frac{1}{|N|}\sum_{i\in N} s_{i,j}=\frac{1}{|N'|}\sum_{i\in N'} s_{i,j}' = x'_j$. Consequently, 
     \begin{align*}
         x_j''=\frac{1}{|N\cup N'|} \sum_{i\in N\cup N'} s_{i,j}''=\frac{|N|}{|N\cup N'|}\cdot \frac{1}{|N|} \sum_{i\in N} s_{i,j} + \frac{|N'|}{|N\cup N'|}\cdot \frac{1}{|N'|} \sum_{i\in N'} s_{i,j}'=x_j,
     \end{align*}
     which shows that reinforcement is satisfied.\medskip

     \emph{\maxrule{}}: 
     Let $\mathbf{x}$, $\mathbf{x}'$, and $\mathbf{x}''$ be the outcomes chosen by \maxrule{} for $\mathI$, $\mathI'$, and $\mathI''$, respectively, and suppose that $\mathbf{x}=\mathbf{x}'$.
     Let $\alpha=\frac{\sum_{j\in [m]}\max_{i\in N} s_{i, j}}{\sum_{j\in [m]}\max_{i\in N'} s'_{i, j}}$.
     Then, $\max_{i\in N} s_{i,j}=\alpha \cdot \max_{i\in N'} s_{i,j}'$ for all $j\in [m]$. 
     We assume without loss of generality that $\alpha\geq 1$, as otherwise we can exchange the roles of $\mathbf{x}$ and $\mathbf{x'}$. Consequently, $\max_{i\in N\cup N'} s_{i,j}''=\max_{i\in N} s_{i,j}$ for all $j\in [m]$, which shows that $\mathbf{x}''=\mathbf{x}$.\medskip

     \emph{\minrule{}}: 
     Let $\mathbf{x}$, $\mathbf{x}'$, and $\mathbf{x}''$ be the outcomes chosen by \minrule{} for $\mathI$, $\mathI'$, and $\mathI''$, respectively, and suppose that $\mathbf{x}=\mathbf{x}'$.
     First, if $\min_{i\in N} s_{i,j}=0$ for all $j\in [m]$, 
     then $\min_{i\in N\cup N'} s_{i,j}''=0$ for all $j\in [m]$ and $x_j''=\frac{1}{m}=x_j=x_j'$ for all $j\in [m]$. 
     A similar argument applies if 
     $\min_{i\in N'} s_{i,j}'=0$ for all $j\in [m]$.

     Next, assume that $\min_{i\in N} s_{i,\ell}>0$ and $\min_{i\in N'} s_{i,\ell'}'>0$ for some $\ell,\ell'\in [m]$. For this case, the analysis is similar to that for \maxrule{}. Namely, 
     let $\alpha=\frac{\sum_{j\in [m]}\min_{i\in N} s_{i, j}}{\sum_{j\in [m]}\min_{i\in N'} s'_{i, j}}$.
     Then, $\mathbf{x}=\mathbf{x}'$ implies that $\min_{i\in N} s_{i,j}=\alpha \cdot \min_{i\in N'} s_{i,j}'$ for all $j\in[m]$. We assume without loss of generality that $\alpha\leq 1$, as otherwise we can exchange the roles of $\mathbf{x}$ and $\mathbf{x}'$ in our argument.
     Consequently, $\min_{i\in N\cup N'} s_{i,j}''=\min_{i\in N} s_{i,j}$ for all $j\in [m]$, which shows that $\mathbf{x}''=\mathbf{x}$.\medskip

     \emph{\geo{}}: 
     Let $\mathbf{x}$, $\mathbf{x}'$, and $\mathbf{x}''$ be the outcomes chosen by \geo{} for $\mathI$, $\mathI'$, and $\mathI''$, respectively, and suppose that $\mathbf{x}=\mathbf{x}'$.
     If $\min_{i\in N} s_{i,j}=0$ for all $j\in [m]$ or $\min_{i\in N'} s_{i,j}'=0$ for all $j\in [m]$, then
     $\min_{i\in N\cup N'} s_{i,j}''=0$ for all $j\in [m]$ and hence
     $x''_j=\frac{1}{m}=x_j=x_j'$ for all $j\in [m]$. 

     On the other hand, suppose that there are indices $\ell, \ell'\in [m]$ such that $\min_{i\in N} s_{i,\ell}>0$ and $\min_{i\in N} s'_{i,\ell'}>0$. Define $\xi(\mathI)=\sum_{j\in [m]} (\prod_{i\in N} s_{i,j})^{1/|N|}$ and $\xi(\mathI')=\sum_{j\in [m]} (\prod_{i\in N'} s_{i,j}')^{1/|N'|}$. For each $j\in [m]$, we have 
     \begin{align*}
         \left(\prod_{i\in N\cup N'} s_{i,j}''\right)^{\frac{1}{|N|+|N'|}}&=\left(\prod_{i\in N} s_{i,j}\cdot \prod_{i\in N'}s_{i,j}' \right)^{\frac{1}{|N|+|N'|}} \\
         &= \left((x_j\cdot \xi(\mathI))^{|N|}\cdot (x_j'\cdot \xi(\mathI'))^{|N'|}\right)^{\frac{1}{|N|+|N'|}}\\
         &=x_j\cdot \left(\xi(\mathI)^{|N|}\cdot \xi(\mathI')^{|N'|}\right)
         ^{\frac{1}{|N|+|N'|}}.
     \end{align*}
     Since this holds for all $j\in [m]$, \geo{} satisfies reinforcement, because
     \begin{align*}
         x_j''=\frac{x_j\cdot \left(\xi(\mathI)^{|N|}\cdot \xi(\mathI')^{|N'|}\right)^{\frac{1}{|N|+|N'|}}
         }{\sum_{k\in [m]} \left(x_{k}\cdot\left(\xi(\mathI)^{|N|}\cdot \xi(\mathI')^{|N'|}\right)
         ^{\frac{1}{|N|+|N'|}}\right)}=x_j.
     \end{align*}

     \emph{\egal{}}: Let $\mathbf{x}$, $\mathbf{x}'$, and $\mathbf{x}''$ denote the outcomes chosen by \egal{} for $\mathI$, $\mathI'$, and $\mathI''$, respectively. 
     Assume for contradiction that $\mathbf{x}=\mathbf{x'}$ but $\mathbf{x}''\neq \mathbf{x}$. 
     For a set of agents $T$, we define $v_T(\mathbf{y})$ as the vector that lists the disutilities of all agents in~$T$ for the outcome $\mathbf{y}$ in non-increasing order. 
     Since $\mathbf{x}=\mathbf{x'}$ is chosen for $\mathI$ and $\mathI'$, we get \emph{(i)} $v_N(\mathbf{x})<_\mathit{lex} v_N(\mathbf{x''})$ or 
     $v_N(\mathbf{x})= v_N(\mathbf{x''})$, and 
     \emph{(ii)} $v_{N'}(\mathbf{x})<_\mathit{lex} v_{N'}(\mathbf{x''})$ or $v_{N'}(\mathbf{x})= v_{N'}(\mathbf{x''})$.
     On the other hand, $v_{N\cup N'}(\mathbf{x})>_\mathit{lex} v_{N\cup N'}(\mathbf{x''})$ or $v_{N\cup N'}(\mathbf{x})= v_{N\cup N'}(\mathbf{x''})$, as $\mathbf{x}''$ is chosen for $\mathI''$.
     It is easy to see that these conditions can only be satisfied if $v_N(\mathbf{x})= v_N(\mathbf{x''})$ and $v_{N\cup N'}(\mathbf{x})= v_{N\cup N'}(\mathbf{x''})$.
     However, consistent tie-breaking requires that if \egal{} returns $\mathbf{x}$ for $\mathI$, then it does not return $\mathbf{x}''$ for $\mathI''$, which contradicts our assumptions. \medskip

     \emph{\med{}}: Finally, we show that \med{} satisfies reinforcement if $m=2$ but fails this property when $m\geq 3$. 
     First, assume that $m=2$, and let $\mathbf{x}$, $\mathbf{x}'$, and $\mathbf{x}''$ denote the outcomes chosen by \med{} on $\mathI$, $\mathI'$, and $\mathI''$, respectively. 
     As usual, we assume that $\mathbf{x}=\mathbf{x}'$. 
     Since $x''_2=1- x''_1$ and $x_2=1-x_1$, it suffices to show that $x''_1=x_1$.
     By Proposition~\ref{prop:nonorm}, we have 
     $\text{med}_{i\in N} s_{i,1}=x_1=x_1'=\text{med}_{i\in N'} s_{i,1}'$. Hence, 
     our claim for $m=2$ is implied by the following lemma, whose proof is relegated to \Cref{app:omitted-proofs}.
     \begin{restatable}{lemma}{medians}
         Consider two multisets of real numbers $A$ and $B$ such that $\text{\em med} (A)=\text{\em med}(B)=z$. Then $\text{\em med}(A\cup B) = z$ as well.
     \end{restatable}

    Now, for the case $m\geq 3$, consider the following instances $\mathI^{12}$ (defined for an odd number of agents $n\geq 5$) and $\mathI^{13}$ (defined for a single agent $n+1$). 
    
\begin{center}
    \begin{tabular}{ c | c c c}
          $\mathcal{I}^{12}$ & $s_{i,1}$ & $s_{i,2}$ & $s_{i,3}$\\ 
         \hline \hline
         $i\in \{1,\dots, \frac{n-1}{2}\}$ & $\frac{2}{3}$ & $\frac{1}{3}$ & $0$ \\ 
         $i\in \{\frac{n+1}{2},\dots, n-1\}$ & $0$ & $\frac{1}{3}$ & $\frac{2}{3}$ \\
         $n$ & $\frac{1}{2}$ & $0$  & $\frac{1}{2}$ 
        \end{tabular}   
    \qquad\qquad
    \begin{tabular}{ c | c c c}
          $\mathcal{I}^{13}$ & $s_{i,1}$ & $s_{i,2}$ & $s_{i,3}$\\ 
         \hline \hline
         $n+1$ & $\frac{3}{8}$ & $\frac{2}{8}$ & $\frac{3}{8}$
        \end{tabular}   
    \end{center}
All candidates other than $c_1$, $c_2$, and $c_3$ receive score~$0$ from all agents and can be ignored. It holds that $\text{med}_{i\in N} s_{i,1}^{12}=\frac{1}{2}$, $\text{med}_{i\in N} s_{i,2}^{12}=\frac{1}{3}$, and $\text{med}_{i\in N} s_{i,3}^{12}=\frac{1}{2}$, so \med{} returns the outcome $\mathbf{x}=(\frac{3}{8},\frac{2}{8}, \frac{3}{8})$ for $\mathI^{12}$. 
Moreover, $\mathI^{13}$ consists of a single agent with preference $(\frac{3}{8},\frac{2}{8}, \frac{3}{8})$, so \med{} returns $\mathbf x$ on $\mathI^{13}$ as well. 
However, for the combined instance, the medians are $\text{med}_{i\in N\cup \{n+1\}} s_{i,1}=\frac{1}{2}\cdot(\frac{1}{2}+\frac{3}{8})=\frac{7}{16}$, $\text{med}_{i\in N\cup \{n+1\}} s_{i,2}=\frac{1}{3}$, and $\text{med}_{i\in N\cup \{n+1\}} s_{i,3}=\frac{1}{2}\cdot(\frac{1}{2}+\frac{3}{8})=\frac{7}{16}$.
This means that for the new outcome $\mathbf{x}'$ we have $x_1'=\frac{21}{58}\neq \frac{3}{8}$, and reinforcement is violated. 
\end{proof}

\begin{remark}
    We do not specify the boundary on $n$ for which \med{} satisfies reinforcement. The reason for this is that reinforcement is a variable-electorate property that is usually studied under the assumption that there is an infinite set of possible agents. The same reasoning applies to participation (Section~\ref{sec:incent}). 
\end{remark}

\section{Incentive Properties}\label{sec:incent}

The last group of evaluation criteria that we consider for our rules is their incentive properties. 
For this, \cite{freeman2021truthfulbudget} have shown that both \im{} and \util{} satisfy strategyproofness and participation. 
By contrast, we show that all other rules we consider are not strategyproof, but on the positive side, all except \med{} satisfy participation.

\begin{theorem}
\label{thm:strategyproofness}
    The following claims hold.
    \begin{enumerate}[label=(\arabic*),topsep=4pt,itemsep=0pt]
        \item \im{} and \util{} satisfy strategyproofness and participation.
        \item \avg{}, \maxrule{}, \minrule{}, \geo{}, and \egal{} fail strategyproofness for all $n \geq 2$ and $m \geq 2$ but satisfy participation. 
        \item \med{} satisfies strategyproofness when $m=2$ and $n$ is odd, but fails this property when $m=2$ and $n$ is even, or when $m\geq 3$. Moreover, \med{} satisfies participation if $m=2$ but fails it for all $m\geq 3$.
    \end{enumerate}
\end{theorem}
\begin{proof}
    Claim 1 follows directly from the work of \citet[Thms.~2, 8, and 12]{freeman2021truthfulbudget}, so we only prove Claims 2 and~3 here.\medskip

    \noindent\textbf{Claim 2}: We first consider participation, and then show that all considered rules fail strategyproofness by giving a common counterexample.\medskip

    \emph{Participation}: 
     We consider two instances $\mathcal{I}$ and $\mathcal{I}'$ such that $\mathcal{I}'$ is derived from $\mathcal{I}$ by adding an agent~$i$ with preference $\mathbf{s}_i$. Given an aggregation rule $F$, let $\mathbf{x}$ and $\mathbf{x}'$ denote the outcomes chosen by $F$ on $\mathI$ and $\mathI'$, respectively.
    To show that \avg{}, \minrule{}, \maxrule{}, and \geo{} satisfy participation, we first prove an auxiliary lemma; the proof for \egal{} will not rely on this lemma. 
    \begin{lemma}\label{lem:part}
    Define $X^+=\{j\in [m]\colon x_j'>x_j\}$ and $X^-=\{j\in [m]\colon x_j'<x_j\}$. 
    For an agent $i$ such that (a) $s_{i,j}\geq x_j'$ for all $j\in X^+$ or (b) $s_{i,j}\leq x_j'$ for all $j\in X^-$, it holds that~$i$ weakly prefers $\mathbf{x}'$ to $\mathbf{x}$.   
    \end{lemma}
    \begin{innerproof}
    We will prove the lemma only for case~(a), as case~(b) is symmetric. In case~(a), for all $j\in X^+$ it holds that 
    \begin{align*}
        |s_{i,j}-x_{j}'|-|s_{i,j}-x_j|=(s_{i,j}-x_{j}')-(s_{i,j}-x_j)=x_j-x_j'.
    \end{align*}
    On the other hand, we observe that
    \begin{align*}
        \sum_{j\in [m]\setminus X^+} \left(|s_{i,j}-x_j'|-|s_{i,j}-x_j|\right)&\leq \sum_{j\in [m]\setminus X^+} |x_j-x_j'|=\sum_{j\in [m]\setminus X^+}(x_j-x_j'),
    \end{align*}
    where the first transition follows from the triangle inequality and the second transition uses the definition of $X^+$. 
    Based on these observations, we conclude
    that $i$ weakly prefers $\mathbf{x}'$ to $\mathbf{x}$,
     since 
    \begin{align*}
        d_i(\mathbf{x}')-d_i(\mathbf{x})&=\sum_{j\in [m]} |s_{i,j}-x_{j}'|-\sum_{j\in [m]} |s_{i,j}-x_{j}|\\
        &= \sum_{j\in X^+} \left(|s_{i,j}-x_{j}'|- |s_{i,j}-x_{j}|\right)+\sum_{j\in [m]\setminus X^+} \left(|s_{i,j}-x_{j}'|- |s_{i,j}-x_{j}|\right)\\
        &\leq \sum_{j\in X^+} (x_j-x_j') + \sum_{j\in [m]\setminus X^+} (x_j-x'_j)\\
        &=\sum_{j\in [m]}x_j-\sum_{j\in [m]}x'_j\\
        &=0.
    \end{align*}
This completes the proof of the lemma.
\end{innerproof}

    Thus, to show that our rules satisfy participation, it  suffices to prove that they satisfy $s_{i,j}\geq x_j'$ for all $j\in X^+$ or $s_{i,j}\leq x_j'$ for all $j\in X^-$. For this, we consider each rule individually.\medskip 
    
    \emph{\avg}: For \avg{}, for each $j\in X^+$ we have
    \begin{align*}
        x_j'=\frac{1}{n+1}\sum_{\ell\in N\cup \{i\}} s_{\ell,j}= \frac{s_{i, j}}{n+1}+
        \frac{n}{n+1}\cdot \frac{1}{n}\sum_{\ell\in N} s_{\ell, j} = 
        \frac{1}{n+1}(s_{i,j}+n x_j) \le
        \frac{1}{n+1}(s_{i,j}+n x'_j).
    \end{align*}
    Multiplying both sides by $n+1$ and rearranging the terms, we conclude that $s_{i, j}\ge x'_j$ for all $j\in X^+$, so condition~(a) of Lemma~\ref{lem:part} is satisfied.  
    
  \medskip

    \emph{\maxrule{}}: 
    Note that $\max_{\ell\in N\cup\{i\}} s_{\ell,j} = \max\{s_{i,j}, \max_{\ell\in N} s_{\ell,j}\}$ for all $j\in [m]$. 
    We claim that $\max_{\ell\in N\cup\{i\}} s_{\ell,j} = s_{i,j}$
    for all $j\in X^+$, because otherwise $\max_{\ell\in N\cup \{i\}} s_{\ell,j}=\max_{\ell\in N} s_{\ell,j}$ and hence 
        \begin{align*}
        x_j'=\frac{\max_{\ell\in N\cup \{i\}} s_{\ell,j}}{\sum_{k\in [m]} \max_{\ell\in N\cup\{i\}} s_{\ell, k}}\leq \frac{\max_{\ell\in N} s_{\ell,j}}{\sum_{k\in [m]} \max_{\ell\in N} s_{\ell,k}} =x_j.
        \end{align*}
    Further, since $\sum_{k\in [m]} \max_{\ell\in N\cup\{i\}} s_{\ell,k}\geq \sum_{k\in [m]} s_{i,k} = 1$, we have $s_{i,j}\geq \frac{s_{i,j}}{\sum_{k\in [m]} \max_{\ell\in N\cup\{i\}} s_{\ell,k}}=x'_j$ for all $j\in X^+$. By Lemma~\ref{lem:part}, this proves that \maxrule{} satisfies participation.\medskip

    \emph{\minrule{}}: For \minrule{}, we first consider the case where $\min_{\ell\in N\cup \{i\}} s_{\ell,j}=0$ for all $j\in [m]$. 
    Define 
    $$
    Z=\{j\in [m]\colon \min_{\ell\in N} s_{\ell,j}>0\},\qquad 
    Z'=\{j\in [m]\colon s_{i,j}>0\};
    $$
    note that $Z'\neq\emptyset$ because $\sum_{j\in [m]}s_{i, j}=1$.
    Now, if $Z=\emptyset$, then we have $\min_{\ell\in N} s_{\ell,j}=0$ for all $j\in [m]$, and $\mathbf{x}=\mathbf{x}'=(\frac{1}{m},\dots, \frac{1}{m})$, so participation is satisfied in this case. On the other hand, suppose that $Z\neq\emptyset$, and
    observe that $Z\cap Z'=\emptyset$ because for every $j\in Z\cap Z'$ we would have
    $\min_{\ell\in N\cup \{i\}} s_{\ell,j}>0$. Further, $x_j>0$ if and only if $j\in Z$, so $\sum_{j\in Z} x_j=1$; similarly, $\sum_{j\in Z'} s_{i, j} = 1$.
    Thus, we can write
    $$
    d_i({\mathbf x}) = \sum_{j\in Z} |s_{i,j}-x_j|+ \sum_{j\in Z'} |s_{i,j}-x_j|+
    \sum_{j\in [m]\setminus(Z\cup Z')} |s_{i,j}-x_j| 
    \ge \sum_{j\in Z} x_j +
    \sum_{j\in Z'} s_{i, j} = 2. 
    $$
    As we have $d_i({\mathbf y})\le 2$ for all possible outcomes $\mathbf y$, it follows that
    $d_i({\mathbf x}')\le d_i({\mathbf x})$, so participation is satisfied again.

    Next, consider the case where there is a candidate $c_{j^*}$ such that $\min_{\ell\in N\cup \{i\}} s_{\ell,j^*}>0$.
    We claim that $s_{i,j}\leq x_j'$ for all $j\in X^-$. The argument is similar to that for \maxrule{}.
    Specifically, we have $\min_{\ell\in N\cup\{i\}} s_{\ell,j} = \min\{s_{i,j}, \min_{\ell\in N} s_{\ell,j}\}$ for all $j\in [m]$. 
    Observe that $\min_{\ell\in N\cup\{i\}} s_{\ell,j} = s_{i,j}$
    for all $j\in X^-$, because otherwise $\min_{\ell\in N\cup \{i\}} s_{\ell,j}=\min_{\ell\in N} s_{\ell,j}$ and hence 
        \begin{align*}
        x_j'=\frac{\min_{\ell\in N\cup \{i\}} s_{\ell,j}}{\sum_{k\in [m]} \min_{\ell\in N\cup\{i\}} s_{\ell, k}}\geq \frac{\min_{\ell\in N} s_{\ell,j}}{\sum_{k\in [m]} \min_{\ell\in N} s_{\ell,k}} =x_j.
        \end{align*}
    Further, since $\sum_{k\in [m]} \min_{\ell\in N\cup\{i\}} s_{\ell,k}\leq \sum_{k\in [m]} s_{i,k} = 1$, we have $s_{i,j}\leq \frac{s_{i,j}}{\sum_{k\in [m]} \min_{\ell\in N\cup\{i\}} s_{\ell,k}}=x'_j$ for all $j\in X^-$. By Lemma~\ref{lem:part}, this proves that \minrule{} satisfies participation.

\medskip

    \emph{\geo{}}: To show that \geo{} satisfies participation if $\prod_{\ell\in N\cup \{i\}} s_{\ell,j}=0$ for all $j\in [m]$, we use the same argument as for \minrule{}.
    Hence, from now on we assume that there is a candidate $c_{j^*}$ such that $\sqrt[n+1]{\prod_{\ell\in N\cup \{i\}} s_{\ell,j^*}}>0$. Let $\xi(\mathI)=\sum_{j\in [m]} \sqrt[n]{\prod_{\ell\in N} s_{\ell,j}}$ and $\xi(\mathI')=\sum_{j\in [m]} \sqrt[n+1]{\prod_{\ell\in N\cup\{i\}} s_{\ell,j}}$, and note that $\xi(\mathI')>0$. By the definition of $x_j$, for each $j\in[m]$ we have 
    \begin{align*}
        \prod_{\ell\in N} s_{\ell,j}=(x_j\cdot \xi(\mathI))^n.
    \end{align*}
    Consequently, for all $j\in [m]$ it holds that
    \begin{align*}
        x_j'=\frac{(s_{i,j})^{1/(n+1)}\cdot (x_j\cdot \xi(\mathI))^{n/(n+1)}}{\xi(\mathI')}=(s_{i,j})^{1/(n+1)}\cdot (x_j)^{n/(n+1)}\cdot \frac{\xi(\mathI)^{n/(n+1)}}{\xi(\mathI')}.
    \end{align*}
    We consider two cases.
    
    \medskip
    \underline{Case 1}: $\xi(\mathI)^{n/(n+1)}\geq \xi(\mathI')$. In this case, for each $j\in X^-$ we have 
    \begin{align*}
        x_j'\geq (s_{i,j})^{1/(n+1)}\cdot (x_j)^{n/(n+1)}\ge 
        (s_{i,j})^{1/(n+1)}\cdot (x'_j)^{n/(n+1)}. 
    \end{align*}
    By raising both sides to the power of $n+1$ and rearranging the terms, we obtain $x_j'\ge s_{i, j}$.
    
    \medskip
    \underline{Case 2}:
    $\xi(\mathI)^{n/(n+1)}< \xi(\mathI')$. In this case,  for all $j\in X^+$ we have 
        \begin{align*}
    x_j'\leq  (s_{i,j})^{1/(n+1)}\cdot (x_j)^{n/(n+1)}\le(s_{i, j})^{1/(n+1)}\cdot (x_j')^{n/(n+1)}.
        \end{align*}
        By raising both sides to the power of $n+1$ and rearranging the terms, we obtain $x_j'\le s_{i, j}$.

    \medskip
    We can thus conclude that \geo{} satisfies participation based on Lemma~\ref{lem:part}.\medskip

    \emph{\egal{}}: To show that \egal{} satisfies participation, we will provide a direct argument instead of relying on Lemma~\ref{lem:part}. 
    Assume for the sake of contradiction that $d_i(\mathbf{x})<d_i(\mathbf{x'})$. 
    For a set of agents~$T$, let $v_T(\mathbf{y})$ be a vector that contains the disutilities of all agents in $T$ for a given score vector $\mathbf{y}$ in non-increasing order.
    Since \egal{} chooses $\mathbf{x}$ for $\mathcal{I}$, it holds that either $v_N(\mathbf{x})=v_N(\mathbf{x'})$ or $v_N(\mathbf{x'})>_\mathit{lex} v_N(\mathbf{x})$. In either case, 
    $d_i(\mathbf{x})<d_i(\mathbf{x'})$ 
    would imply $v_{N\cup \{i\}}(\mathbf{x'})>_\mathit{lex} v_{N\cup \{i\}}(\mathbf{x})$, a contradiction with \egal{} choosing $\mathbf{x}'$ for $\mathcal{I}'$. Hence, it must be the case that $d_i(\mathbf{x})\ge d_i(\mathbf{x'})$, which means that \egal{} satisfies participation.\medskip
    
    \emph{Strategyproofness}: To prove the claim regarding strategyproofness, we consider the following instances $\mathcal{I}^{14}$ and $\mathcal{I}^{15}$. All candidates $c_j$ with $j\geq 3$ receive score $0$ from all agents and are ignored for the rest of the proof.  
 \begin{center}
    \begin{tabular}{ c | c c  }
          $\mathcal{I}^{14}$ & $s_{i,1}$ & $s_{i,2}$ \\ 
         \hline \hline
         $1$ & $\frac{4}{5}$ & $\frac{1}{5}$ \\ 
         $i\in \{2,\dots, n\}$ & $\frac{1}{5}$ & $\frac{4}{5}$ 
        \end{tabular}
        \qquad\qquad
    \begin{tabular}{ c | c c c c }
          $\mathcal{I}^{15}$ & $s_{i,1}$ & $s_{i,2}$  \\ 
         \hline \hline
         $1$ & $1$ & $0$ \\ 
         $i\in \{2,\dots, n\}$ & $\frac{1}{5}$ & $\frac{4}{5}$  
        \end{tabular}       
    \end{center}\smallskip 
    Here, \avg{} chooses the vector $\mathbf{x}=(\frac{1}{5}+\frac{3}{5n}, \frac{4}{5}-\frac{3}{5n})$ for $\mathI^{14}$ and $\mathbf{x}'=(\frac{1}{5}+\frac{4}{5n}, \frac{4}{5}-\frac{4}{5n})$ for $\mathI^{15}$, \maxrule{} respectively chooses $\mathbf{x}=(\frac{1}{2},\frac{1}{2})$ and $\mathbf{x}'=(\frac{5}{9},\frac{4}{9})$, \minrule{} chooses $\mathbf{x}=(\frac{1}{2},\frac{1}{2})$ and $\mathbf{x}'=(1,0)$, \geo{} chooses $\mathbf{x}=(\frac{\sqrt[n]{16}}{4+\sqrt[n]{16}}, \frac{4}{4+\sqrt[n]{16}})$ and $\mathbf{x}'=(1, 0)$, and \egal{} chooses $\mathbf{x}=(\frac{1}{2},\frac{1}{2})$ and $\mathbf{x}'=(\frac{3}{5}, \frac{2}{5})$. In each case, it can be checked that agent~$1$ benefits by deviating from $\mathI^{14}$ to $\mathI^{15}$, so strategyproofness is violated.\medskip

    \noindent\textbf{Claim 3}: Finally, we turn to \med{} and again consider participation and strategyproofness separately.\medskip
    
    \emph{Participation}: First, we show that \med{} satisfies participation when $m=2$. Consider two instances $\mathI$ and $\mathI'$ with $m=2$ such that $\mathI'$ is derived from $\mathI$ by adding an agent $i$ with preference $\mathbf{s}_i$. Let $\mathbf{x}$ and $\mathbf{x'}$ denote the outputs of \med{} on $\mathI$ and $\mathI'$, respectively.
    If ${\mathbf x}={\mathbf x}'$, our claim is trivially true, so assume without loss of generality that $x_1'>x_1$. We will show that this implies $s_{i,1}\geq x_1'$. 
    Recall that by Proposition~\ref{prop:nonorm},
    we have 
    $x_1=\text{med}_{\ell\in N} s_{\ell,1}$ and $x'_1=\text{med}_{\ell\in N\cup \{i\}} s_{\ell,1}$.
    Assume for contradiction that $s_{i,1}<x_{1}'$. Since $x_1'>x_1$ and there are at least $\lceil\frac{n}{2}\rceil$ agents $\ell\in N$ with $x_1\geq s_{\ell,1}$, there are at least $\lceil\frac{n}{2}\rceil+1$ agents $\ell$ in $N\cup\{i\}$ with $x_1'>s_{\ell,1}$. This means that $\text{med}_{\ell\in N\cup \{i\}} s_{\ell,1}<x_1'$, 
    a contradiction.
    Thus, we conclude that $s_{i, 1}\ge x'_1>x_1$. Moreover, $x_1+x_2=x'_1+x'_2=1$ implies $x'_2<x_2$. 
    Therefore, we can now use 
    Lemma~\ref{lem:part} with $X^+=\{1\}$ and $X^-=\{2\}$ to conclude that \med{} satisfies participation when $m=2$. 

    To see that \med{} fails participation when $m\geq 3$, we consider the instances $\mathI^{12}$ and $\mathI^{13}$ in the proof of \Cref{thm:consistencyprops} that were used to show that \med{} fails reinforcement. In this example, $\mathI^{13}$ consists of a single agent whose ideal distribution coincides with the outcome of \med{} for $\mathI^{12}$. However, in the combined instance, the outcome changes, which means that participation is violated for this agent.\medskip
    
    \emph{Strategyproofness}: We now focus on strategyproofness. First, we show that if $m=2$ and $n$ is odd, \med{} is strategyproof. To this end, we note that \med{} simply returns the distribution of the agent who assigns the $(\frac{n+1}{2})$-th highest score to $c_1$. Further, for all
    $\ell\in N$, we have
    $s_{\ell, 1}+s_{\ell, 2}=x_1+x_2=1$ and hence $d_\ell(\mathbf{x})=|s_{\ell,1}-x_1|+|s_{\ell,2}-x_2|=
    |s_{\ell,1}-x_1|+|(1-s_{\ell,1})-(1-x_1)|=2|s_{\ell,1}-x_1|$, so strategyproofness follows directly from the well-known result of \citet{Moul80a}. On the other hand, if $m=2$ and $n$ is even, \med{} returns the average between the scores of two agents, which allows agents to benefit by misreporting. For instance, if one agent reports $(\frac{2}{3}, \frac{1}{3})$, $\frac{n}{2}-1$ agents report $(1,0)$, and $\frac{n}{2}$ agents report $(0,1)$, \med{} chooses the outcome $(\frac{1}{3}, \frac{2}{3})$. However, if the first agent reports $(1,0)$, \med{} returns $(\frac{1}{2}, \frac{1}{2})$, which constitutes a beneficial manipulation for this agent. This example can be generalized to $m\geq 3$ by adding candidates that receive a score of $0$ from every agent.
    
    It remains to address the case $m\geq 3$ and odd $n\geq 3$. 
    In this case, consider the instance $\mathI^{16}$ below. One can check that \med{} chooses $\mathbf{x}=(\frac{1}{3}, \frac{1}{3}, \frac{1}{3})$ on $\mathI^{16}$. This means that $d_n(\mathbf{x})=(\frac{2}{5}-\frac{1}{3})+(\frac{3}{5}-\frac{1}{3})+\frac{1}{3}=\frac{2}{3}$. Next, suppose that agent $n$ reports the distribution $(\frac{8}{15}, \frac{7}{15}, 0)$. It can be verified that \med{} now picks the distribution $(\frac{2}{5}, \frac{3}{10}, \frac{3}{10})$, which means that $d_n(\mathbf{x}')=\frac{3}{5}<\frac{2}{3}$. Hence, \med{} is indeed manipulable if $m\geq 3$ and $n\geq 3$ is odd. 
\begin{center}
    \begin{tabular}{ c | c c c c }
          $\mathcal{I}^{16}$ & $s_{i,1}$ & $s_{i,2}$ & $s_{i,3}$ \\ 
         \hline \hline
         $i\in \{1,\dots, \frac{n-1}{2}\}$ & $0$ & $\frac{2}{5}$ & $\frac{3}{5}$ \\
         $i\in \{\frac{n+1}{2},\dots, n-1\}$ & $\frac{3}{5}$ & $0$ & $\frac{2}{5}$ \\
         $n$ & $\frac{2}{5}$ & $\frac{3}{5}$ & $0$ \\
    \end{tabular}
    $\qquad\qquad$
\begin{tabular}{ c | c c c c }
          $\mathcal{I}^{17}$ & $s_{i,1}$ & $s_{i,2}$ & $s_{i,3}$ \\ 
         \hline \hline
         $i\in \{1,\dots, \frac{n-1}{2}\}$ & $0$ & $\frac{2}{5}$ & $\frac{3}{5}$ \\
         $i\in \{\frac{n+1}{2},\dots, n-1\}$ & $\frac{3}{5}$ & $0$ & $\frac{2}{5}$ \\
         $n$ & $\frac{8}{15}$ & $\frac{7}{15}$ & $0$ \\
    \end{tabular}
    \end{center}
\medskip This completes the proof. 
\end{proof}

\section{Characterizations of the Average Rule}
\label{sec:characterization}

Thus far, we have conducted an extensive analysis of specific natural rules with respect to a number of desirable axioms.
In this section, we complement those findings by presenting two characterizations of the \avg{} rule. These results highlight the strong appeal of \avg{} by showing that it is the only rule within large classes of rules that, for example, satisfies score-unanimity and independence. 

To state these results, we will introduce two further axioms. 
First, we say that an aggregation rule~$F$ is \emph{anonymous} if the identities of the agents do not matter, i.e., $F(\mathI)=F(\pi(\mathI))$ for all instances $\mathI$ and permutations $\pi:N\rightarrow N$, where $\mathI'=\pi(\mathI)$ is the instance defined by $\mathbf{s}_{\pi(i)}'=\mathbf{s}_i$ for all $i\in N$. Second, we say that an aggregation rule $F$ is \emph{continuous} if it is a continuous function, i.e., 
for every $n, m\in{\mathbb N}$ and every sequence of $n$-agent, $m$-candidate instances 
$\mathI^1, \mathI^2, \dots$  such that $\lim_{t\rightarrow\infty} \mathI^t$ exists, it holds that
$F(\lim_{t\rightarrow\infty}\mathI^t)=\lim_{t\rightarrow\infty} F(\mathI^t)$.
We remark that both of these conditions are very weak: all of the rules considered in this paper are anonymous, and all except \minrule{} and \geo{} are continuous.\footnote{\minrule{} and \geo{} fail continuity because they return the distribution $(\frac{1}{m},\dots, \frac{1}{m}$) when $\min_{i\in N} s_{i,j}=0$ for all $j\in [m]$. Apart from this corner case, these rules are also continuous.} 

\begin{theorem}\label{lem:indavg}
The following claims hold for each $n\ge 2$.
\begin{enumerate}[label=(\arabic*),topsep=4pt,itemsep=0pt]
    \item \avg{} is the only aggregation rule that satisfies anonymity, score-unanimity, and independence if $m\geq 3$.
    \item \avg{} is the only coordinate-wise aggregation rule that satisfies anonymity, continuity, and score-unanimity if $m\geq 4$.
\end{enumerate}
\end{theorem}
\begin{proof}
    It is easy to verify that \avg{} satisfies anonymity and continuity, and we have argued that it satisfies 
    score-unanimity (Theorem~\ref{thm:efficiencyprops}) and independence (Theorem~\ref{thm:consistencyprops}). Thus, we focus on showing that the given sets of axioms indeed characterize \avg{}. We will establish the two claims separately.\medskip

    \noindent\textbf{Claim 1}: Let $F$ be an aggregation rule that satisfies anonymity, score-unanimity, and independence. 
    In particular, this means that $F(\mathI)_j=\gamma$ whenever $s_{i,j}=\gamma$ for all $i\in N$. 
    We establish the result in three steps.
    First, we will show that for every $\gamma\in (0,1]$, there is a constant $C_\gamma$ such that 
    for all $n$-agent instances $\mathI$ satisfying $s_{i, j}=\gamma$ and 
    $s_{i', j}=0$ for some $j\in [m]$, 
    $i\in N$ and all $i'\in N\setminus\{i\}$, it holds that
    $F(\mathI)_j=C_\gamma$. 
    In the second step, we will prove that $C_\gamma=\frac{\gamma}{n}$. 
    Based on this insight, in the last step we will derive that $F$ coincides with \avg{}.
    \medskip

    \emph{Step 1}: 
    To prove our first claim, 
    fix $n\ge 2$, $m\ge 3$ and
    consider two $n$-agent $m$-candidate instances $\mathI^1$ and $\mathI^2$ for which there exist two candidates $c_{j_1}$ and $c_{j_2}$ and two agents $i_1$ and $i_2$ such that \emph{(i)} $s_{i_1,j_1}^1=\gamma$ and $s_{i',j_1}^1=0$ for all $i'\in N\setminus \{i_1\}$, and \emph{(ii)} $s_{i_2,j_2}^2=\gamma$ and $s_{i',j_2}^2=0$ for all $i'\in N\setminus \{i_2\}$. 
    We will show that $F(\mathI^1)_{j_1}=F(\mathI^2)_{j_2}$. Since $\mathI^1$ and $\mathI^2$ are chosen arbitrarily, this implies that there exists a constant $C_\gamma$ such that $F(\mathI)_j=C_\gamma$ for all $n$-agent $m$-candidate instances in which a single agent assigns score $\gamma$ to candidate $c_j$, while all other agents assign score $0$ to that candidate.

    To prove that $F(\mathI^1)_{j_1}=F(\mathI^2)_{j_2}$, let $\hat \mathI^2$ denote the instance derived from $\mathI^2$ by exchanging the preferences of agents $i_2$ and $i_1$. We then have $\hat s_{i_1, j_2}^2=\gamma$ and $\hat s_{i',j_2}^2=0$ for all other agents $i'\in N\setminus \{i_1\}$. Clearly, anonymity requires that $F(\hat \mathI^2)=F(\mathI^2)$, as we derive $\hat \mathI^2$ by renaming the agents in $\mathI^2$. Now, if $j_1=j_2$, it holds that $\hat s_{i, j_1}^2=s_{i, j_1}^1$ for all $i\in N$, so independence implies that $F(\mathI^1)_{j_1}=F(\hat \mathI^2)_{j_2}$ in this case. Hence, if $j_1=j_2$, it follows that $F(\mathI^1)_{j_1}=F(\mathI^2)_{j_2}$.
    
    Next, assume that $j_1\neq j_2$. 
    In this case, we let $c_{j_3}$ denote an arbitrary candidate in $C\setminus \{c_{j_1}, c_{j_2}\}$ and consider the instances $\tilde \mathI^1$ and $\tilde \mathI^2$ described in Table~\ref{tab:tilde}.
\begin{table}   
\begin{center}
        \begin{tabular}{ c | c c c c }
          $\tilde\mathI^1$ & $s_{i,j_1}$ & $s_{i,j_2}$ & $s_{i,j_3}$ & $s_{i,j}$ for $j\in [m]\setminus \{j_1,j_2,j_3\}$\\ 
         \hline \hline
         $i_1$ & $\gamma$ & $1-\gamma$ & $0$ & $0$ \\ 
         $i'\in N\setminus \{i_1\}$ & $0$ & $1-\gamma$ & $\gamma$ & $0$ \\ 
        \end{tabular}\medskip
        
        \begin{tabular}{ c | c c c c }
        \ $\tilde\mathI^2$ & $s_{i,j_1}$ & $s_{i,j_2}$ & $s_{i,j_3}$ & $s_{i,j}$ for $j\in [m]\setminus \{j_1,j_2,j_3\}$\\ 
         \hline \hline
         $i_1$ & $1-\gamma$ & $\gamma$ & $0$ & $0$ \\ 
         $i'\in N\setminus \{i_1\}$ & $1-\gamma$ & $0$ & $\gamma$ & $0$ \\ 
        \end{tabular}  
    \end{center}
    \caption{Instances $\tilde \mathI^1$ and $\tilde \mathI^2$ in the proof of Theorem~\ref{lem:indavg}.}
    \label{tab:tilde}
    \end{table}
    We infer from score-unanimity that $F(\tilde \mathI^1)_{j_2}=F(\tilde \mathI^2)_{j_1}=1-\gamma$ and $F(\tilde \mathI^1)_{j}=F(\tilde \mathI^2)_{j}=0$ for all $j\in [m]\setminus \{j_1,j_2,j_3\}$.  
    We therefore have $F(\tilde \mathI^1)_{j_1}=\gamma-F(\tilde \mathI^1)_{j_3}$ and $F(\tilde \mathI^2)_{j_2}=\gamma-F(\tilde \mathI^2)_{j_3}$.
    Moreover, independence implies that $F(\tilde \mathI^1)_{j_3}=F(\tilde \mathI^2)_{j_3}$. Hence, $F(\tilde \mathI^1)_{j_1}=F(\tilde \mathI^2)_{j_2}$.
    On the other hand, we have 
    $s_{i,j_1}^1=\tilde s_{i,j_1}^1$ and $\hat s_{i, j_2}^2=\tilde s_{i,j_2}^2$ for all $i\in N$, so by
    independence, $F(\mathI^1)_{j_1}=F(\tilde \mathI^1)_{j_1}$ and $F(\hat \mathI^2)_{j_2}=F(\tilde \mathI^2)_{j_2}$.
    Hence, we conclude that $F(\mathI^1)_{j_1}=F(\hat \mathI^2)_{j_2}=F(\mathI^2)_{j_2}$.
    \medskip

    \emph{Step 2}: For our second step, we fix a value $\gamma\in (0,1]$ and let $C_\gamma$ denote the constant derived in Step~1. 
    The goal of this step is to show that $C_\gamma=\frac{\gamma}{n}$. 
    To prove this claim, let $c_{j_1}$, $c_{j_2}$, and $c_{j_3}$ denote three distinct candidates. 
    We will prove by induction that 
    for the instances $\mathI^k$ shown below, it holds that
    $F(\mathI^k)_{j_1}=k\cdot C_\gamma$.   
    \begin{center}
        \begin{tabular}{ c | c c c c }
        \ $\mathI^k$ & $s_{i,j_1}$ & $s_{i,j_2}$ & $s_{i,j_3}$ & $s_{i,j}$ for $j\in [m]\setminus \{j_1,j_2,j_3\}$\\ 
         \hline \hline
         $i\in \{1,\dots, k\}$ & $\gamma$   &  $0$ & $1-\gamma$ & $0$\\
    $i\in \{k+1,\dots, n\}$& $0$ &  $\gamma$   &  $1-\gamma$ & $0$\\
        \end{tabular}
    \end{center}
    In particular, this means that $F(\mathI^{n-1})_{j_1}=(n-1)\cdot C_\gamma$. Moreover, we infer from Step 1 that $F(\mathI^{n-1})_{j_2}=C_\gamma$, and score-unanimity implies $F(\mathI^{n-1})_{j}=0$ for all $j\in [m]\setminus \{j_1, j_2, j_3\}$ and $F(\mathI^{n-1})_{j_3}=1-\gamma$. Since $\sum_{j\in [m]} F(\mathI^{n-1})_j=1$, it follows that $C_\gamma=\frac{\gamma}{n}$.

    For the proof that $F(\mathI^k)_{j_1}=k\cdot C_\gamma$ for all $k\in \{1,\dots, n-1\}$, we first note that $F(\mathI^1)_{j_1}=C_\gamma$ by Step 1. Next, we inductively assume that $F(\mathI^k)_{j_1}=k\cdot C_\gamma$ for some $k\in \{1,\dots, n-2\}$ and aim to show that $F(\mathI^{k+1})_{j_1}=(k+1)\cdot C_\gamma$. 
    To this end, we consider the instances $\hat \mathI^k$ and $\tilde \mathI^k$ shown below.\smallskip

     \begin{center}
        \begin{tabular}{ c | c c c c }
         $\hat \mathI^k$ & $s_{i,j_1}$ & $s_{i,j_2}$ & $s_{i,j_3}$ & $s_{i,j}$ for $j\in [m]\setminus \{j_1,j_2,j_3\}$\\ 
         \hline \hline
         $i\in \{1,\dots, k\}$ & $\gamma$   &  $0$ & $1-\gamma$ & $0$\\
         $k+1$ & $0$ & $\gamma$ & $1-\gamma$ & $0$\\
    $i\in \{k+2,\dots, n\}$& $0$ &  $0$   &  $1$ & $0$\\
        \end{tabular}\medskip
        
        \begin{tabular}{ c | c c c c }
         $\tilde \mathI^k$ & $s_{i,j_1}$ & $s_{i,j_2}$ & $s_{i,j_3}$ & $s_{i,j}$ for $j\in [m]\setminus \{j_1,j_2,j_3\}$\\ 
         \hline \hline
         $i\in \{1,\dots, k+1\}$ & $\gamma$   &  $0$ & $1-\gamma$ & $0$\\
    $i\in \{k+2,\dots, n\}$& $0$ &  $0$   &  $1$ & $0$\\
        \end{tabular}
    \end{center}

    By independence and the induction hypothesis, we have $F(\hat \mathI^k)_{j_1}=F(\mathI^{k})_{j_1}=k\cdot C_\gamma$. 
    Moreover, Step~1 shows that $F(\hat \mathI^k)_{j_2}=C_\gamma$ and score-unanimity requires that $F(\hat \mathI^k)_{j}=0$ for all $j\in [m]\setminus\{j_1,j_2,j_3\}$. Hence, we derive $F(\hat \mathI^k)_{j_3}=1-(k+1)\cdot C_\gamma$. 
    Next, we turn to the instance $\tilde \mathI^k$. By score-unanimity, we have $F(\tilde \mathI^k)_{j}=0$ for all $j\in [m]\setminus \{j_1,j_3\}$, and hence $F(\tilde \mathI^k)_{1}+F(\tilde \mathI^k)_{3}=1$.
    On the other hand, independence implies $F(\tilde \mathI^k)_{j_3}=F(\hat \mathI^k)_{j_3}=1-(k+1)\cdot C_\gamma$.
    It follows that $F(\tilde \mathI^k)_{j_1}=(k+1)\cdot C_\gamma$. 
    Finally, independence implies $F(\mathI^{k+1})_{j_1}=F(\tilde \mathI^k)_{j_1}=(k+1)\cdot C_\gamma$, which completes the induction step.\medskip

    \emph{Step 3}: For our last step, we will show that $F$ corresponds to \avg{}.
    Fix $n\ge 2$, $m\ge 3$, and consider an 
    $n$-agent, $m$-candidate instance $\mathI$ and an arbitrary candidate $c_{j_1}$. Our goal is to show that $F(\mathI)_{j_1}=\frac{1}{n}\sum_{i\in N} s_{i, j_1}$. 
    To prove this claim, we take a candidate $c_{j_2} \ne c_{j_1}$, and consider the following family of instances $\mathI^k$ for $k\in [n]$:
    \begin{itemize}
    \item[\emph{(i)}] 
    $s_{i, j_1}^k=s_{i,j_1}$ and $s_{i,j_2}^k=1-s_{i,j_1}$ for all $i\in \{1,\dots,k \}$, 
    \item[\emph{(ii)}]
    $s_{i,j_1}^k=0$ and $s_{i,j_2}^k=1$ for all $i\in \{k+1,\dots,n\}$, and 
    \item[\emph{(iii)}]
    $s_{i,j}^k=0$ for all $i\in N$ and $j\in [m]\setminus \{j_1,j_2\}$. 
    \end{itemize}
    By independence, it holds that $F(\mathI)_{j_1}=F(\mathI^n)_{j_1}$, so our goal is to show that $F(\mathI^n)_{j_1}=\frac{1}{n}\sum_{i\in N} s_{i, j_1}$. 

    We will prove by induction on $k\in [n]$ that $F(\mathI^k)_{j_1}=\frac{1}{n}\sum_{i=1}^k s_{i, j_1}$. For the base case $k=1$, we observe that $F(\mathI^1)_{j_1}=\frac{s_{1,j_1}^1}{n}=\frac{1}{n}\sum_{i=1}^1 s_{i,j_1}$ due to Step~2 (if $s_{1, j_1}>0$) or score-unanimity (if $s_{1, j_1}=0$). 
    Next, we inductively assume that $F(\mathI^k)_{j_1}=\frac{1}{n}\sum_{i=1}^k s_{i,j_1}$ for some $k\in \{1,\dots, n-1\}$, and aim to show the same for $k+1$. 
    
    If $s_{k+1,\, j_1}=0$, this follows by independence, as $F(\mathI^{k+1})_{j_1}=F(\mathI^{k})_{j_1}=\frac{1}{n}\sum_{i=1}^{k+1} s_{i,j_1}$. 
    Thus, assume that $s_{k+1,\, j_1}>0$ and consider the instance $\bar \mathI^k$ derived from $\mathI^k$ by setting $\bar s_{k+1,\,j_3}^k=s_{k+1,\,j_1}$ for some arbitrary ${j_3}\in [m]\setminus \{j_1,j_2\}$ and $\bar s_{k+1,\,j_2}^k=1-s_{k+1,\,j_1}$. 
    By independence, it holds that $F(\bar \mathI^{k})_{j_1}=F(\mathI^k)_{j_1}=\frac{1}{n} \sum_{i=1}^k s_{i, j_1}$. Furthermore, $F(\bar \mathI^k)_{j_3}=\frac{s_{k+1,\,j_1}}{n}$ by Step 2. Since $F(\bar\mathI^k)_j=0$ for all $j\in [m]\setminus \{j_1,j_2,j_3\}$ due to score-unanimity, we infer that $F(\bar \mathI^k)_{j_2}=1-F(\bar \mathI^k)_{j_1}-F(\bar \mathI^k)_{j_3}=1-\frac{1}{n}\sum_{i=1}^{k+1} s_{i, j_1}$. 
    
    Next, consider the instances $\bar \mathI^k$ and $\mathI^{k+1}$, and candidate $j_2$. By independence, $F(\mathI^{k+1})_{j_2}=F(\bar \mathI^k)_{j_2}=1-\frac{1}{n}\sum_{i=1}^{k+1} s_{i, j_1}$. Moreover, by score-unanimity it holds that
    $F(\mathI^{k+1})_{j}=0$ for all $j\in [m]\setminus \{j_1,j_2\}$, and hence $F(\mathI^{k+1})_{j_1}=\frac{1}{n}\sum_{i=1}^{k+1} s_{i, j_1}$. 
    This completes the induction step. 
    It follows that $F(\mathI)_{j_1}=F(\mathI^n)_{j_1}=\frac{1}{n}\sum_{i\in N} s_{i,j_1}$, and we conclude that $F$ coincides with \avg{}.\medskip

    \noindent\textbf{Claim 2}: Let $F$ be a coordinate-wise aggregation rule that satisfies anonymity, continuity, and score-unanimity. 
For each $j\in [m]$, let $f_j$ denote the coordinate-aggregation function of $F$ for the $j$-th coordinate. 
In a slight abuse of notation, we will write $f_j(\mathI)$ to mean $f_j(s_{1,j},\dots, s_{n,j})$. 
Since $F$ is scale-invariant, in the sense that its output does not change if we multiply all functions $f_j$ by a constant, we can assume without loss of generality that $f_1(0.5,\dots,0.5)=0.5$.

In the remainder of the proof, we will show that for each $n\ge 2$, $m\ge 4$
and each $n$-agent $m$-candidate instance $\mathI$, it holds that
$f_j(\mathI)=\frac{1}{n}\sum_{i\in N} s_{i,j}$ for each $j\in [m]$. To this end, we first prove this claim for the case where $c_j$ receives the same score $\gamma \in [0,1)$ from all agents.
In the second step, we then use our first characterization of \avg{} to show that $f_j(\mathI)=\frac{1}{n}\sum_{i\in N} s_{i,j}$ for all instances $\mathI$ and candidates $c_j$ with $\max_{i\in N} s_{i,j}<1$. Finally, in the last step we use continuity to infer that $F$ corresponds to \avg{}. 
\medskip

\emph{Step 1}: As the first step, we show that $f_j(\gamma,\dots, \gamma)=\gamma$ for all $j\in [m]$ and $\gamma\in [0,1)$.
The general outline of our argument is as follows. Consider an instance $\mathI$ where $s_{i, j}=\gamma$ for all $i\in N$.
By score-unanimity we have 
$F(\mathI)_j=\gamma$ and, furthermore, 
$F(\mathI)_j = \frac{f_j(\mathI)}{\sum_{j'\in[m]} f_{j'}(\mathI)}$.
Thus, we need to establish that 
$\sum_{j'\in[m]} f_{j'}(\mathI)=1$.
To this end, it suffices to show that 
$F(\mathI)_k = f_k(\mathI)$ for some $k\in [m]$.

To start, observe that $f_j(0,\dots, 0)=0$ for all $j\in [m]$ because of score-unanimity. 
Next, we will argue that $f_j(0.5,\dots, 0.5)=0.5$ for all $j\in [m]$. 
To this end, fix some candidate $c_j\neq c_1$ and consider the instance $\mathI^1$ where 
for each $i\in N$ it holds that $s_{i,1}^1=s_{i,j}^1=0.5$ and $s_{i,j'}=0$ for all $j'\in [m]\setminus \{1,j\}$. 
By score-unanimity, we have $F(\mathI^1)_1=F(\mathI^1)_j=0.5$. 
Since we assume that $f_1(0.5,\dots, 0.5)=0.5$, we obtain
$f_1(\mathI^1)=F(\mathI^1)_1$.
As argued in the previous paragraph, 
this implies $f_j(\mathI^1)=F(\mathI^1)_j=0.5$, and therefore
$f_j(0.5,\dots,0.5)=0.5$ for all $j\in [m]$.

Next, we show that $f_j(\gamma,\dots,\gamma)=\gamma$ for all $\gamma\in (0,0.5)$ and $j\in [m]$. 
To this end, we fix three distinct candidates $c_{j}$, $c_{k}$, and $c_{\ell}$, and consider the instance $\mathI^2$ such that $s_{i,j}^2=\gamma$, $s_{i,k}^2=0.5$, and $s_{i, \ell}^2=0.5-\gamma$ for all $i\in N$. All other candidates receive score $0$ from all agents. 
By score-unanimity we obtain 
$F(\mathI^2)_j=\gamma$ and $F(\mathI^2)_k=0.5$, and
by the analysis in the previous paragraph
we have $f_k(\mathI^2)=0.5=F(\mathI^2)_k$.
As argued at the start of the proof, this implies $f_j(\mathI^2)=F(\mathI^2)_j$, so we conclude that $f_{j}(\gamma,\dots,\gamma)=\gamma$.

Finally, we consider the case $\gamma\in (0.5,1)$. 
Fix two candidates $c_{j}$ and $c_{k}$, and consider the instance~$\mathI^3$ where $s_{i,j}^3=\gamma$ and $s_{i,k}^3=1-\gamma$ for all $i\in N$. All other candidates again obtain score $0$ from all agents. 
By score-unanimity, we have 
$F(\mathI^3)_{j}=\gamma$ and
$F(\mathI^3)_{k}=1-\gamma$. Moreover, 
since $1-\gamma<0.5$, the argument in the previous paragraph implies that $f_k(\mathI^3) = 1-\gamma=F(\mathI^3)_{k}$. As argued at the start of the proof, this implies
$f_j(\mathI^3) = F(\mathI^3)_{j}$, so we infer that $f_{j}(\gamma,\dots,\gamma)=\gamma$.
This completes the proof for Step 1.\medskip

\emph{Step 2}: Next, we prove that $f_j(\mathI)=\frac{1}{n}\sum_{i\in N} s_{i,j}$ for all instances $\mathI$ and candidates $c_j$ with $\max_{i\in N} s_{i,j}<1$. Fix an arbitrary instance $\mathI^*$ and a candidate $c_j$ that satisfy our requirements. Moreover, define $\epsilon=1-\max_{i\in N} s^*_{i,j}$, and note that $\epsilon>0$ by our choice of $j$. 
If $\epsilon=1$, then in $\mathI^*$ all agents assign score $0$ to $c_{j}$, and score-unanimity immediately implies that $f_{j}(\mathI^*)=0=\frac{1}{n}\sum_{i\in N} s_{i,j}^*$. 

Hence, we assume that $\epsilon<1$, and consider three other candidates $c_k, c_\ell, c_t$. 
Furthermore, we define another aggregation rule $G$ for the candidates $\{c_j, c_k, c_\ell\}\subsetneq C$ as follows:  
Given an instance~$\mathI$ on these three candidates, we construct an extended instance $\mathI^E$ on $C$ with the same set of agents $N$ by setting, for each $i\in N$, 
$$
s_{i,j'}^E=
\begin{cases}
(1-\epsilon)s_{i,j'} &\text{for $j'\in \{j,k,\ell\}$},\\ 
\epsilon &\text{for $j'=t$},\\ 
0 &\text{for $j'\in [m]\setminus \{j,k,\ell,t\}$.}
\end{cases}
$$
Then, $G(\mathI)_{j'}=\frac{1}{1-\epsilon} \cdot F(\mathI^E)_{j'}$ for all $j'\in \{j,k,\ell\}$. 

Our goal is to show that $G$ coincides with \avg{}. Before we prove this, let us show why this implies that $f_{j}(\mathI^*)=\frac{1}{n}\sum_{i\in N} s_{i,j}^*$.
Indeed, if $G$ is equal to \avg{}, for each instance $\mathI$ on $\{c_{j}, c_{k}, c_{\ell}\}$
we have 
$$
F(\mathI^E)_{j}=(1-\epsilon) G(\mathI)_{j}=(1-\epsilon)\cdot\frac{1}{n}\sum_{i\in N} s_{i, j} = \frac{1}{n}\sum_{i\in N} s_{i, j}^E.
$$
 Furthermore, by score-unanimity and Step~1, we have $F(\mathI^E)_{t}=\epsilon=f_{t}(\mathI^E)$, which implies that $\sum_{j'\in [m]} f_{j'}(\mathI^E)=1$. In turn, this means that $F(\mathI^E)_{j'}=f_{j'}(\mathI^E)$ for all $j'\in [m]$, so it follows that $f_{j}(\mathI^E)=\frac{1}{n} \sum_{i\in N} s^{E}_{i, j}$ for all instances $\mathI$ on $\{c_{j}, c_{k}, c_{\ell}\}$. 
Now, consider an instance $\mathI$
on $\{c_{j}, c_{k}, c_{\ell}\}$ such that for each $i\in N$ we have $s_{i, j} = \frac{s^*_{i,j}}{1-\epsilon}$ and $s_{i, k}=s_{i, \ell}=\frac{1-s_{i, j}}{2}$; note that $s^*_{i, j}\le 1-\epsilon$ and hence $s_{i, j}\le 1$, so $\mathI$ is well-defined.
We have $s^*_{i,j}=s^E_{i,j}$ for all $i\in N$, so we conclude that $f_{j}(\mathI^*)=\frac{1}{n}\sum_{i\in N} s_{i, j}^*$. 

It remains to show that $G$ is indeed \avg{}. 
To this end, we aim to employ our first characterization and show that $G$ is a well-defined aggregation rule that satisfies anonymity, score-unanimity, and independence.
First, it is easy to verify that $G$ is well-defined. Indeed,
for each instance $\mathI$ it holds that $F(\mathI^E)_{j'}\ge 0$
for all $j'\in \{j,k,\ell\}$, and hence $G(\mathI)_{j'}\geq 0$ for all  $j'\in \{j,k,\ell\}$.  
Moreover, score-unanimity implies that $F(\mathI^E)_{t}=\epsilon$ and $F(\mathI^E)_{j'}=0$ for all $j'\in[m] \setminus \{j,k,\ell,t\}$, so $F(\mathI^E)_{j}+F(\mathI^E)_{k}+F(\mathI^E)_{\ell}=1-\epsilon$. 
It follows that $G(\mathI)_{j}+G(\mathI)_{k}+G(\mathI)_{\ell}=1$ for all instances $\mathI$. 

Next, we show that $G$ is anonymous. 
Consider two instances $\mathI$ and $\hat \mathI$ such that $\hat \mathI$ is derived from~$\mathI$ by permuting the agents. 
Consequently, the instances $\mathI^E$ and $\hat \mathI^E$ can be derived from each other by permuting the agents, so the anonymity of $F$ implies that $F(\mathI^E)=F(\hat \mathI^E)$. 
This also means that $G(\mathI)=G(\hat \mathI)$, and so $G$ is anonymous. 

To show that $G$ is score-unanimous, we consider an instance $\mathI$ such that all agents $i\in N$ assign the same score $\gamma$ to some candidate $c_{j'}$ where $j'\in\{j, k, \ell\}$. 
Hence, all agents assign score $(1-\epsilon)\gamma$ to $c_{j'}$ in the extended instance $\mathI^E$. 
Score-unanimity of $F$ then implies that $F(\mathI^E)_{j'}=(1-\epsilon)\gamma$, which entails that 
$G(\mathI)_{j'}=\frac{1}{1-\epsilon}\cdot F(\mathI^E)_{j'}=\gamma$. 
It follows that $G$ satisfies this axiom, too. 

Finally, we show that $G$ satisfies independence. 
To this end, consider two instances $\mathI$, $\hat \mathI$ on $\{c_{j}, c_{k}, c_{\ell}\}$, and 
$j'\in\{j, k, \ell\}$ such that $s_{i, j'}=\hat s_{i,j'}$ for all $i\in N$.
This means that $s^E_{i,j'}=\hat s^E_{i,j'}$ for all $i\in N$, so $f_{j'}(\mathI^E)=f_{j'}(\hat \mathI^E)$. 
Moreover, by score-unanimity and Step~1, it holds that $F(\mathI^E)_{t}=\epsilon = f_{t}(\mathI^E)$ and $ F(\hat \mathI^E)_{t} = \epsilon = f_{t}(\hat \mathI^E)$. 
We hence infer that $\sum_{t\in [m]} f_{t}(\mathI^E)=\sum_{t\in [m]} f_{t}(\hat \mathI^E)=1$. 
This means that $F(\mathI^E)_{j'}=f_{j'}(\mathI^E)=f_{j'}(\hat \mathI^E)=F(\hat \mathI^E)_{j'}$, so $G$ satisfies independence
because $G(\mathI)_{j'}=\frac{1}{1-\epsilon}\cdot F(\mathI^E)_{j'}=\frac{1}{1-\epsilon}\cdot F(\hat\mathI^E)_{j'}=G(\hat \mathI)_{j'}$. 

Since the rule $G$ satisfies all axioms of Claim~1, we conclude that it coincides with \avg{}, which completes the proof for this step.\medskip

\emph{Step 3}: By the insights of Step 2, we obtain $f_j(\mathI)=\frac{1}{n}\sum_{i\in N} s_{i,j}$ for all $n$-agent instances $\mathI$ and candidates $c_j$ such that $\max_{i\in N} s_{i,j}<1$. 
This means that $F$ is equal to \avg{} for all instances $\mathI$ with $\max_{i\in N,\, j\in [m]} s_{i,j}<1$. 
In order to extend the result to instances where some agents assign score $1$ to some candidate, we use the continuity of $F$. 
Specifically, consider an instance $\mathI^*$ such that $s^*_{i,j}=1$ for 
some $i\in N, j\in [m]$.
Let $\mathI$ denote the instance where every agent assigns score $\frac{1}{m}$ to every candidate. 
We can now consider the sequence of instances $\mathI^k$ defined by $s_{i,j}^k=\frac{1}{2^k} s_{i,j} + (1-\frac{1}{2^k})s_{i,j}^*$ for all $i\in N$, $j\in[m]$. 
Clearly, this sequence converges to $\mathI^*$.
Moreover, for every instance $\mathI^k$ and all $j\in [m]$, it holds that $F(\mathI^k)_j=\frac{1}{n}\sum_{i\in N} s_{i,j}^k$ due to Step 2. 
Hence, we can infer by continuity that $F(\mathI^*)_j=\lim_{k\to\infty} F(\mathI^k)_j=\frac{1}{n}\sum_{i\in N} s_{i,j}^*$.
This shows that $F$ coincides with \avg{}, as desired.
\end{proof}

\begin{remark}
    For both of our characterizations, all axioms are necessary. Every dictatorial aggregation rule (i.e., a rule that always returns the score vector of a specific agent) satisfies all given axioms except anonymity. Every constant aggregation rule satisfies all given axioms except score-unanimity. \util{} satisfies all given axioms except independence and coordinate-wiseness. Furthermore, the coordinate-wise rule defined by the coordinate-aggregation function $f_j(\mathI)=1$ if there is an agent $i$ with $s_{i,j}=1$ and $f_j(\mathI)=\frac{1}{n}\sum_{i\in N} s_{i,j}$ otherwise satisfies all axioms except independence and continuity.
    Finally, the condition that $m\geq 3$ is necessary for our first characterization, as otherwise independence becomes trivial and, e.g., \util{} satisfies all given axioms (anonymity is immediate, and for score-unanimity and independence, see Theorem~\ref{thm:efficiencyprops} and Theorem~\ref{thm:consistencyprops}, respectively). Similarly, the condition that $m\geq 4$ is necessary for the second characterization because \med{} satisfies all conditions if $m\leq 3$ (anonymity and continuity are immediate, and for score-unanimity this is shown in Theorem~\ref{thm:efficiencyprops}). 
\end{remark}

\begin{remark}
    One can check that, in our first characterization, it is possible to replace score-unanimity and anonymity with score-representation. That is, \avg{} is the only aggregation rule that satisfies score-representation and independence. To see this, note that for instances of the form $\tilde \mathI^1$ (see Table~\ref{tab:tilde}), score-representation implies that candidate $c_{j_1}$ receives probability at least $\frac{\gamma}{n}$, candidate $c_{j_2}$ receives probability at least $1-\gamma$, and candidate $c_{j_3}$ receives probability at least $\frac{(n-1)\gamma}{n}$.
    Since the probabilities must sum up to $1$, this is only possible if all of these bounds are tight.
    By independence and the fact that this argument does not depend on the identities of agents or candidates, this immediately completes the proof of Steps~1 and 2.
    For Step~3, we used score-unanimity to conclude that every candidate that receives score $0$ from all agents is assigned probability $0$; this also follows from score-representation in conjunction with independence.
    
    Perhaps surprisingly, this claim does not hold for our second characterization. Indeed, consider the coordinate-wise rule $G$ whose coordinate-aggregation functions $g_j$ assign to each candidate the minimum probability that meets the requirements imposed by score-representation.
    Note that since \avg{} satisfies score-representation, for every instance $\mathI$
    and all $j\in [m]$ we have $g_j(\mathI)\le \avg{}(\mathI)_j$, and hence $\sum_{j\in [m]}g_j(\mathI)\le \sum_{j\in [m]}\avg{}(\mathI)_j=1$, which implies that $G(\mathI)_j\ge g_j(\mathI)$ for each $j\in [m]$.
    Therefore, $G$ satisfies score-representation, anonymity, and continuity. 
\end{remark}

\begin{remark}
    Another natural way to characterize \avg{} is to rely on convexity. Specifically, we say that an aggregation rule $F$ is \emph{weakly convex} if $F(\mathI'')=F(\mathI)$ for all instances $\mathI$, $\mathI'$, and $\mathI''$ such that $F(\mathI)=F(\mathI')$ and there exists $\lambda\in (0,1)$ with $s_{i,j}''=\lambda s_{i,j} + (1-\lambda) s_{i,j}'$ for all $i\in N$, $j\in [m]$. Then, one can show that \avg{} is the only aggregation rule that satisfies weak convexity, anonymity, and score-unanimity. The idea is that, given an instance $\mathI$, we can consider all possible permutations $\pi:N\rightarrow N$, apply each of them to $\mathI$, 
    and then take the average of all permuted instances. This results in an instance $\mathI^*$ where for each $j\in [m]$ it holds that
    all agents assign score $\frac{1}{n}\sum_{i\in N} s_{i,j}$ to $c_j$. Hence, score-unanimity requires $F(\mathI^*)_j=\frac{1}{n}\sum_{i\in N} s_{i,j}$ for each $j\in [m]$, and weak convexity and anonymity imply that the same holds for $\mathI$. 
\end{remark}

\begin{remark}
    All coordinate-wise (and anonymous) rules from \Cref{tab:summary} violate Pareto optimality. 
    Intuitively, this does not come as a surprise, as coordinate-wise rules are unable to take into account correlations between scores and Pareto improvements.
    Formally, the first two steps in the proof of Claim~2 of \Cref{lem:indavg} show that any coordinate-wise, anonymous, and score-unanimous (which is necessary for Pareto optimality) rule has to coincide with \avg{} on all instances with $s_{i,j} \neq 1$ for all $i\in N$, $j\in[m]$. It is not difficult to see---e.g., by adapting the instance $\mathcal{I}^1$ in the proof of \Cref{thm:efficiencyprops}---that no such rule satisfies Pareto optimality. 
    Since independence implies coordinate-wiseness, this shows that any Pareto optimal and anonymous rule necessarily fails independence.
\end{remark}

\section{Conclusion}

In this paper, we have analyzed aggregation rules for portioning with cardinal preferences from an axiomatic perspective.
Specifically, we considered a natural model in which each agent reports her ideal distribution of a homogeneous resource over a set of candidates, and her disutility for a distribution corresponds to the $\ell_1$ distance from her ideal distribution.
We investigated rules based on coordinate-wise aggregation or welfare aggregation as well as the independent markets rule of \citet{freeman2021truthfulbudget} with respect to efficiency, fairness, consistency, and incentive properties.
Our results, which are summarized in \Cref{tab:summary}, show that the rule that simply returns the average of the agents' reports satisfies most of the studied axioms.
In particular, even though this rule violates strategyproofness and Pareto optimality,\footnote{Note also that rules that satisfy strategyproofness under $\ell_1$ utilities may fail to do so under different utility models---see, e.g., Endnote~1 in the work of \citet{brandt2024optimal}.
Moreover, the average rule satisfies Pareto optimality under $\ell_2$ utilities, as it minimizes the sum of squared $\ell_2$ disutilities \citep[p.~129]{dodge1999multivariate}.} it is the only rule among the ones we consider that fulfills the strong fairness notion of score-representation as well as the strong consistency property of independence.
To further strengthen this point, we provided two characterizations demonstrating that the average rule is the only rule within large classes of rules that satisfies, for example, independence and score-unanimity at the same time.

We believe that our paper can serve as a basis for extensive future research in the domain of cardinal portioning.
For instance, it could be worthwhile to derive additional characterizations of the average rule, and the insights that our findings offer may also be helpful toward characterizations of further rules.
One could verify whether rules that ``approximate'' the average rule preserve approximate versions of the properties that the average rule satisfies.
It would also be interesting to examine additional axioms, especially those concerning fairness, which is important but arguably not yet well-understood in this setting.
One such axiom is membership in the \emph{core}, which intuitively means that no subset of agents can guarantee a better outcome (in the sense of Pareto improvement) for themselves by using their proportional share of the resource.
The core strengthens both Pareto optimality (since the latter only imposes this requirement on the set of all agents) and single-minded proportionality, so \Cref{tab:summary} immediately implies that none of the rules we consider always returns an outcome in the core.
Whether any such rule exists is therefore an intriguing question, which we leave for future work.

\section*{Acknowledgments}
 This work was partially supported by the AI Programme of The Alan Turing Institute, by the Deutsche Forschungsgemeinschaft under grants BR 2312/11-2 and BR 2312/12-1, by the Singapore Ministry of Education under grant number MOE-T2EP20221-0001, by the NSF-CSIRO grant on ``Fair Sequential Collective Decision-Making'' (RG230833), and by an NUS Start-up Grant.
 Most of this research was done while the second author was at the Technical University of Munich.
We thank the anonymous reviewers of ECAI 2023 and Artificial Intelligence Journal for constructive feedback, and Erel Segal-Halevi for insightful comments.

\bibliographystyle{plainnat}
\bibliography{abb,main,group}

\appendix

\section{Omitted Proofs}
\label{app:omitted-proofs}

In this appendix, we prove two claims regarding computational aspects that were made in the body of the paper, as well as a technical lemma regarding medians of two sets.

First, we show that \egal{} (with leximin tie-breaking) can be computed in polynomial time. 

\begin{proposition}\label{thm:egal_compute}
    \emph{\egal{}} can be computed in polynomial time.
\end{proposition}
\begin{proof}
    Similarly to \citet[Alg.~1]{airiau2019portioning}, we formulate a series of linear programs (LP) for finding an \egal{} outcome. Let the objective function be
    \begin{equation*}
        \text{minimize } \xi
    \end{equation*}
    \noindent
    subject to the following constraints:
    \begin{enumerate}[label=(\arabic*),topsep=4pt,itemsep=0pt]
    \item $\sum_{j \in [m]} x_j = 1$;
        
    \item $x_j \geq 0 \text{ for each $j\in [m]$}$;
    
    \item $z_{i,j} \geq s_{i,j} - x_j \text{ and } z_{i,j} \geq x_j - s_{i,j}$ for each $i\in N$, $j\in [m]$;
    
    \item $\sum_{j \in [m]} z_{i,j} \leq \xi$ for each $i\in N$.
    \end{enumerate}
    Note that $z_{i,j}$ is an upper bound on $i$'s disutility for candidate $c_j$, so in an optimal solution it holds that $z_{i, j}=|s_{i, j}-x_j|$ for all $i\in N$, $j\in [m]$.
    
    This allows us to minimize the largest disutility~$\xi$. There is an agent $i$
    that has disutility $\xi$ in every leximin outcome: 
    indeed, if for every $i\in N$ there is an outcome in which $i$ incurs disutility less than $\xi$ and every other agent incurs disutility at most $\xi$, then by averaging these outcomes across all $i\in N$, we obtain an outcome in which every agent's disutility is less than $\xi$, contradicting the choice of $\xi$.
    To find such an agent $i^*$, for each $i\in N$ we formulate an LP that computes the maximum $\delta_i$ such that there exists an outcome in which agent~$i$ incurs disutility at most $\xi - \delta_i$ while every other agent incurs disutility at most $\xi$; for $i^*$ we have $\delta_{i^*}=0$.
    We fix the disutility of $i^*$ to $\xi$, and continue by finding the second largest disutility, and so on. At each iteration we fix the disutility of one agent by solving $O(n)$ LPs, so
    the total number of LPs is $O(n^2)$.
\end{proof}

Next, we prove that we can check in polynomial time whether an outcome is Pareto optimal. 

\begin{proposition} \label{thm:po_check}
    Determining whether an outcome~$\mathbf{x}'$ is Pareto optimal can be done in polynomial time.
\end{proposition}
\begin{proof}
    Suppose we are given an outcome $\mathbf{x}'$ for an instance $\mathcal{I} = (\mathbf{s}_1,\dots, \mathbf{s}_n)$ and we want to determine whether ${\mathbf x}'$ is Pareto optimal.
    For each $i \in N$ and $j \in [m]$, let $z'_{i,j} = |x'_j - s_{i,j}|$. The quantities $z'_{i,j}$ can be computed from the input, and will appear in the constraints of the linear program below.
    
    We formulate a linear program as follows.
    \begin{equation*}
        \text{minimize } \sum_{i \in N} \sum_{j \in [m]} z_{i,j}, 
    \end{equation*}    
    subject to the following constraints:
    \begin{enumerate}[label=(\arabic*),topsep=4pt,itemsep=0pt]
    \item $\sum_{j \in [m]} x_j = 1$;
    
    \item $x_j \geq 0$ for each $j\in [m]$;
    
    \item $z_{i,j} \geq s_{i,j} - x_j \text{ and } z_{i,j} \geq x_j - s_{i,j}$ for each $i \in N$, $j \in [m]$;
    
    \item $\sum_{j \in [m]} z_{i,j} \leq \sum_{j \in [m]} z'_{i,j}$ for each $i\in N$.
    \end{enumerate}
    
    Just as in the proof of Proposition~\ref{thm:egal_compute}, 
    in an optimal solution it holds that $z_{i, j}=|s_{i, j}-x_j|$ for all $i\in N$, $j\in [m]$. Thus, condition~(4)
    can be interpreted as requiring that $d_i({\mathbf x})\le d_i({\mathbf x}')$ for all $i\in N$, i.e., optimal solutions to our LP correspond to outcomes that are at least as good as ${\mathbf x}'$ for all agents. 
    It follows that a solution that Pareto dominates ${\mathbf x}'$ corresponds to an outcome for which the value of the objective function is strictly less than $\sum_{i \in N} \sum_{j \in [m]} z'_{i,j}$.
    Therefore, to check whether ${\mathbf x}'$ is Pareto dominated by some other outcome, it suffices to solve our LP and compare the optimal value to $\sum_{i \in N} \sum_{j \in [m]} z'_{i,j}$. 
\end{proof}

Finally, we present the proof of Lemma~5.3.

\medians*

     \begin{proof}
     Let $p=|A|$ and $q=|B|$.
     
     If $p$ is odd, there exists some $a\in A$ with $\text{med}(A)= a$,
         and two disjoint multisets $A^-, A^+\subseteq A\setminus \{a\}$ such that $|A^-|=|A^+|=\frac{p-1}{2}$,  $a'\le a$ for all $a'\in A^-$, 
         and $a''\ge a$ for all $a''\in A^+$.
         
     If $p$ is even, there exist $a^-, a^+\in A$ with $a^-\le a^+$ and $\text{med}(A)= (a^-+a^+)/2$,
         and two disjoint multisets $A^-, A^+\subseteq A\setminus \{a^-, a^+\}$ such that $|A^-|=|A^+|=\frac{p-2}{2}$, $a'\le a^-$ for all $a'\in A^-$, 
         and $a''\ge a^+$ for all $a''\in A^+$. 
         
     Similarly, we can represent $B$ as $B=B^-\cup\{b\}\cup B^+$ (if $q$ is odd) or
     $B=B^-\cup\{b^-, b^+\}\cup B^+$ (if $q$ is even).
     This representation has the property that $|B^-|=|B^+|$ 
     and furthermore, $b'\le b$ for all $b'\in B^-$ and $b''\ge b$ for all $b''\in B^+$
     (if $q$ is odd), and 
     $b'\le b^-$ for all $b'\in B^-$ and $b''\ge b^+$ for all $b''\in B^+$
     (if $q$ is even).   
     We consider the following cases.

     \medskip
         \underline{Case 1}: $p$ and $q$ are both odd. Then $a=b=z$, and we have 
         $c\le z$ for each $c\in A^-\cup B^-$ and $c\ge z$ for each $c\in A^+\cup B^+$.
         Thus, the median of $A\cup B$ is $(a+b)/2 = z$.

         \medskip
         \underline{Case 2}: $p$ is odd and $q$ is even. 
         Then $a=(b^-+b^+)/2=z$, and hence $b^-\le a\le b^+$.
         For each $c\in A^-\cup B^-\cup\{b^-\}$ we have $c\le a$, and 
         for each $c\in A^+\cup B^+\cup\{b^+\}$ we have $c\ge a$. It follows that 
         the median of $A\cup B$ is $a=z$.

         \medskip
         \underline{Case 3}: $p$ is even and $q$ is odd. In this case, the median is $b=z$; the argument is the same as in the previous case with the roles of $A$ and $B$ reversed.

         \medskip
         \underline{Case 4}: $p$ and $q$ are both even. Then $(a^-+a^+)/2=(b^-+b^+)/2=z$. Assume without loss of generality that $b^-\le a^-$, so $b^+\ge a^+$.
         For each $c\in A^-\cup B^-\cup\{b^-\}$ we have $c\le a^-$, and
          for each $c\in A^+\cup B^+\cup\{b^+\}$ we have $c\ge a^+$. It follows that 
         the median of $A\cup B$ is $(a^-+a^+)/2=z$. 

     \medskip
     This completes the proof of the lemma.
     \end{proof}

\section{Moving Phantoms Rules}
\label{app:moving-phantoms}

In this appendix, we consider two additional moving phantoms rules, namely, the \ladder{} rule of \citet{freeman2024project} and the piecewise uniform (\pu{}) rule of \citet{caragiannis2024truthful}.
The moving phantoms of these rules are defined as follows: 
\begin{equation*}
    f_k^\ladder{}(t) = \max \left(t - \frac{k}{n}, 0\right) \text{ for all } k \in \{0,\dots,n\} \text{ and all } t \in [0,1]
\end{equation*}
and 
\begin{equation*}
    f_k^\pu{}(t) = 
    \begin{cases}
        0 & \text{if } t < \frac{1}{2} \text{ and } \frac{k}{n} < \frac{1}{2}\\
        \frac{4tk}{n}-2t & \text{if } t < \frac{1}{2} \text{ and } \frac{k}{n} \geq \frac{1}{2} \\
        \frac{k(2t-1)}{n} & \text{if } t \geq \frac{1}{2} \text{ and } \frac{k}{n} < \frac{1}{2} \\
        \frac{k(3-2t)}{n}-2+2t & \text{if } t \geq \frac{1}{2} \text{ and } \frac{k}{n} \geq \frac{1}{2} 
    \end{cases}\qquad
    \text{ for all } k \in \{0,\dots,n\} \text{ and all } t \in [0,1].
\end{equation*}

\citet[Thms.~2 and 3]{freeman2021truthfulbudget} have shown that every moving phantoms rule is strategyproof and score-monotone, so both \ladder{} and \pu{} satisfy these axioms. Furthermore, \citet[p.~9709]{freeman2024project} and \citet[Thm.~3]{caragiannis2024truthful} proved that both rules satisfy single-minded proportionality. 
We shall therefore investigate these rules with respect to our remaining axioms.
As it turns out, these rules exhibit similar behavior to \im{}---a noteworthy exception is that \pu{} fails reinforcement while \im{} and \ladder{} satisfy this property. 

To begin with, we show that \im{}, \ladder{}, and \pu{} all reduce to the \emph{uniform phantom} rule of \citet{caragiannis2016truthful} when there are only $m=2$ candidates.
Consequently, the three rules coincide and exhibit the same properties in this case. 
\begin{proposition}\label{prop:movingphantomm=2}
    \im{}, \ladder{}, and \pu{} coincide when $m = 2$. 
\end{proposition}
\begin{proof}
    To prove this claim, we recall Proposition 1 of \citet{freeman2021truthfulbudget}, which states that the uniform phantom rule is the only aggregation rule for $m=2$ that satisfies anonymity, continuity, single-minded proportionality, and strategyproofness. 
    Moreover, these authors observed that, in addition to satisfying strategyproofness, every moving phantoms rule is anonymous and continuous \citep[p.~10]{freeman2021truthfulbudget}. Finally, as we mentioned earlier, \im{}, \ladder{}, and \pu{} all satisfy single-minded proportionality. 
    It follows that when $m = 2$, all three rules coincide with the uniform phantom rule, and therefore with one another.
\end{proof}

We now analyze \ladder{} and \pu{} separately, starting with \ladder{}.

\begin{theorem}
    The following claims hold.
    \begin{enumerate}[label=(\arabic*),topsep=4pt,itemsep=0pt]
        \item \ladder{} satisfies range-respect (and therefore Pareto optimality and score-unanimity) when $m=2$, but fails score-unanimity (and therefore range-respect and Pareto optimality) when $m \geq 3$ and $n \geq 2$.
        \item \ladder{} satisfies score-representation when $m=2$, but fails to do so for all $m \geq 3$ and $n \geq 2$. 
        \item \ladder{} satisfies independence when $m= 2$, but fails to do so for all $m\geq 3$ and $n\geq 2$. 
        \item \ladder{} satisfies reinforcement.
        \item \ladder{} satisfies participation.
    \end{enumerate}
\end{theorem}
\begin{proof}
    We prove each of the claims separately. 
    Moreover, we omit the proofs for the case $m=2$, as all corresponding claims follow from \Cref{prop:movingphantomm=2} together with the analysis of \im{}.\medskip

    \noindent \textbf{Claim 1}: We only need to show that \ladder{} fails score-unanimity when $m\geq 3$ and $n\geq 2$. To this end, consider the instance $\mathcal{I}^5$ from the proof of \Cref{thm:efficiencyprops}, which we reproduce below for convenience. All candidates $c_j$ with $j\geq 4$ receive a score of $0$ from all agents and can thus be ignored.
 
\begin{center}
        \begin{tabular}{ c | c c c c }
          $\mathcal{I}^5$ & $s_{i,1}$ & $s_{i,2}$ & $s_{i,3}$  \\ 
         \hline \hline
         $1$ & $\frac{n+1}{n+2}$ & $\frac{1}{n+2}$ & $0$ \\ 
         $i\in \{2,\dots, n\}$ &  $\frac{n+1}{n+2}$ & $0$ & $\frac{1}{n+2}$ 
        \end{tabular}       
    \end{center}
    For this instance, score-unanimity requires that $x_1 = \frac{n+1}{n+2}$. 
    
    If $n=2$, we claim that \ladder{} assigns a probability of $\frac{2}{3}$ to $c_1$. 
    Indeed, for $t^* = \frac{2}{3}$, we have $f_0^\ladder(t^*) = \frac{2}{3}$, $f_1^\ladder(t^*)=\frac{1}{6}$, and $f_2^\ladder{}(t^*)=0$. 
    Hence, the medians for candidates $c_1$, $c_2$, $c_3$ are $\frac{2}{3}$, $\frac{1}{6}$, $\frac{1}{6}$, respectively.
    Since the medians sum up to $1$, this means that \ladder{} returns $(\frac{2}{3}, \frac{1}{6}, \frac{1}{6})$, which violates score-unanimity.
    
    If $n \geq 3$, we claim that \ladder{} assigns a probability of $\frac{n^2+n-1}{n(n+2)}=\frac{n+1-\frac{1}{n}}{n+2}$ to $c_1$. 
    Indeed, for $t^*=\frac{n^2+n-1}{n(n+2)}$, it holds for all $k\in \{0,\dots, n\}$ that $f_k^\ladder(t^*)\leq f_0^\ladder{}(t^*)=\frac{n^2+n-1}{n(n+2)}<\frac{n+1}{n+2}$, so the median for $c_1$ is $\frac{n^2+n-1}{n(n+2)}$. 
    Next, for $c_2$, we observe that
    \begin{align*}
    f_{n-1}^\ladder(t^*)=\frac{n^2+n-1}{n(n+2)}-\frac{n-1}{n}=\frac{n^2+n-1-(n-1)(n+2)}{n(n+2)}=\frac{1}{n(n+2)}<\frac{1}{n+2}.
    \end{align*}
    Since $n-1$ agents assign a score of $0$ to $c_2$, this means that the median for this candidate is $\frac{1}{n(n+2)}$. 
    Finally, for $c_3$, we note that $f^\ladder_{n-2}(t^*)=\frac{1}{n(n+2)}+\frac{1}{n}>\frac{1}{n+2}$, so the median for this candidate is $\frac{1}{n+2}$. 
    Since the three medians sum up to $1$, \ladder{} returns $(\frac{n^2+n-1}{n(n+2)}, \frac{1}{n(n+2)}, \frac{1}{n+2})$, which violates score-unanimity.\medskip

    \noindent \textbf{Claim 2}: We only need to show that \ladder{} fails score-representation for all $m\geq 3$ and $n\geq 2$. To this end, consider again the instance $\mathcal{I}^5$ in the proof of Claim~1. In this instance, all agents assign score $\frac{n+1}{n+2}$ to $c_1$, but \ladder{} assigns a score strictly less than $\frac{n+1}{n+2}$ to this candidate, so score-representation is violated.\medskip

    \noindent \textbf{Claim 3}: To see that \ladder{} fails independence for $m\geq 3$ and $n\geq 2$, observe that \ladder{} is unanimous, i.e., if every agent reports the same score vector ${\mathbf s} = (s_1, \dots, s_m)$, the rule outputs $\mathbf s$ (see the proof of Claim~1 of \Cref{thm:consistencyprops}).
    On the other hand, we have shown in Claim~1 that \ladder{} fails score-unanimity for all $m\geq 3$ and $n\geq 2$. 
    These two observations together suffice to conclude that \ladder{} fails independence.\medskip

    \noindent \textbf{Claim 4}: We next show that \ladder{} satisfies reinforcement. 
    Assume for contradiction that this is not true.
    This means that there exist two instances $\mathI^1$ and $\mathI^2$ with disjoint electorates $N^1$ and $N^2$ such that $\ladder$ outputs a vector $\mathbf{x}$ for both $\mathI^1$ and $\mathI^2$ but a different vector $\mathbf{y}$ for the profile $\mathI^3$ which concatenates $\mathI^1$ and $\mathI^2$. 
    Let $n_1=|N^1|$, $n_2=|N^2|$, $n_3=|N^1\cup N^2|=n_1+n_2$, and for $j\in [3]$ and $i\in \{0,\dots,n_j\}$, let $p_i^j$ denote the position of the $i$-th lowest phantom at the time when $\mathbf{x}$ (for $j\in[2]$) or $\mathbf{y}$ (for $j=3$) is returned for $\mathI^j$. 
    In particular, this means that \ladder{} chooses for each instance $\mathI^j$ and each candidate $c_k$ the median of the multiset $\{s_{i,k}^j\colon i\in N^j\}\cup \{p_i^j\colon i\in \{0,\dots, n_j\}\}$. 
    For any $z\in [0,1]$ and $j\in[3]$, let $p^j(z):=|\{i\in \{0,\dots, n_j\}\colon p_i^j\leq z\}|$ be the number of phantoms that are at most~$z$ for profile $\mathI^j$, and for any $k\in[m]$, let $q^j_k(z):=|\{i\in N^j\colon s_{i,k}^j\leq z\}|$ be the number of agents who report a score of at most $z$ for candidate $c_k$. 
    Since $\mathI^3$ is a concatenation of $\mathI^1$ and $\mathI^2$, we have $q^3_k(z)=q^1_k(z)+q^2_k(z)$ for all $k\in [m]$ and $z\in [0,1]$. 

    Since $\mathbf{x}\neq\mathbf{y}$ and each of these vectors sums up to $1$, there must exist an index $k\in [m]$ such that $y_k<x_k$. 
    Because \ladder{} returns $\mathbf{x}$ for $\mathI^1$ and $\mathI^2$, this means that $p^j(x_k)+q^j_k(x_k)\geq n_j+1$ and $p^j(y_k)+q^j_k(y_k)\leq n_j$ for $j\in \{1,2\}$. 
    Moreover, it holds that $p^3(y_k)+q^3_k(y_k)\geq n_3+1=n_1+n_2+1$. 
    Since $q^3_k(y_k)=q^1_k(y_k)+q^2_k(y_k)$, we infer that
    \[p^3(y_k)\geq n_1+n_2+1-q_k^1(y_k)-q_k^2(y_k)\geq p^1(y_k)+p^2(y_k)+1.\]
    We will next show that 
    \begin{align}
    p^3(z)\geq p^1(z)+p^2(z)-1 \label{eq:phantoms}
    \end{align}
    for all $z \in [0,1]$. 
    To this end, let $z^*\in [0,1]$ be such that $p^3(z^*)=p^3(y_k)$ and $p^3(z^*-\epsilon)<p^3(y_k)$ for all $\epsilon>0$, which means that $z^*\leq y_k$. 
    Since $p^j(z^*)\leq p^j(y_k)$ for $j\in [2]$, we have $p^3(z^*)\geq p^1(z^*)+p^2(z^*)+1$. 
    Furthermore, by definition of $z^*$, there must be a phantom at $z^*$ for $\mathI^3$. 
    
    Fix any $z\in [0,1]$.
    Since we have shown that \eqref{eq:phantoms} holds for $z = z^*$, we may assume that $z \ne z^*$.
    We consider two cases.
    \medskip

    \underline{Case~1}: 
    $z<z^*$. 
    By definition of \ladder{}, we have $p^3(z)=p^3(z^*)-\lceil{(z^*-z)n_3}\rceil$. 
    On the other hand, it holds that $p^j(z)\leq p^j(z^*)-\lfloor (z^*-z)n_j\rfloor$ for $j\in [2]$, as there is at least one phantom in each interval of the form $[r, r+\frac{1}{n_j}]$ for any~$r$. 
    Combining these with the fact that $p^3(z^*)\geq p^1(z^*)+p^2(z^*)+1$, we get 
    \begin{align*}
        p^3(z)&=p^3(z^*)-\lceil{(z^*-z)n_3}\rceil\\
        &\geq  p^1(z^*)+p^2(z^*)+1-\lceil{(z^*-z)n_3}\rceil\\
        &\geq p^1(z)+p^2(z)+1+\lfloor (z^*-z)n_1\rfloor+\lfloor (z^*-z)n_2\rfloor-\lceil (z^*-z)n_3\rceil.
    \end{align*}
    
    Hence, in order to establish \eqref{eq:phantoms}, it suffices to show that $\lfloor (z^*-z)n_1\rfloor+\lfloor (z^*-z)n_2\rfloor-\lceil (z^*-z)n_3\rceil\geq -2$. 
    To this end, we first recall that $n_3=n_1+n_2$, so we need to show that $\lfloor (z^*-z)n_1\rfloor+\lfloor (z^*-z)n_2\rfloor-\lceil (z^*-z)(n_1+n_2)\rceil\geq -2$. 
    If $(z^*-z)(n_1+n_2)$ is an integer, we can simply drop the ceiling function and infer that $\lfloor (z^*-z)n_1\rfloor-(z^*-z)n_1+\lfloor (z^*-z)n_2\rfloor- (z^*-z)n_2\geq -2$. 
    Suppose now that $(z^*-z)(n_1+n_2)$ is not an integer, so we have $\lceil (z^*-z)(n_1+n_2)\rceil= \lfloor (z^*-z)(n_1+n_2)\rfloor+1$.
    Thus, it remains to prove that $\lfloor (z^*-z)n_1\rfloor+\lfloor (z^*-z)n_2\rfloor-\lfloor (z^*-z)(n_1+n_2)\rfloor\geq -1$. 
    Letting $\epsilon = (z^*-z)(n_1+n_2)-  \lfloor (z^*-z)(n_1+n_2)\rfloor > 0$, we have 
    \begin{align*}
    \lfloor (z^*-z)n_1\rfloor&+\lfloor (z^*-z)n_2\rfloor-\lfloor (z^*-z)(n_1+n_2)\rfloor\\
    &\geq (z^*-z)n_1-1 + (z^*-z)n_2 -1 - (z^*-z)(n_1+n_2)+\epsilon
    >-2.
     \end{align*}
    Since all terms in the first expression are integers, it holds that $\lfloor (z^*-z)n_1\rfloor +\lfloor (z^*-z)n_2\rfloor-\lfloor (z^*-z)(n_1+n_2)\rfloor \ge -1$, as desired.

    \medskip

    \underline{Case 2}: $z>z^*$. 
    First, if $p^3(z)=n_3+1$ (i.e., all phantoms in $\mathI^3$ are at most $z$), then $p^3(z)\geq p^1(z)+p^2(z)-1$ holds because $n_3=n_1+n_2$, $p^1(z)\leq n_1+1$, and $p^2(z)\leq n_2+1$.
    Assume therefore that $p^3(z)<n_3+1$, which means that at least one phantom is above $z$. 
    By definition of \ladder{}, we have $p^3(z)=p^3(z^*)+\lfloor(z-z^*)n_3\rfloor$. On the other hand, it holds that $p^j(z)\leq p^j(z^*)+\lceil(z-z^*)n_j\rceil$ for $j\in [2]$. 
    Combining these with the fact that $p^3(z^*)\geq p^1(z^*)+p^2(z^*)+1$, we get 
    \begin{align*}
        p^3(z)&=p^3(z^*)+\lfloor(z-z^*)n_3\rfloor\\
        &\geq p^1(z^*)+p^2(z^*)+1+\lfloor(z-z^*)n_3\rfloor\\
        &\geq p^1(z)+p^2(z) +1  - \lceil(z-z^*)n_1\rceil - \lceil(z-z^*)n_2\rceil+\lfloor(z-z^*)n_3\rfloor. 
    \end{align*}
    
    Hence, in order to establish \eqref{eq:phantoms}, it suffices to show that $- \lceil(z-z^*)n_1\rceil - \lceil(z-z^*)n_2\rceil+\lfloor(z-z^*)n_3\rfloor\geq -2$. 
    To this end, we recall that $n_3=n_1+n_2$. 
    If $(z-z^*)(n_1+n_2)$ is an integer, we can simply drop the floor function and infer that $- \lceil(z-z^*)n_1\rceil + (z-z^*)n_1 - \lceil(z-z^*)n_2\rceil+ (z-z^*)n_2 \geq -2$.
    Suppose now that $(z-z^*)(n_1+n_2)$ is not an integer, so we have $\lfloor(z-z^*)(n_1+n_2)\rfloor=\lceil(z-z^*)(n_1+n_2)\rceil-1$. 
    Thus, it remains to prove that $- \lceil(z-z^*)n_1\rceil - \lceil(z-z^*)n_2\rceil+\lceil(z-z^*)(n_1+n_2)\rceil\geq -1$. 
    Letting $\epsilon=\lceil(z-z^*)(n_1+n_2)\rceil-(z-z^*)(n_1+n_2)>0$, we have
    \begin{align*}
        - \lceil(z-z^*)n_1\rceil &- \lceil(z-z^*)n_2\rceil+\lceil(z-z^*)(n_1+n_2)\rceil\\
        &\geq -(z-z^*)n_1-1-(z-z^*)n_2-1+(z-z^*)(n_1+n_2)+\epsilon
        >-2. 
    \end{align*}
    Since all terms in the first expression are integers, it holds that $- \lceil(z-z^*)n_1\rceil - \lceil(z-z^*)n_2\rceil+\lceil(z-z^*)(n_1+n_2)\rceil\geq -1$, as desired.
    \medskip
    
    It follows that \eqref{eq:phantoms} holds for both Cases~1 and 2.

    \medskip

    We will now derive a contradiction by showing that $y_\ell\leq x_\ell$ for all $\ell\in [m]$. 
    Indeed, since $y_k<x_k$ by assumption, this implies that $\sum_{\ell\in [m]} y_\ell<\sum_{\ell\in [m]} x_\ell$, so one of our two output vectors is not a feasible score vector, thereby yielding a contradiction. 
    Fix an index $\ell\in [m]$. 
    Since \ladder{} returns the score $x_\ell$ for candidate $c_\ell$ in $\mathcal{I}^1$ and $\mathcal{I}^2$, we have $p^1(x_\ell)+q^1_\ell(x_\ell)\geq n_1+1$ and $p^2(x_\ell)+q^2_\ell(x_\ell)\geq n_2+1$. 
    As $q^3_\ell(x_\ell)=q^1_\ell(x_\ell)+q^2_\ell(x_\ell)$ and $p^3(x_\ell)\geq p^1(x_\ell)+p^2(x_\ell)-1$, we conclude that 
    \begin{align*}
        p^3(x_\ell)+q^3_\ell(x_\ell)\geq p^1(x_\ell)+q^1_\ell(x_\ell)+p^2(x_\ell)+q^2_\ell(x_\ell)-1\geq n_1+n_2+1 = n_3+1. 
    \end{align*}
    Since $y_\ell$ is the smallest value $y$ in the multiset $\{s_{i,k}^3\colon i\in N^3\}\cup \{p_i^3\colon i\in \{0,\dots, n_3\}\}$ such that $p^3(y)+q^3_\ell(y)\geq n_3+1$, it follows that $y_\ell \le x_\ell$, as desired.\medskip

    \noindent \textbf{Claim 5}: Finally, we show that \ladder{} satisfies participation. 
    Assume for contradiction that this is not true. 
    Thus, there exist two instances $\mathI$ and $\mathI'$ with corresponding outcomes $\mathbf{x}$ and $\mathbf{x'}$ such that $\mathI'$ can be obtained from $\mathI$ by adding a single agent $i$ and $d_i(\mathbf{x})<d_i(\mathbf{x}')$. 
    Let $\mathI''$ be the instance derived from $\mathI'$ by letting agent $i$ report $\mathbf{x}$. 
    By reinforcement (Claim~4), \ladder{} returns $\mathbf{x}$ for $\mathI''$. 
    However, this means that agent $i$ can manipulate in $\mathI'$ by reporting $\mathbf{x}$, thereby contradicting the strategyproofness of \ladder{}.
\end{proof}

Lastly, we proceed to the analysis of \pu{}.

\begin{theorem}
    The following claims hold.
    \begin{enumerate}[label=(\arabic*),topsep=4pt,itemsep=0pt]
        \item \pu{} satisfies range-respect (and therefore score-unanimity) if $m\leq 3$ or  $n= 2$ or $(m,n) = (4,4)$, but fails score-unanimity (and therefore range-respect) if $m \ge 4$ and $n\ge 3$ and $(m,n) \ne (4,4)$.
        It satisfies Pareto optimality when $m=2$ or $n=2$, but fails to do so for all $m\geq 3$ and $n\geq 3$.
        \item \pu{} satisfies score-representation when $m=2$, but fails to do so for all $m \geq 3$ and $n\geq 2$.
        \item \pu{} satisfies independence when $m= 2$, but fails to do so for all $m\geq 3$ and $n\geq 2$. 
        \item \pu{} satisfies reinforcement when $m=2$, but fails to do so for all $m\geq 3$.
        \item \pu{} satisfies participation.
    \end{enumerate}
\end{theorem}

\begin{proof}
    We prove each of the claims separately. 
    Further, just as for \ladder{}, we omit the proofs for $m=2$, as all corresponding claims follow from \Cref{prop:movingphantomm=2} and the analysis of \im{}.\medskip

    \noindent \textbf{Claim 1}:
    First, observe that if $n = 2$, one phantom of \pu{} reaches $1$ before the remaining two phantoms move away from $0$.
    When the first phantom reaches~$1$ (at time $t = 1/2$), the medians for the $m$ candidates are $\min(s_{1,1},s_{2,1}), \min(s_{1,2},s_{2,2}),\dots,\min(s_{1,m},s_{2,m})$, which sum to at most $\sum_{j\in [m]}s_{1,j} = 1$.
    Hence, we may assume that normalization occurs at time $t^* \ge 1/2$.
    At time $t^*$, one phantom is at~$1$ while another phantom is at~$0$, so for every $j\in [m]$, we have $\min(s_{1,j},s_{2,j}) \le x_j \le \max(s_{1,j},s_{2,j})$.
    Thus, range-respect is satisfied. By \Cref{thm:efficiency_implications}, this further means that \pu{} is Pareto optimal when $n=2$. 

    Next, consider the case $m = 3$.
    Since there is always a phantom at~$0$, it holds that $x_j \le \max_{i\in N}s_{i,j}$ for all $j\in [3]$.
    If normalization occurs at time $t^* \ge 1/2$, then one phantom is at~$1$ and $x_j \ge \min_{i\in N}s_{i,j}$ for all $j\in [3]$, so range-respect is satisfied.
    Thus, assume that normalization occurs at time $t^* < 1/2$, and suppose for contradiction that $x_j < \min_{i\in N}s_{i,j}$ for some $j\in [3]$. 
    Without loss of generality, let $j = 1$.
    Note that $\lfloor \frac{n}{2}\rfloor + 1$ phantoms remain at~$0$ for $t^* < 1/2$, and recall that for each $j\in [3]$, $x_j$ is the $(n+1)$-th smallest value among the $n+1$ phantoms and the $n$ agents' scores for $c_j$.
    Since $(n+1) - (\lfloor \frac{n}{2}\rfloor + 1) = \lceil \frac{n}{2}\rceil$, we have that $x_j$ is at most the $\lceil \frac{n}{2}\rceil$-th smallest value among $s_{1,j},\dots,s_{n,j}$.
    This implies for each $j\in \{2,3\}$ that at most $\lceil \frac{n}{2}\rceil - 1$ agents $i\in [n]$ satisfy $s_{i,j} < x_j$.
    Since $n - 2(\lceil \frac{n}{2}\rceil - 1) = 2+n - 2\lceil \frac{n}{2}\rceil \ge 1$, there exists an index $i^*\in [n]$ such that $x_2 \le s_{i^*,2}$ and $x_3 \le s_{i^*,3}$.
    Moreover, since $x_1 < \min_{i\in N}s_{i,1}$, it holds that $x_1 < s_{i^*,1}$.
    Putting these together yields $x_1+x_2+x_3 < s_{i^*,1}+s_{i^*,2}+s_{i^*,3} = 1$, a contradiction.
    Hence, \pu{} satisfies range-respect when $m = 3$.

    The proof that PU satisfies range-respect when $(m,n) = (4,4)$ is similar.
    We assume that normalization occurs at time $t^* < 1/2$ and $x_1 < \min_{i\in N}s_{i,1}$.
    For each $j\in\{2,3,4\}$, at most one index $i\in [4]$ has the property that $s_{i,j} < x_j$. In particular, when $t^*<\frac{1}{2}$, three phantoms are at $0$, so there exists an index $i^*\in[4]$ such that $x_j \le s_{i^*,j}$ for all $j\in\{2,3,4\}$.
    Since $x_1 < s_{i^*,1}$, it follows that $\sum_{j\in[4]} x_j < \sum_{j\in[4]} s_{i^*,j} = 1$, a contradiction.

    We now turn to score-unanimity and first show that \pu{} fails this property whenever $m \geq 4$, $n \geq 3$, and $n\neq 4$. 
    Consider the following instance~$\mathcal{I}^{18}$, where all candidates $c_j$ with $j\ge 5$ receive score $0$ from all agents and can be ignored.
    \begin{center}
        \begin{tabular}{ c | c c c c }
          $\mathcal{I}^{18}$ & $s_{i,1}$ & $s_{i,2}$ & $s_{i,3}$ & $s_{i,4}$  \\ 
         \hline \hline
         $i\in \{1,\dots,\lceil\frac{n}{3}\rceil\}$ & $\frac{n-1}{n}$ & $\frac{1}{2n}$ & $\frac{1}{2n}$ & $0$ \\ 
         $i\in \{\lceil\frac{n}{3}\rceil+1,\dots,\lceil\frac{2n}{3}\rceil\}$ & $\frac{n-1}{n}$ & $\frac{1}{2n}$ & $0$ & $\frac{1}{2n}$\\ 
         $i\in \{\lceil\frac{2n}{3}\rceil+1,\dots,n\}$ & $\frac{n-1}{n}$ & $0$ & $\frac{1}{2n}$ & $\frac{1}{2n}$\\ 
        \end{tabular}
\end{center}
    We claim that \pu{} returns the vector $(1-\frac{3}{2n}, \frac{1}{2n},  \frac{1}{2n},  \frac{1}{2n})$ for this instance, which violates score-unanimity for $c_1$. 
    To see this, let $t^*=\frac{1}{2}-\frac{3}{4n}$. 
    First, it holds that $f_k^\pu{}(t^*)\leq f_n^\pu{}(t^*)=1-\frac{3}{2n}$ for all $k\in \{0,\dots, n\}$, so the median of the multiset $\{f_0^\pu{}(t^*), \dots, f_n^\pu(t^*), s_{1,1},\dots, s_{n,1}\}$ is $1-\frac{3}{2n}$. 
    Next, because $t^*<\frac{1}{2}$, there are $\lceil\frac{n+1}{2}\rceil$ phantoms at~$0$. 
    Moreover, for each $j\in \{2,3,4\}$, there are at most $\lceil{\frac{n}{3}}\rceil$ agents~$i$ with $s_{i,j}=0$. 
    One can check that $\lceil\frac{n+1}{2}\rceil+\lceil{\frac{n}{3}}\rceil\leq n$ for all $n\geq 3$ with $n\neq 4$.
    Now, observe that if $k>\frac{n}{2}$, then $k\geq \frac{n}{2}+\frac{1}{2}$ and consequently $f_k^\pu{}(t^*)=\frac{4t^*k}{n}-2t^*\geq \frac{4t^*(n/2+1/2)}{n}-2t^*=\frac{2t^*}{n}$. 
    By definition of~$t^*$, we have $\frac{2t^*}{n}=\frac{1}{n}-\frac{3}{2n^2}\geq \frac{1}{n}-\frac{1}{2n} = \frac{1}{2n}$. 
    Hence, all phantoms $f_k^\pu{}(t^*)$ with $k>\frac{n}{2}$ are at or above the values submitted by the agents.
    It follows that $x_j=\frac{1}{2n}$ for all $j\in \{2,3,4\}$. 

    Next, we show that \pu{} also fails score-unanimity when $m \geq 5$ and $n = 4$. 
    Consider the following instance $\mathcal{I}^{19}$, where all candidates $c_j$ with $j\geq 6$ receive score $0$ from all agents and can be ignored.
    \begin{center}
        \begin{tabular}{ c | c c c c c }
          $\mathcal{I}^{19}$ & $s_{i,1}$ & $s_{i,2}$ & $s_{i,3}$ & $s_{i,4}$ & $s_{i,5}$  \\ 
         \hline \hline
         $1$ & $\frac{7}{10}$ & $\frac{1}{10}$ & $\frac{1}{10}$ & $\frac{1}{10}$ & $0$\\
         $2$ & $\frac{7}{10}$ & $\frac{1}{10}$ & $\frac{1}{10}$ & $0$ & $\frac{1}{10}$\\
         $3$ & $\frac{7}{10}$ & $\frac{1}{10}$ & $0$ & $\frac{1}{10}$ & $\frac{1}{10}$\\
         $4$ & $\frac{7}{10}$ & $0$ & $\frac{1}{10}$ & $\frac{1}{10}$ & $\frac{1}{10}$ 
        \end{tabular}       
    \end{center}    
    For this instance, score-unanimity requires that $x_1 = \frac{7}{10}$. 
    However, it can be verified that \pu{} returns the vector $(\frac{3}{5}, \frac{1}{10}, \frac{1}{10}, \frac{1}{10}, \frac{1}{10})$ at time $t^*=\frac{3}{10}$.

    Finally, since Pareto optimality is equivalent to range respect if $n = 2$ (\Cref{thm:efficiency_implications}), it remains to show that \pu{} fails Pareto optimality when $m\geq 3$ and $n\geq 3$. 
    Consider the instance $\mathI^1$, which was used in \Cref{thm:efficiencyprops} to show that \avg{} fails Pareto optimality. Additional candidates receive score~$0$ from every agent. 
    \begin{center}
        \begin{tabular}{ c | c c c c}
          $\mathcal{I}^1$ & $s_{i,1}$ & $s_{i,2}$ & $s_{i,3}$ \\ 
         \hline \hline
         $1$ & $0$ & $\frac{1}{2}$ & $\frac{1}{2}$\\ 
         $2$ & $\frac{1}{2}$ & $\frac{1}{2}$ & $0$\\ 
         $i\in\{3,\dots,n\}$ & $0$ & $0$ & $1$ \\
        \end{tabular}       
    \end{center}
    We distinguish between the following three cases.
    
    \medskip
    \underline{Case~1}: $n = 3$.
    In this case, \pu{} returns $\mathbf{x}=(\frac{1}{9}, \frac{4}{9}, \frac{4}{9})$ at time $t^*=\frac{2}{3}$. 
    However, $\mathbf{x}$ is not Pareto optimal, as every agent weakly prefers the vector $\mathbf{y}=(\frac{1}{18}, \frac{1}{2},\frac{4}{9})$ and agent~$1$ strictly prefers $\mathbf{y}$ to $\mathbf{x}$.

    \medskip
    \underline{Case~2}: $n = 4$.
    In this case, \pu{} returns $\mathbf{x}=(\frac{1}{6}, \frac{1}{3}, \frac{1}{2})$ at time $t^*=\frac{5}{6}$. 
    However, $\mathbf{x}$ is not Pareto optimal, as every agent weakly prefers the vector $\mathbf{y} = (0,\frac{1}{2},\frac{1}{2})$ and agent~$1$ strictly prefers $\mathbf{y}$ to $\mathbf{x}$.

    \medskip
    \underline{Case~3}: $n \ge 5$.
    We claim that \pu{} returns $\mathbf{x}=(\frac{1}{n}\cdot \frac{4}{5}, \frac{2}{n}\cdot \frac{4}{5}, 1-\frac{3}{n}\cdot\frac{4}{5})$ at time $t^*=\frac{9}{10}$. 
    First, the median for candidate $c_1$ is $f_1^\pu(t^*)=\frac{1}{n}\cdot (2\cdot\frac{9}{10}-1)=\frac{1}{n}\cdot\frac{4}{5}$, because $n-1$ agents report $0$, $f_0^\pu{}(t^*)=0$, and $f_1^\pu(t^*)<\frac{1}{2}$. 
    Next, for $c_2$, the median is $f_2^\pu{}(t^*)=\frac{2}{n}\cdot (2\cdot\frac{9}{10}-1)=\frac{2}{n}\cdot\frac{4}{5}$, because $n-2$ agents report~$0$ and $f_2^\pu{}(t^*) < \frac{1}{2}$. 
    Finally, for $c_3$, the median is $f^{\pu{}}_{n-2}(t^*)=\frac{n-2}{n}\cdot (3-2\cdot \frac{9}{10})-2+2\cdot\frac{9}{10}=1-\frac{3}{n}\cdot \frac{4}{5}$, because $n-2$ agents report~$1$, $f^{\pu{}}_{n-2}(t^*)<f^{\pu{}}_{n-1}(t^*)<f^{\pu{}}_{n}(t^*)$, and $f^{\pu{}}_{n-2}(t^*)>\frac{1}{2}$. 
    Since the medians sum up to $1$, it follows that \pu{} indeed returns $\mathbf{x}$. 
    However, $\mathbf{x}$ is not Pareto optimal, as every agent weakly prefers $\mathbf{y}=(0, \frac{3}{n}\cdot \frac{4}{5}, 1-\frac{3}{n}\cdot \frac{4}{5})$ and agent $1$ strictly prefers $\mathbf{y}$ to $\mathbf{x}$. 
    
    \medskip
    This completes the proof of Claim~1.
    
    \medskip

    \noindent \textbf{Claim 2}: 
    Next, we show that \pu{} fails score-representation for all $m \geq 3$ and $n \geq 2$. 
    Consider the following instance $\mathI^{20}$, where all candidates $c_j$ with $j\geq 4$ receive score $0$ from all agents and can be ignored. 
    \begin{center}
        \begin{tabular}{ c | c c c }
          $\mathcal{I}^{20}$ & $s_{i,1}$ & $s_{i,2}$ & $s_{i,3}$   \\ 
         \hline \hline
         $1$ & $0$ & $\frac{1}{2}$ & $\frac{1}{2}$\\ 
         $i \in \{2,\dots,n\}$ & $1$ & $0$ & $0$\\ 
        \end{tabular}
    \end{center}
    We claim that \pu{} returns $\mathbf{x} = (1-\frac{4}{3n}, \frac{2}{3n}, \frac{2}{3n})$.
    To see this, let $t^*=\frac{5}{6}$. 
    Because $n-1$ agents report~$1$ for candidate $c_1$, the median for this candidate is $f_{n-1}^\pu{}(t^*)=\frac{n-1}{n}\cdot(3-2\cdot\frac{5}{6})-2+2\cdot\frac{5}{6}=1-\frac{4}{3n}$. 
    Next, for both $c_2$ and $c_3$, because $n-1$ agents report $0$, the median is $f_1(t^*)=\frac{1}{n}(2\cdot\frac{5}{6}-1)=\frac{2}{3n}$; note that $f_1(t^*)<0.5$ as $n\geq 2$. 
    Since the medians sum up to $1$, it follows that \pu{} indeed returns $\mathbf{x}$. 
    Now, because $1-\frac{4}{3n}<\frac{n-1}{n}=1\cdot \frac{\mathcal{N}(\mathI^{20}, c_1, 1)}{n}$, score-representation is violated for $c_1$. 
    
    \medskip

    \noindent \textbf{Claim 3}:
    For our third claim, we turn to independence. 
    When $m=3$, \pu{} satisfies score-unanimity (Claim~1) and anonymity \citep{freeman2021truthfulbudget} for all $n\geq 2$. 
    Hence, it must fail independence in these cases, as the average rule is the only rule that satisfies anonymity, score-unanimity, and independence when $m\geq 3$ and $n\geq 2$ (Claim~1 of \Cref{lem:indavg}), and the average rule is different from \pu{} in these cases.\footnote{For example, the average rule fails strategyproofness for all $n\ge 2$ and $m\ge 2$ (Claim~2 of \Cref{thm:strategyproofness}), while \pu{} satisfies strategyproofness for all $n,m$ \citep{freeman2021truthfulbudget}.} 
    Moreover, we can extend this conclusion to any $m\geq 4$ by noting that \pu{} is invariant under adding candidates that receive score~$0$ from all agents. 
    \medskip

    \noindent \textbf{Claim 4}: 
    We now show that \pu{} fails reinforcement when $m\geq 3$. 
    Consider the following instances $\mathI^{21}$ and $\mathI^{22}$. 
    As usual, all candidates $c_j$ with $j\geq 4$ obtain score $0$ from all agents and can thus be ignored.  

    \begin{center}
    \begin{tabular}{ c | c c c c }
          $\mathcal{I}^{21}$ & $s_{i,1}$ & $s_{i,2}$ & $s_{i,3}$ \\ 
         \hline \hline
         $1$ & $1$ & $0$ & $0$ \\
         $2$ & $1$ & $0$ & $0$ \\
         $3$ & $0$ & $1$ & $0$\\
         $4$ & $0$ & $0$ & $1$\\
         $5$ & $0$ & $0$ & $1$\\
    \end{tabular}
    $\qquad\qquad\qquad\qquad$
    \begin{tabular}{ c | c c c c }
          $\mathcal{I}^{22}$ & $s_{i,1}$ & $s_{i,2}$ & $s_{i,3}$ \\ 
         \hline \hline
         $1$ & $1$ & $0$ & $0$ \\
         $2$ & $0$ & $1$ & $0$ \\
         $3$ & $\frac{1}{2}$ & $0$ & $\frac{1}{2}$ \\
         $4$ & $\frac{1}{2}$ & $0$ & $\frac{1}{2}$ \\
         $5$ & $0$ & $\frac{1}{2}$ & $\frac{1}{2}$ \\
    \end{tabular}
    \end{center}

    Note that \pu{} returns $\mathbf{x}=(\frac{2}{5}, \frac{1}{5}, \frac{2}{5})$ for $\mathI^{21}$ due to single-minded proportionality. 
    We claim that \pu{} also returns $\mathbf{x}$ for $\mathI^{22}$. To see this, let $t^*=\frac{3}{4}$, which implies that the phantom locations are $\{0, \frac{1}{10}, \frac{1}{5}, \frac{2}{5}, \frac{7}{10}, 1\}$.
    Hence, the medians can be computed as follows:
    \begin{align*}
        &\left\{0, 0, 0, \frac{1}{10}, \frac{1}{5}, \mathbf{\frac{2}{5}}, \frac{1}{2}, \frac{1}{2}, \frac{7}{10}, 1, 1\right\}, \\
        &\left\{0, 0, 0, 0, \frac{1}{10}, \mathbf{\frac{1}{5}}, \frac{2}{5}, \frac{1}{2}, \frac{7}{10}, 1, 1\right\}, \\
        &\left\{0, 0, 0, \frac{1}{10}, \frac{1}{5}, \mathbf{\frac{2}{5}}, \frac{1}{2}, \frac{1}{2}, \frac{1}{2}, \frac{7}{10}, 1\right\}.
    \end{align*}
    This shows that \pu{} indeed returns $\mathbf{x}$ for $\mathI^{22}$.

    Now, consider the instance $\mathI^*$ that concatenates $\mathI^{21}$ and $\mathI^{22}$. 
    Let $t^* = \frac{23}{26}$, which implies that the $k$-th phantom is at $\frac{k}{10}\cdot \frac{10}{13}=\frac{k}{13}$ if $k<5$ and $\frac{k}{10}(3-\frac{23}{13})-2+\frac{23}{13}=\frac{k}{10}\cdot \frac{16}{13} - \frac{3}{13}$ if $k\geq 5$. 
    More explicitly, this means that the phantoms are at $\{0, \frac{1}{13}, \frac{2}{13}, \frac{3}{13}, \frac{4}{13}, \frac{5}{13}, \frac{66}{130}, \frac{82}{130}, \frac{98}{130}, \frac{114}{130}, 1\}$. 
    Hence, \pu{} returns the outcome $\mathbf{y}=(\frac{5}{13}, \frac{3}{13}, \frac{5}{13})$ for $\mathI^*$, as witnessed by the following computation of the medians, where the values reported by agents are gray and the phantom values are black: 
    \begin{align*}
        &\left\{\color{gray}0,0,0,0,0, \color{black}0,\frac{1}{13}, \frac{2}{13}, \frac{3}{13}, \frac{4}{13}, \mathbf{\frac{5}{13}}, \color{gray}\frac{1}{2}, \frac{1}{2}, \color{black}\frac{66}{130}, \frac{82}{130}, \frac{98}{130}, \frac{114}{130}, 1, \color{gray}1,1, 1\right\},\\
        &\left\{\color{gray}0,0,0,0,0, 0, 0, \color{black}0,\frac{1}{13}, \frac{2}{13}, \mathbf{\frac{3}{13}}, \frac{4}{13}, {\frac{5}{13}}, \color{gray}\frac{1}{2}, \color{black}\frac{66}{130}, \frac{82}{130}, \frac{98}{130}, \frac{114}{130}, 1, \color{gray}1,1\right\},\\
        &\left\{\color{gray}0,0,0,0,0, \color{black}0,\frac{1}{13}, \frac{2}{13}, \frac{3}{13}, \frac{4}{13}, \mathbf{\frac{5}{13}}, \color{gray}\frac{1}{2}, \frac{1}{2}, \frac{1}{2}, \color{black}\frac{66}{130}, \frac{82}{130}, \frac{98}{130}, \frac{114}{130}, 1, \color{gray}1,1\right\}.
    \end{align*}
    Since $\mathbf{y} \ne\mathbf{x}$, it follows that \pu{} fails reinforcement.
    \medskip

    \noindent \textbf{Claim 5}: Finally, to show that \pu{} satisfies participation, we will make use of the following general statement.
    For a moving phantoms rule~$F$, we denote by $f_{k,n}^F{}(t)$ the position of the $k$-th phantom at time~$t$ when there are $n$ agents.

    \begin{lemma}
    \label{lem:phantom-participation}
    Suppose that $F$ is a moving phantoms rule such that $f_{k-1,n}^F(t) \le f_{k-1,n-1}^F(t) \le f_{k,n}^F(t)$ for all $n\geq 2$, $k\in[n]$, and $t\in[0,1]$.
    Then, $F$ satisfies participation.
    \end{lemma}
    
    \begin{innerproof}
    Consider any instance $\mathI$ with $n-1$ agents for some $n\ge 2$, and let $F(\mathI) = \mathbf{x}$.
    We will show that if we add an agent~$i'$ who reports~$\mathbf{x}$ to obtain an instance $\mathI'$, then $F(\mathI') = \mathbf{x}$ as well.
    This suffices to establish the participation of $F$.
    Indeed, if agent $i'$ instead reports some other vector $\mathbf{y}$ and the outcome changes in such a way that $d_{i'}(\mathbf{x}) < d_{i'}(F(\mathI'))$, then agent~$i'$ can manipulate by reporting $\mathbf{x}$ in $\mathI'$, which contradicts the fact that $F$ is strategyproof \citep{freeman2021truthfulbudget}.
    
    Next, let $t$ be the time when $\mathbf{x}$ is returned for $\mathI$.
    We claim that at the same time $t$, the outcome $\mathbf{x}$ is also returned for $\mathI'$.
    Consider any $r\in [m]$.
    In the instance $\mathI$, at time $t$, the number of phantoms that are at most $x_r$ and the number of agents who report a score of at most $x_r$ for candidate $c_r$ sum up to at least~$n$.
    Now, at the same time $t$ in $\mathI'$, the number of phantoms that are at most $x_r$ does not decrease because $f_{k-1,n}^F(t) \le f_{k-1,n-1}^F(t)$ for all $k\in[n]$, and the number of agents who report a score of at most~$x_r$ increases by~$1$ due to agent~$i'$.
    This means that the two quantities sum up to at least $n+1$ for $\mathI'$.
    By similar arguments using the fact that $f_{k-1,n-1}^F(t) \le f_{k,n}^F(t)$ for all $k\in[n]$, the number of phantoms that are at least $x_r$ and the number of agents who report a score of at least $x_r$ sum up to at least $n+1$ for $\mathI'$.
    Hence, the median for candidate $c_r$ at time~$t$ remains $x_r$, and therefore $F(\mathI') = \mathbf{x}$.
    \end{innerproof}
    
    It remains to show that \pu{} satisfies the condition of \Cref{lem:phantom-participation}.
    We divide our analysis into two cases depending on the value of $t$.
    \medskip
    
    \underline{Case~1}: $t < \frac{1}{2}$. First, if $\frac{k-1}{n-1} < \frac{1}{2}$, then $f_{k-1,n-1}^\pu{}(t) = 0 \le f_{k,n}^\pu{}(t)$ by the definition of \pu{}. 
    Since $\frac{k-1}{n} \le \frac{k-1}{n-1} < \frac{1}{2}$, it further follows that $f_{k-1,n}^\pu{}(t) = 0 = f_{k-1,n-1}^\pu{}(t)$, so our desired inequality holds.

    Next, assume that $\frac{k-1}{n-1} \ge \frac{1}{2}$, so $f_{k-1,n-1}^\pu{}(t) = \frac{4t(k-1)}{n-1} - 2t$.
    Since $\frac{k}{n} \ge \frac{k-1}{n-1} \ge \frac{1}{2}$, we have $f_{k,n}^\pu{}(t) = \frac{4tk}{n} - 2t \ge \frac{4t(k-1)}{n-1} - 2t = f_{k-1,n-1}^\pu{}(t)$.
    Next, if $\frac{k-1}{n} \ge \frac{1}{2}$, it holds that $f_{k-1,n}^\pu{}(t) = \frac{4t(k-1)}{n} - 2t \le \frac{4t(k-1)}{n-1} - 2t = f_{k-1,n-1}^\pu{}(t)$.
    On the other hand, if $\frac{k-1}{n} < \frac{1}{2}$, then $f_{k-1,n}^\pu{}(t) = 0 \le f_{k-1,n-1}^\pu{}(t)$.
    \medskip
    
    \underline{Case~2}:
    $t \ge \frac{1}{2}$.
    Assume first that $\frac{k-1}{n-1} < \frac{1}{2}$, which means that $f_{k-1,n-1}^\pu{}(t) = \frac{(k-1)(2t-1)}{n-1}$.
    Since $\frac{k-1}{n} \le \frac{k-1}{n-1} < \frac{1}{2}$, we have $f_{k-1,n}^\pu{}(t) = \frac{(k-1)(2t-1)}{n} \le \frac{(k-1)(2t-1)}{n-1} = f_{k-1,n-1}^\pu{}(t)$.
    Next, if $\frac{k}{n} < \frac{1}{2}$, it holds that $f_{k,n}^\pu{}(t) = \frac{k(2t-1)}{n} \ge \frac{(k-1)(2t-1)}{n-1} = f_{k-1,n-1}^\pu{}(t)$.
    On the other hand, if $\frac{k}{n}\ge \frac{1}{2}$, then since $\frac{k-1}{n-1} < \frac{1}{2}$, it must be that $\frac{k}{n} = \frac{1}{2}$.
    Hence, we have $f_{k,n}^\pu{}(t) = t - \frac{1}{2} \ge \frac{(k-1)(2t-1)}{n-1} = f_{k-1,n-1}^\pu{}(t)$.

    Assume now that $\frac{k-1}{n-1} \ge \frac{1}{2}$, so $f_{k-1,n-1}^\pu{}(t) = \frac{(k-1)(3-2t)}{n-1} - 2 + 2t$.
    Since $\frac{k}{n} \ge \frac{k-1}{n-1} \ge \frac{1}{2}$, we have $f_{k,n}^\pu{}(t) = \frac{k(3-2t)}{n} - 2 + 2t \ge \frac{(k-1)(3-2t)}{n-1} - 2 + 2t = f_{k-1,n-1}^\pu{}(t)$.
    Next, if $\frac{k-1}{n} \ge \frac{1}{2}$, it holds that $f_{k-1,n}^\pu{}(t) = \frac{(k-1)(3-2t)}{n} - 2 + 2t \le \frac{(k-1)(3-2t)}{n-1} - 2 + 2t = f_{k-1,n-1}^\pu{}(t)$.
    On the other hand, if $\frac{k-1}{n} < \frac{1}{2}$, then since $\frac{k-1}{n-1} \ge \frac{1}{2}$, it must be that $\frac{k-1}{n-1} = \frac{1}{2}$.
    Hence, we have $f_{k-1,n-1}^\pu{}(t) = t - \frac{1}{2} \ge \frac{(k-1)(2t-1)}{n} = f_{k-1,n}^\pu{}(t)$.
    \medskip

    In both cases, it holds that $f_{k-1,n}^\pu{}(t) \le f_{k-1,n-1}^\pu{}(t) \le f_{k,n}^\pu{}(t)$, as desired.
\end{proof}

\end{document}